\documentclass[12pt]{article}
\usepackage{amsmath}
\usepackage{prodint}
\usepackage{graphicx}
\usepackage{enumerate}
\usepackage{natbib}
\usepackage{url} 
\usepackage{binomexp}
\usepackage{mathtools}
\usepackage{amsfonts}
\usepackage{mathrsfs}
\RequirePackage{subcaption}
\usepackage{booktabs}
\usepackage{bbm}
\usepackage{dsfont}
\usepackage{longtable}
\usepackage{tabulary}
\usepackage{yhmath}
\usepackage{bm}
\usepackage{amssymb}
\usepackage[nodisplayskipstretch]{setspace}
\usepackage{tikz}
\usetikzlibrary{arrows.meta, decorations.pathmorphing, decorations.pathreplacing, positioning, calc, shadows.blur}
\usepackage{float}
\usepackage{color}
\usepackage{stmaryrd}
\usepackage{paralist}
\usepackage{multirow}
\usepackage{threeparttable}
\usepackage{soul}
\usepackage{algorithm,algpseudocode}
\makeatletter
\newcommand{\algmargin}{\the\ALG@thistlm}
\makeatother
\newlength{\whilewidth}
\algdef{SE}[parWHILE]{parWhile}{EndparWhile}[1]
  {\parbox[t]{\dimexpr\linewidth-\algmargin}{%
     \hangindent\whilewidth\strut\algorithmicwhile\ #1\ \algorithmicdo\strut}}{\algorithmicend\ \algorithmicwhile}%
\algnewcommand{\parState}[1]{\State%
  \parbox[t]{\dimexpr\linewidth-\algmargin}{\strut #1\strut}}

\newcommand{\Var}{{\mathbb{V}\mathrm{ar}}}
\newcommand{\Cov}{{\mathbb{C}\mathrm{ov}}}

\newcommand{\bbE}{\mathbb{E}}

\newcommand{\bbP}{\mathbb{P}}

\newcommand{\bbR}{\mathbb{R}}

\newcommand{\bbone}{\mathds{1}}

\newcommand{\qed}{\hfill \mbox{\raggedright \rule{.07in}{.1in}}}
\usepackage[colorlinks=true,linkcolor=blue,citecolor=black]{hyperref}

\usepackage{amsthm}
\theoremstyle{plain}
\newtheorem{theorem}{Theorem}
\newtheorem{lemma}{Lemma}
\newtheorem{proposition}{Proposition}
\newtheorem{corollary}{Corollary}
\theoremstyle{definition}
\newtheorem{definition}{Definition}
\newtheorem{assumption}{Assumption}

\newcommand{\blind}{1}

\allowdisplaybreaks

\usetikzlibrary{shapes,decorations,arrows,calc,arrows.meta,fit,positioning}
\tikzset{
     -Latex,auto,node distance =1 cm and 1 cm,semithick,
     state/.style ={ellipse, draw, minimum width = 0.7 cm},
     point/.style = {circle, draw, inner sep=0.04cm,fill,node contents={}},
    bidirected/.style={Latex-Latex,dashed},
     el/.style = {inner sep=2pt, align=left, sloped}
 }

\definecolor{vizblue}{HTML}{2A78D6}
\definecolor{vizorange}{HTML}{EB6834}
\definecolor{vizink}{HTML}{0B0B0B}
\definecolor{vizsec}{HTML}{52514E}
\definecolor{vizmuted}{HTML}{898781}
\definecolor{vizgrid}{HTML}{E1E0D9}
\definecolor{vizaxis}{HTML}{C3C2B7}
\definecolor{vizcell}{HTML}{E4E3DD}
\definecolor{vizband}{HTML}{EAF2FC}

\renewcommand{\Var}{\operatorname{Var}}
\renewcommand{\Cov}{\operatorname{Cov}}
\newcommand{\cA}{\mathcal{A}}
\newcommand{\cB}{\mathcal{B}}
\newcommand{\cC}{\mathcal{C}}
\newcommand{\cE}{\mathcal{E}}
\newcommand{\cN}{\mathcal{N}}
\newcommand{\cS}{\mathcal{S}}
\newcommand{\cT}{\mathcal{T}}
\newcommand{\cV}{\mathcal{V}}
\newcommand{\cW}{\mathcal{W}}
\newcommand{\lpi}{\underline{\pi}}
\newcommand{\lv}{\underline{v}}
\newcommand{\htau}{\widehat{\tau}}
\newcommand{\htheta}{\widehat{\theta}}
\newcommand{\hvartheta}{\widehat{\vartheta}}
\newcommand{\hmu}{\widehat{\mu}}
\newcommand{\hV}{\widehat{V}}
\theoremstyle{plain}
\newtheorem{propositionS}{Proposition}

\usepackage{tocloft}
\begin{document}

\def\spacingset#1{\renewcommand{\baselinestretch}%
{#1}\small\normalsize} \spacingset{1}

\addtocontents{toc}{\protect\setcounter{tocdepth}{-1}}


\if1\blind
{
  \title{\Large \bf Design-based effect sizes and their synthesis in randomized multiple-baseline experiments}
  \author{Xinyuan Chen\vspace{0.2cm}\\
    Department of Mathematics and Statistics, \\
    Mississippi State University, Mississippi State, MS, USA \\
    xchen@math.msstate.edu
    \and
    Zhigao Liang\vspace{0.2cm}\\
    Department of Teacher Education and Leadership, \\
    Mississippi State University, Mississippi State, MS, USA}
  \maketitle
} \fi

\if0\blind
{
  \bigskip
  \bigskip
  \bigskip
  \begin{center}
    {\Large \bf Design-Based Effect Sizes and Their Synthesis in Randomized Multiple-Baseline Experiments}
\end{center}
  \medskip
} \fi

\begin{abstract}
    Effect sizes for single-case experimental designs usually rely on models for the baseline trend and the serial dependence, which are challenging to check with few sessions. We develop design-based effect sizes for multiple-baseline experiments with a randomized intervention onset. We define case-level and pooled standardized effects with potential outcomes and propose weighted estimators that are unbiased or asymptotically unbiased without an outcome model. The usual difference between phase means is generally biased when the untreated outcome changes over sessions. For research synthesis, we provide variance estimators and asymptotic confidence intervals that require no model for serial dependence, and variance formulas for planning. Simulations and three randomized studies of reading and writing interventions illustrate the methods.
\end{abstract}

\noindent%

\spacingset{1.75} 

\section{Introduction}\label{sec:intro}

In single-case experimental designs, we evaluate the effect of an intervention by measuring a few cases repeatedly. Each case begins in a baseline phase (untreated), the intervention starts at a certain session (the onset), and it continues until the study ends. The multiple-baseline design, in which the intervention starts at different sessions for different cases, is common among them \citep{ShadishSullivan2011}. Its staggered onsets show the effect within and across series \citep{Horner2005}, and because the intervention is never withdrawn, it suits skills that do not return to baseline once learned, such as reading and writing, and classrooms and clinics where withdrawing an effective intervention is undesirable. In education and special education, single-case studies are an accepted source of evidence on the effect of an intervention, particularly for low-incidence problems \citep{Kratochwill2013, Cook2015}. \citet{Horner2005} proposed that a practice be considered evidence-based when at least five acceptable single-case studies, by at least three researchers in three locations and with at least 20 participants in total, document its effect, and the standards of the What Works Clearinghouse \citep{Kratochwill2013} and the Council for Exceptional Children \citep{Cook2015} likewise draw on single-case studies to identify evidence-based practices. Because a single study includes only a few cases, researchers often use a research synthesis, or meta-analysis, to combine results from several studies into a single summary of the effect.

A synthesis requires an effect size for each case or study, where the measures in use compare the intervention phase with the baseline phase. Examples are the nonoverlap of all pairs \citep{ParkerVannest2009}, the log response ratio \citep{Pustejovsky2018}, the standardized mean differences of \citet{Hedges2012, Hedges2013} and \citet{Pustejovsky2014}, and the coefficients of multilevel models \citep{VandenNoortgate2008, ChenPustejovsky2024}. These methods, however, are largely model-based, i.e., their validity depends on whether the model correctly specifies the series. Specifically, the baseline needs to be stable, or the model needs to specify its trend correctly, as well as its serial dependence, to obtain the standard errors of the effect. In practice, a baseline phase of five to ten sessions in these studies usually gives little information for checking such models. \citet{Joo2022} compared eight methods for estimating effects from multiple-baseline studies by simulation, and the bias was substantial when their assumptions were violated.

Randomization offers a basis for inference that requires no model for the series. It is recommended in the methodology literature to draw the session at which the intervention starts at random \citep{EdgingtonOnghena2007, KratochwillLevin2010, Tanious2024}, and several randomization schemes exist for multiple-baseline designs \citep{MarascuiloBusk1988, WampoldWorsham1986, KoehlerLevin1998, LevinFerron2021}. Randomization is not uncommon in applied single-case studies \citep{TaniousOnghena2021}, and its use has increased steadily, although randomizing the onset remains the least frequent type of randomization \citep{Kittelman2025}. So far, work on randomization has focused on testing, as randomization tests address the sharp null hypothesis that the intervention has no effect in any case at any session. Inverting them gives confidence intervals under a constant additive effect \citep{Michiels2017}, and \citet{Chen2026} tests bounded and quantile null hypotheses on session-level effects. In comparison, estimation has received less attention. \citet{Ferron2014} defined the effect in a multiple-baseline (MB) study using potential outcomes and, by simulation, studied an estimator that contrasts cases in the intervention phase with cases still in baseline at the same time. They observed that ``although the role of randomization in statistical testing is well established in the MB literature, the role of randomization in the estimation of causal effects from MB studies is not'' (p. 495).

This paper derives effect sizes and their synthesis from a design-based perspective, where the intervention onset is randomized. The potential outcomes of every case are viewed as fixed, and the onset is the sole source of randomness, as in design-based analyses of staggered adoption \citep{AtheyImbens2022, RothSantAnna2023} and of time series and panel experiments \citep{BojinovShephard2019, Bojinov2021}. We obtain four main results. First, weighted comparisons of treated and untreated sessions are unbiased for a case-level effect for every number of cases and sessions, and they remain unbiased after division by a scale computed from the sessions before the first admissible onset. Second, the expectation of the usual phase contrast is a weighted average of session-level effects plus a drift term, and randomization does not remove the drift. Third, a pooled effect across cases and studies has a variance estimator that requires no model for serial dependence, and its confidence interval attains at least the nominal coverage as the number of independently randomized cases or studies grows. Fourth, a planning formula gives the precision of the pooled effect from the number of admissible onsets and the numbers of sessions before and after them. Proofs and additional results are given in the Supplemental Material.

\section{Setup}

\subsection{Cases, Sessions and Onsets}

Suppose there are $N$ cases, which may come from several studies. Case $i=1,\ldots,N$ is observed at sessions $t=1,\ldots,T_i$. The intervention for case $i$ starts at session $S_i$ (its onset) and remains in place in every later session. The onset is drawn from a set $\cS_i$ of at least two admissible sessions, with $p_i(s)=\bbP(S_i=s)>0$ for $s\in\cS_i$. The cases are partitioned into blocks $\cC_1,\ldots,\cC_G$, and the vectors of onsets of different blocks are mutually independent. Notably, each case is a block and $G=N$ when every onset is drawn separately \citep{MarascuiloBusk1988}, and the study becomes a block when a study assigns a fixed list of onsets to its cases by a random permutation \citep{WampoldWorsham1986}. We operate in a design-based framework, where the cases and their potential outcomes are viewed as fixed, and the sole source of randomness is the draw of the onsets. All probabilities, expectations, and variances hereafter refer to this draw unless specified otherwise.

Let $a_i=\min\cS_i-1$ and $b_i=\max\cS_i$, with $a_i\ge1$. The sessions of case $i$ form three sets, i.e.,
\begin{equation*}
    \cA_i=\{1,\ldots,a_i\},\qquad\cW_i=\{a_i+1,\ldots,b_i-1\},\qquad\cB_i=\{b_i,\ldots,T_i\}.
\end{equation*}
The pre-window sessions $\cA_i$ come before every admissible onset and thus are untreated; whereas the post-window sessions $\cB_i$ are treated under every admissible onset. At a window session in $\cW_i$, the case is treated under some onsets and untreated under others. There are $m_i=b_i-a_i-1\ge1$ window sessions and $n_i=T_i-b_i+1\ge1$ post-window sessions. Let $W_{it}(s)=\bbone(t\ge s)$ denote that session $t$ is treated under onset $s$, and let
\begin{equation*}
    W_{it}=W_{it}(S_i),\qquad\pi_{it}=\bbP(S_i\le t).
\end{equation*}
The probability $\pi_{it}$ is strictly between 0 and 1 for window sessions, equals 0 for pre-window sessions, and equals 1 for post-window sessions. Figure \ref{fig:window} illustrates the three sets and the probabilities $\pi_{it}$ for a case with eight sessions and four equally likely onsets.

\begin{figure}[t]
    \caption{Admissible Onsets, Treatment Probabilities and Outcomes of a Case With Eight Sessions and Onsets Drawn From Sessions 3 to 6}
    \label{fig:window}
    \centering
    \begin{tikzpicture}[-,thin,x=1cm,y=1cm,font=\linespread{1}\selectfont\sffamily\footnotesize,line cap=round,line join=round]
        \fill[vizband,rounded corners=3pt] (2.4,-7.25) rectangle (6,0.16);
        \draw[vizsec,line width=0.5pt] (0.08,0.32) -- (0.08,0.42) -- (2.32,0.42) -- (2.32,0.32);
        \node[above,text=vizink,inner sep=2pt] at (1.2,0.42) {Pre-window $\mathcal{A}_i$};
        \draw[vizsec,line width=0.5pt] (2.48,0.32) -- (2.48,0.42) -- (5.92,0.42) -- (5.92,0.32);
        \node[above,text=vizink,inner sep=2pt] at (4.2,0.42) {Window $\mathcal{W}_i$};
        \draw[vizsec,line width=0.5pt] (6.08,0.32) -- (6.08,0.42) -- (9.52,0.42) -- (9.52,0.32);
        \node[above,text=vizink,inner sep=2pt] at (7.8,0.42) {Post-window $\mathcal{B}_i$};
        \fill[vizcell,rounded corners=1.6pt] (0.035,-0.44) rectangle (1.165,0);
        \fill[vizcell,rounded corners=1.6pt] (1.235,-0.44) rectangle (2.365,0);
        \fill[vizblue,rounded corners=1.6pt] (2.435,-0.44) rectangle (3.565,0);
        \fill[vizblue,rounded corners=1.6pt] (3.635,-0.44) rectangle (4.765,0);
        \fill[vizblue,rounded corners=1.6pt] (4.835,-0.44) rectangle (5.965,0);
        \fill[vizblue,rounded corners=1.6pt] (6.035,-0.44) rectangle (7.165,0);
        \fill[vizblue,rounded corners=1.6pt] (7.235,-0.44) rectangle (8.365,0);
        \fill[vizblue,rounded corners=1.6pt] (8.435,-0.44) rectangle (9.565,0);
        \node[left,text=vizsec,inner sep=2pt] at (-0.42,-0.22) {Onset 3};
        \fill[vizcell,rounded corners=1.6pt] (0.035,-0.97) rectangle (1.165,-0.53);
        \fill[vizcell,rounded corners=1.6pt] (1.235,-0.97) rectangle (2.365,-0.53);
        \fill[vizcell,rounded corners=1.6pt] (2.435,-0.97) rectangle (3.565,-0.53);
        \fill[vizblue,rounded corners=1.6pt] (3.635,-0.97) rectangle (4.765,-0.53);
        \fill[vizblue,rounded corners=1.6pt] (4.835,-0.97) rectangle (5.965,-0.53);
        \fill[vizblue,rounded corners=1.6pt] (6.035,-0.97) rectangle (7.165,-0.53);
        \fill[vizblue,rounded corners=1.6pt] (7.235,-0.97) rectangle (8.365,-0.53);
        \fill[vizblue,rounded corners=1.6pt] (8.435,-0.97) rectangle (9.565,-0.53);
        \node[left,text=vizink,inner sep=2pt] at (-0.42,-0.75) {\textbf{Onset 4, drawn}};
        \fill[vizorange] (-0.34,-0.64) -- (-0.12,-0.75) -- (-0.34,-0.86) -- cycle;
        \fill[vizcell,rounded corners=1.6pt] (0.035,-1.5) rectangle (1.165,-1.06);
        \fill[vizcell,rounded corners=1.6pt] (1.235,-1.5) rectangle (2.365,-1.06);
        \fill[vizcell,rounded corners=1.6pt] (2.435,-1.5) rectangle (3.565,-1.06);
        \fill[vizcell,rounded corners=1.6pt] (3.635,-1.5) rectangle (4.765,-1.06);
        \fill[vizblue,rounded corners=1.6pt] (4.835,-1.5) rectangle (5.965,-1.06);
        \fill[vizblue,rounded corners=1.6pt] (6.035,-1.5) rectangle (7.165,-1.06);
        \fill[vizblue,rounded corners=1.6pt] (7.235,-1.5) rectangle (8.365,-1.06);
        \fill[vizblue,rounded corners=1.6pt] (8.435,-1.5) rectangle (9.565,-1.06);
        \node[left,text=vizsec,inner sep=2pt] at (-0.42,-1.28) {Onset 5};
        \fill[vizcell,rounded corners=1.6pt] (0.035,-2.03) rectangle (1.165,-1.59);
        \fill[vizcell,rounded corners=1.6pt] (1.235,-2.03) rectangle (2.365,-1.59);
        \fill[vizcell,rounded corners=1.6pt] (2.435,-2.03) rectangle (3.565,-1.59);
        \fill[vizcell,rounded corners=1.6pt] (3.635,-2.03) rectangle (4.765,-1.59);
        \fill[vizcell,rounded corners=1.6pt] (4.835,-2.03) rectangle (5.965,-1.59);
        \fill[vizblue,rounded corners=1.6pt] (6.035,-2.03) rectangle (7.165,-1.59);
        \fill[vizblue,rounded corners=1.6pt] (7.235,-2.03) rectangle (8.365,-1.59);
        \fill[vizblue,rounded corners=1.6pt] (8.435,-2.03) rectangle (9.565,-1.59);
        \node[left,text=vizsec,inner sep=2pt] at (-0.42,-1.81) {Onset 6};
        \fill[vizcell,rounded corners=1.6pt] (10.05,-0.37) rectangle (10.47,-0.07);
        \node[right,text=vizsec,inner sep=2pt] at (10.55,-0.22) {Untreated};
        \fill[vizblue,rounded corners=1.6pt] (10.05,-0.87) rectangle (10.47,-0.57);
        \node[right,text=vizsec,inner sep=2pt] at (10.55,-0.72) {Treated};
        \node[left,text=vizink,align=right,inner sep=2pt] at (-0.42,-3.28) {Probability\\treated, $\pi_{it}$};
        \draw[vizgrid,line width=0.4pt] (0,-3.78) -- (9.6,-3.78);
        \node[left,text=vizmuted,inner sep=1.5pt,font=\sffamily\scriptsize] at (-0.02,-3.78) {0};
        \draw[vizgrid,line width=0.4pt] (0,-2.78) -- (9.6,-2.78);
        \node[left,text=vizmuted,inner sep=1.5pt,font=\sffamily\scriptsize] at (-0.02,-2.78) {1};
        \fill[vizblue] (2.85,-3.78) [rounded corners=1.6pt] -- (2.85,-3.53) -- (3.15,-3.53) [sharp corners] -- (3.15,-3.78) -- cycle;
        \node[above,text=vizink,inner sep=1.5pt] at (3,-3.53) {0.25};
        \fill[vizblue] (4.05,-3.78) [rounded corners=1.6pt] -- (4.05,-3.28) -- (4.35,-3.28) [sharp corners] -- (4.35,-3.78) -- cycle;
        \node[above,text=vizink,inner sep=1.5pt] at (4.2,-3.28) {0.5};
        \fill[vizblue] (5.25,-3.78) [rounded corners=1.6pt] -- (5.25,-3.03) -- (5.55,-3.03) [sharp corners] -- (5.55,-3.78) -- cycle;
        \node[above,text=vizink,inner sep=1.5pt] at (5.4,-3.03) {0.75};
        \fill[vizblue] (6.45,-3.78) [rounded corners=1.6pt] -- (6.45,-2.78) -- (6.75,-2.78) [sharp corners] -- (6.75,-3.78) -- cycle;
        \fill[vizblue] (7.65,-3.78) [rounded corners=1.6pt] -- (7.65,-2.78) -- (7.95,-2.78) [sharp corners] -- (7.95,-3.78) -- cycle;
        \fill[vizblue] (8.85,-3.78) [rounded corners=1.6pt] -- (8.85,-2.78) -- (9.15,-2.78) [sharp corners] -- (9.15,-3.78) -- cycle;
        \draw[vizaxis,line width=0.6pt] (0,-3.78) -- (9.6,-3.78);
        \node[left,text=vizink,align=right,inner sep=2pt] at (-0.42,-5.98) {Outcome, $Y_{it}$};
        \draw[vizaxis,line width=0.4pt] (0,-7.13) -- (9.6,-7.13);
        \node[left,text=vizmuted,inner sep=1.5pt,font=\sffamily\scriptsize] at (-0.02,-7.13) {0};
        \draw[vizgrid,line width=0.4pt] (0,-5.98) -- (9.6,-5.98);
        \node[left,text=vizmuted,inner sep=1.5pt,font=\sffamily\scriptsize] at (-0.02,-5.98) {4};
        \draw[vizgrid,line width=0.4pt] (0,-4.83) -- (9.6,-4.83);
        \node[left,text=vizmuted,inner sep=1.5pt,font=\sffamily\scriptsize] at (-0.02,-4.83) {8};
        \draw[vizorange,line width=1.2pt] (3.6,-7.13) -- (3.6,-4.81);
        \node[above,text=vizink,inner sep=1.5pt] at (3.6,-4.81) {Onset $S_i=4$};
        \draw[vizink,line width=0.9pt] (0.18,-6.411) -- (2.22,-6.411);
        \node[above right,text=vizink,inner sep=1.5pt] at (0.18,-5.908) {Pre-window mean 2.5};
        \draw[vizink,line width=0.9pt] (6.18,-5.117) -- (9.42,-5.117);
        \node[below right,text=vizink,inner sep=1.5pt] at (6.18,-5.506) {Post-window mean 7};
        \draw[vizmuted,line width=0.8pt] (0.6,-6.555) -- (1.8,-6.268) -- (3,-6.555) -- (4.2,-5.405) -- (5.4,-5.117) -- (6.6,-5.405) -- (7.8,-4.83) -- (9,-5.117);
        \filldraw[fill=white,draw=vizsec,line width=0.9pt] (0.6,-6.555) circle (2.4pt);
        \filldraw[fill=white,draw=vizsec,line width=0.9pt] (1.8,-6.268) circle (2.4pt);
        \filldraw[fill=white,draw=vizsec,line width=0.9pt] (3,-6.555) circle (2.4pt);
        \filldraw[fill=vizblue,draw=white,line width=0.9pt] (4.2,-5.405) circle (2.9pt);
        \filldraw[fill=vizblue,draw=white,line width=0.9pt] (5.4,-5.117) circle (2.9pt);
        \filldraw[fill=vizblue,draw=white,line width=0.9pt] (6.6,-5.405) circle (2.9pt);
        \filldraw[fill=vizblue,draw=white,line width=0.9pt] (7.8,-4.83) circle (2.9pt);
        \filldraw[fill=vizblue,draw=white,line width=0.9pt] (9,-5.117) circle (2.9pt);
        \node[below,text=vizsec,inner sep=3pt] at (0.6,-7.23) {1};
        \node[below,text=vizsec,inner sep=3pt] at (1.8,-7.23) {2};
        \node[below,text=vizsec,inner sep=3pt] at (3,-7.23) {3};
        \node[below,text=vizsec,inner sep=3pt] at (4.2,-7.23) {4};
        \node[below,text=vizsec,inner sep=3pt] at (5.4,-7.23) {5};
        \node[below,text=vizsec,inner sep=3pt] at (6.6,-7.23) {6};
        \node[below,text=vizsec,inner sep=3pt] at (7.8,-7.23) {7};
        \node[below,text=vizsec,inner sep=3pt] at (9,-7.23) {8};
        \node[below,text=vizink,inner sep=2pt] at (4.8,-7.68) {Session};
        \filldraw[fill=white,draw=vizsec,line width=0.9pt] (10.26,-5.58) circle (2.4pt);
        \node[right,text=vizsec,inner sep=2pt] at (10.55,-5.58) {Untreated};
        \filldraw[fill=vizblue,draw=white,line width=0.9pt] (10.26,-6.08) circle (2.9pt);
        \node[right,text=vizsec,inner sep=2pt] at (10.55,-6.08) {Treated};
    \end{tikzpicture}
    \par\smallskip
    \parbox{0.95\textwidth}{\small\textit{Note.} Top: each row is one admissible onset; the four onsets are equally likely, and session 4 was drawn. Middle: the probability that a session is treated. Bottom: the outcomes observed under the onset drawn, with the pre-window mean and the post-window mean.}
\end{figure}
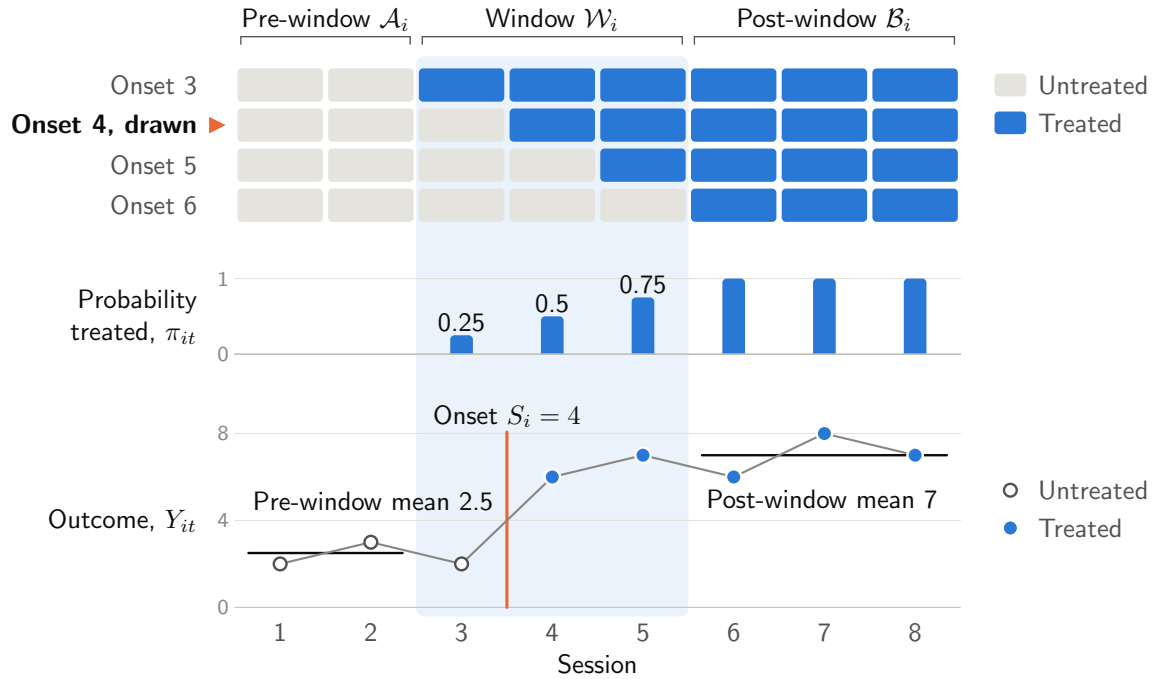

\subsection{Potential Outcomes and Estimands}

Let $Y_{it}(s)$ be the outcome of case $i$ at session $t$ when the onset is $s$. Here, $s$ is an admissible onset or $s=\infty$, which denotes an intervention that does not start during the $T_i$ sessions. Under the Stable Unit Treatment Value Assumption (SUTVA), the outcome of a case does not depend on the onsets of other cases. We thus observe $Y_{it}=Y_{it}(S_i)$.

\begin{assumption}[No anticipation]
    \label{ass:anticipation}
    $Y_{it}(s)=Y_{it}(\infty)$ for every case $i$, every session $t$ and every $s\in\cS_i$ with $s>t$.
\end{assumption}

We write $Y^{(0)}_{it}=Y_{it}(\infty)$ for the baseline outcome. Under Assumption \ref{ass:anticipation}, the outcomes at pre-window sessions are the same under every onset. For a window session $t$, define
\begin{equation}
    \overline Y^{(1)}_{it}=\frac1{\pi_{it}}\sum_{s\in\cS_i,\,s\le t}p_i(s)Y_{it}(s),\qquad\tau_{it}=\overline Y^{(1)}_{it}-Y^{(0)}_{it},\qquad\tau_i=\frac1{m_i}\sum_{t\in\cW_i}\tau_{it}.
    \label{eq:celleffect}
\end{equation}
The treated outcome at session $t$ may depend on the duration of the intervention. $\overline Y^{(1)}_{it}$ averages it over the onsets at or before $t$, with the probabilities that the design dictates. The session effect $\tau_{it}$ is the effect at session $t$ of having started the intervention, and the case effect $\tau_i$ is its average over the window. The case effect therefore describes the window sessions, at which the case has been treated for at most $m_i$ sessions, and it does not describe the post-window sessions. When the treated outcome depends on the duration of the intervention, $\tau_i$ also depends on the onset probabilities, so cases with different windows or onset probabilities have case effects that refer to different sessions and durations.

Cases are often measured on different scales. Let $c_i>0$ be a scale computed from pre-window outcomes of case $i$ or of the cases in its study. Two examples are the standard deviation of the pre-window outcomes pooled across the cases in a study, and the pre-window mean of the case. The second makes the effect a proportional change. The standardized case effect and the pooled effect are
\begin{equation*}
    \theta_i=\frac{\tau_i}{c_i},\qquad\theta=\sum_{i=1}^N\omega_i\theta_i,
\end{equation*}
where the weights $\omega_i\ge0$ sum to one and are fixed by the design. Equal weights $\omega_i=1/N$ give every case the same weight, and weights that sum to the same value within each study give every study the same weight. Block $g$ has weight $\Omega_g=\sum_{i\in\cC_g}\omega_i>0$ and effect $\vartheta_g=\Omega_g^{-1}\sum_{i\in\cC_g}\omega_i\theta_i$. Under Assumption \ref{ass:anticipation}, $c_i$ takes the same value under every onset, so $\theta_i$ and $\theta$ are fixed quantities. The effect $\theta_i$ is the case effect in units of the scale $c_i$.

\section{Estimation}

Let $\overline Y_{i\cA}=a_i^{-1}\sum_{t\in\cA_i}Y_{it}$ and $\overline Y_{i\cB}=n_i^{-1}\sum_{t\in\cB_i}Y_{it}$ be the pre-window mean and the post-window mean of case $i$, and let $q_\nu$ be the $(1-\alpha/2)$-quantile of the $t$ distribution with $\nu$ degrees of freedom, for a level $\alpha\in(0,1)$. Parts (e) and (f) of Definition \ref{def:estimators} below require $\max_g\Omega_g<1/2$, so that $G\ge3$.

\begin{definition}
    \label{def:estimators}
    \begin{enumerate}
        \item[(a)] \emph{Weighted contrast.}
        \begin{equation*}
            \htau^{\rm a}_i=\frac1{m_i}\sum_{t\in\cW_i}\left\{\frac{W_{it}Y_{it}}{\pi_{it}}-\frac{(1-W_{it})Y_{it}}{1-\pi_{it}}\right\}.
        \end{equation*}
        \item[(b)] \emph{Baseline-anchored estimator.} $\htau^{\rm b}_i$ is $\htau^{\rm a}_i$ with $Y_{it}$ replaced by $Y_{it}-\overline Y_{i\cA}$.
        \item[(c)] \emph{Two-anchor estimator.}
        \begin{equation*}
            \htau^{\rm c}_i=\overline Y_{i\cB}-\overline Y_{i\cA}+\frac1{m_i}\sum_{t\in\cW_i}\left\{\frac{W_{it}(Y_{it}-\overline Y_{i\cB})}{\pi_{it}}-\frac{(1-W_{it})(Y_{it}-\overline Y_{i\cA})}{1-\pi_{it}}\right\}.
        \end{equation*}
        \item[(d)] \emph{Standardized estimators.} $\htheta^{\rm a}_i=\htau^{\rm a}_i/c_i$, $\htheta^{\rm b}_i=\htau^{\rm b}_i/c_i$ and $\htheta^{\rm c}_i=\htau^{\rm c}_i/c_i$.
        \item[(e)] \emph{Pooled estimators.}
        \begin{equation*}
            \htheta^{\rm a}=\sum_{i=1}^N\omega_i\htheta^{\rm a}_i,\qquad\htheta^{\rm b}=\sum_{i=1}^N\omega_i\htheta^{\rm b}_i,\qquad\htheta^{\rm c}=\sum_{i=1}^N\omega_i\htheta^{\rm c}_i.
        \end{equation*}
        Let $\htheta$ be one of these three estimators, and let $\htheta_i$ be the standardized estimators from which it is computed. With $\hvartheta_g=\Omega_g^{-1}\sum_{i\in\cC_g}\omega_i\htheta_i$ and $\rho=\sum_{g=1}^G\Omega_g^2/(1-2\Omega_g)$, the variance estimator is
        \begin{equation*}
            \hV=\frac1{1+\rho}\sum_{g=1}^G\frac{\Omega_g^2(\hvartheta_g-\htheta)^2}{1-2\Omega_g},
        \end{equation*}
        and the interval is $\htheta\pm q_{G-1}\hV^{1/2}$.
        \item[(f)] \emph{Pooled-anchor estimator.} Let $\widetilde Y_{it}=(Y_{it}-\overline Y_{i\cA})/c_i$ and
        \begin{equation*}
            \begin{aligned}
                H_{i1}&=\frac1{m_i}\sum_{t\in\cW_i}\frac{W_{it}\widetilde Y_{it}}{\pi_{it}},&\qquad X_{i1}&=\frac1{m_i}\sum_{t\in\cW_i}\frac{W_{it}}{\pi_{it}},\\
                H_{i0}&=\frac1{m_i}\sum_{t\in\cW_i}\frac{(1-W_{it})\widetilde Y_{it}}{1-\pi_{it}},&\qquad X_{i0}&=\frac1{m_i}\sum_{t\in\cW_i}\frac{1-W_{it}}{1-\pi_{it}}.
            \end{aligned}
        \end{equation*}
        For $w\in\{0,1\}$, let $\mu_w=\bbE(\sum_{i=1}^N\omega_iH_{iw})$ and $\hmu_w=\sum_{i=1}^N\omega_iH_{iw}/\sum_{i=1}^N\omega_iX_{iw}$, with $\hmu_w=0$ when the denominator is 0. Then, $\htheta^{\rm f}=\hmu_1-\hmu_0$. With $\widehat e_i=(H_{i1}-\hmu_1X_{i1})-(H_{i0}-\hmu_0X_{i0})$ and $\widehat r_g=\Omega_g^{-1}\sum_{i\in\cC_g}\omega_i\widehat e_i$, define
        \begin{equation*}
            \hV^{\rm f}=\frac{G-1}{G-2}\cdot\frac1{1+\rho}\sum_{g=1}^G\frac{\Omega_g^2\,\widehat r_g^{\,2}}{1-2\Omega_g},
        \end{equation*}
        and the interval is $\htheta^{\rm f}\pm q_{G-2}(\hV^{\rm f})^{1/2}$.
    \end{enumerate}
\end{definition}

The estimator $\htau^{\rm a}_i$ weights each window outcome by the inverse probability of the condition under which it was observed, as in \citet{HorvitzThompson1952}. The estimators $\htau^{\rm b}_i$ and $\htau^{\rm c}_i$ subtract an offset from each outcome before weighting and add the difference of the offsets back. Specifically, in $\htau^{\rm b}_i$, the offset is the pre-window mean on both sides; in $\htau^{\rm c}_i$, the treated side uses the post-window mean, so each side is centered at a mean observed under its own condition. The pooled-anchor estimator $\htheta^{\rm f}$ centers every outcome at the pre-window mean of its case and scales it. It then divides each pooled weighted sum by the pooled sum of its weights. The two ratios $\hmu_1$ and $\hmu_0$ are weighted means of the treated and the untreated window outcomes of all cases, and $\htheta^{\rm f}$ is their difference.

\begin{assumption}[Stable effect]
    \label{ass:duration}
    $Y_{it}(s)=Y_{it}(s')$ for every case $i$, every session $t$ and all $s,s'\in\cS_i$ with $s\le t$ and $s'\le t$.
\end{assumption}

Under Assumption \ref{ass:duration}, the treated outcome at a session does not depend on the duration of the intervention. We write $Y^{(1)}_{it}$ for it, so that $\tau_{it}=Y^{(1)}_{it}-Y^{(0)}_{it}$. The effect is present in full at the onset session, and it may still differ across sessions and cases. An effect that builds up gradually over sessions \citep{SwanPustejovsky2018} violates this assumption. Assumptions \ref{ass:anticipation} and \ref{ass:duration} correspond to the no-anticipation and invariance-to-history assumptions of \citet{AtheyImbens2022} for staggered adoption.

\begin{theorem}
    \label{thm:unbiased}
    Under Assumption \ref{ass:anticipation}, $\bbE(\htau^{\rm a}_i)=\bbE(\htau^{\rm b}_i)=\tau_i$ for every case $i$. Under Assumptions \ref{ass:anticipation} and \ref{ass:duration}, $\bbE(\htau^{\rm c}_i)=\tau_i$ for every case $i$.
\end{theorem}

\begin{corollary}
    \label{cor:standardized}
    Under Assumption \ref{ass:anticipation}, $\bbE(\htheta^{\rm a}_i)=\bbE(\htheta^{\rm b}_i)=\theta_i$ for every case $i$, and $\bbE(\htheta^{\rm a})=\bbE(\htheta^{\rm b})=\theta$. Under Assumptions \ref{ass:anticipation} and \ref{ass:duration}, $\bbE(\htheta^{\rm c}_i)=\theta_i$ for every case $i$, and $\bbE(\htheta^{\rm c})=\theta$.
\end{corollary}

The proofs of Theorem \ref{thm:unbiased} and Corollary \ref{cor:standardized} are in Section \ref{sec:unbiased} of the Supplemental Material. At a window session, the case is treated with probability $\pi_{it}$. The weighted treated outcome therefore has expectation $\overline Y^{(1)}_{it}$, and the weighted untreated outcome has expectation $Y^{(0)}_{it}$. An offset does not change the expectation when it takes the same value under every onset. The pre-window mean has this property under Assumption \ref{ass:anticipation}, and the post-window mean has it under Assumption \ref{ass:duration}. The scale $c_i$ has the same property, so dividing by $c_i$ keeps the estimator unbiased. The result holds for a single case with a single window session (two consecutive admissible onset sessions), and it requires no model for trend or serial dependence.

Consider the example in Figure \ref{fig:window}, with outcomes $(2,3,2,6,7,6,8,7)$ and onset $S_i=4$. The bottom panel of the figure shows these outcomes. Sessions 3 to 5 are window sessions, and session 3 is untreated. Here $\htau^{\rm a}_i=(-2/0.75+6/0.5+7/0.75)/3=6.22$. The pre-window mean is 2.5, and $\htau^{\rm b}_i=(0.5/0.75+3.5/0.5+4.5/0.75)/3=4.56$. The post-window mean is 7, and $\htau^{\rm c}_i=4.5+(0.5/0.75-1/0.5+0/0.75)/3=4.06$. The offsets remove the level of the series from the weighted terms, which is the main source of variance in $\htau^{\rm a}_i$.

\section{Phase Contrasts}

Most effect sizes for single-case designs compare the outcomes from the onset on with the outcomes before the onset. The simplest comparison is the phase contrast, i.e., the mean outcome from the onset on minus the mean outcome before the onset. Under onset $s$, it is
\begin{equation}
    \gamma_i(s)=\frac1{T_i-s+1}\sum_{t=s}^{T_i}Y_{it}(s)-\frac1{s-1}\sum_{t=1}^{s-1}Y_{it}(s),
    \label{eq:contrast}
\end{equation}
and the observed contrast is $\gamma_i(S_i)$. For the example in Figure \ref{fig:window}, it is $6.8-2.33=4.47$. Under Assumption \ref{ass:duration}, define $\tau_{it}=Y^{(1)}_{it}-Y^{(0)}_{it}$ for post-window sessions, and let
\begin{equation}
    \lambda_{it}=\sum_{s\in\cS_i,\,s\le t}\frac{p_i(s)}{T_i-s+1},\qquad\zeta_i=\bbE\left\{\frac1{T_i-S_i+1}\sum_{t=S_i}^{T_i}Y^{(0)}_{it}-\frac1{S_i-1}\sum_{t=1}^{S_i-1}Y^{(0)}_{it}\right\}.
    \label{eq:drift}
\end{equation}

\begin{proposition}
    \label{prop:contrast}
    Under Assumptions \ref{ass:anticipation} and \ref{ass:duration},
    \begin{equation*}
        \bbE\{\gamma_i(S_i)\}=\sum_{t=a_i+1}^{T_i}\lambda_{it}\tau_{it}+\zeta_i,
    \end{equation*}
    where $\lambda_{it}\ge0$ and $\sum_{t=a_i+1}^{T_i}\lambda_{it}=1$. If $Y^{(0)}_{it}=\phi_i+\beta_it$ for every session $t$, then $\zeta_i=\beta_iT_i/2$.
\end{proposition}

The proof of Proposition \ref{prop:contrast} is in Section \ref{sec:contrast} of the Supplemental Material. The drift $\zeta_i$ is the expected phase contrast of the untreated series, i.e., the contrast that the case would show if the intervention had no effect. Every onset compares later sessions with earlier sessions, so averaging over onsets does not remove it. With a linear trend, the drift is $1/2$ the slope times the number of sessions for each onset distribution, as \citet{MichielsOnghena2019} found for randomized AB phase designs. A case that improves by 0.1 per session over 20 sessions has a drift of 1. An event that shifts the untreated outcome during the study, such as a change of teacher, also contributes to $\zeta_i$. The drift also affects nonoverlap measures. If the untreated series strictly increases, every later outcome exceeds every earlier one, and the nonoverlap of all pairs \citep{ParkerVannest2009} equals 1 at every onset when the intervention has no effect. The estimators $\htau^{\rm a}_i$, $\htau^{\rm b}_i$ and $\htau^{\rm c}_i$ contrast a treated and an untreated outcome within the same session, so changes in the untreated outcome across sessions cancel out in expectation. The weights $\lambda_{it}$ also differ from the equal weights of $\tau_i$, as they include the post-window sessions and do not decrease over sessions.

\section{Pooling Cases and Studies}

This section gives the variance of the pooled estimator and confidence intervals for the pooled effect $\theta$.

\begin{proposition}
    \label{prop:variance}
    Let Assumption \ref{ass:anticipation} hold, and Assumption \ref{ass:duration} as well when $\htheta=\htheta^{\rm c}$. Every pooled estimator $\htheta$ in Definition \ref{def:estimators}(e) and its variance estimator $\hV$ satisfy
    \begin{equation}
        \Var(\htheta)=\sum_{g=1}^G\Omega_g^2\Var(\hvartheta_g),\qquad\bbE(\hV)=\Var(\htheta)+\frac1{1+\rho}\sum_{g=1}^G\frac{\Omega_g^2(\vartheta_g-\theta)^2}{1-2\Omega_g}.
        \label{eq:variance}
    \end{equation}
    When $\Omega_g=1/G$ for every block, $\hV=\sum_{g=1}^G(\hvartheta_g-\htheta)^2/\{G(G-1)\}$.
\end{proposition}

The proof of Proposition \ref{prop:variance} is in Section \ref{sec:variance} of the Supplemental Material. Blocks are randomized independently, so the block estimates $\hvartheta_g$ are independent, and $\hV$ estimates $\Var(\htheta)$ from their squared deviations from $\htheta$. The variance of a block estimate cannot be estimated from that block alone, because each block receives one draw of its onsets. A matched-pairs experiment has the same feature. Since the outcomes are fixed and the only randomness is the draw of the onsets, serial dependence within a case is not a source of variance and needs no model. The last term of \eqref{eq:variance} is nonnegative. The variance estimator is unbiased when all block effects are equal and conservative otherwise, as Neyman's variance estimator in a randomized experiment \citep{Ding2024}. When $\Omega_g$ is the share of the cases that belong to block $g$, $\hV$ is the unified small block variance estimator of \citet[Equation 8]{PashleyMiratrix2021} for blocks of varying size. With equal block weights, $\hV$ is the squared standard error of the mean of the block estimates.

Confidence intervals require conditions on the designs as the number of blocks grows. For the pooled-anchor estimator, let $e_i=(H_{i1}-\mu_1X_{i1})-(H_{i0}-\mu_0X_{i0})$, with $\mu_1$ and $\mu_0$ from Definition \ref{def:estimators}(f). Let $L$ denote $\htheta-\theta$ for a pooled estimator $\htheta$ in Definition \ref{def:estimators}(e) and $\sum_{i=1}^N\omega_ie_i$ for the pooled-anchor estimator $\htheta^{\rm f}$. Let $C_\omega<\infty$, $M<\infty$, $\lpi>0$ and $\lv>0$ be constants that do not depend on $G$.

\begin{assumption}[Weights]
    \label{ass:weights}
    $\Omega_g\le C_\omega/G$ for every block $g$.
\end{assumption}

\begin{assumption}[Bounded outcomes]
    \label{ass:bounded}
    $|Y_{it}(s)-\overline Y_{i\cA}|\le Mc_i$ for every case $i$, every session $t$ and every $s\in\cS_i$.
\end{assumption}

\begin{assumption}[Overlap]
    \label{ass:overlap}
    $\lpi\le\pi_{it}\le1-\lpi$ for every case $i$ and every window session $t\in\cW_i$.
\end{assumption}

\begin{assumption}[Nondegenerate variance]
    \label{ass:nondegenerate}
    $G\Var(L)\ge\lv$.
\end{assumption}

Assumption \ref{ass:weights} keeps every block weight within a fixed multiple of the average block weight. Assumption \ref{ass:bounded} bounds the scaled deviations of the outcomes from the pre-window mean. It holds for outcomes with a bounded range, such as the number of correct answers on a test, when the scales stay away from 0. Assumption \ref{ass:overlap} holds with $\lpi=1/K$ when every case has at most $K$ equally likely onsets. Assumption \ref{ass:nondegenerate} is imposed on the $L$ of the estimator under consideration. It excludes estimators whose variance is of smaller order than $1/G$.

\begin{theorem}
    \label{thm:coverage}
    Let Assumption \ref{ass:anticipation} hold, and consider a sequence of populations and designs with $G\to\infty$ along which Assumptions \ref{ass:weights}--\ref{ass:nondegenerate} hold with the same constants. For $\htheta=\htheta^{\rm b}$ with its variance estimator $\hV$, and for the pooled-anchor estimator $\htheta^{\rm f}$,
    \begin{equation}
        \liminf_{G\to\infty}\bbP\left\{|\htheta-\theta|\le q_{G-1}\hV^{1/2}\right\}\ge1-\alpha,\qquad\liminf_{G\to\infty}\bbP\left\{|\htheta^{\rm f}-\theta|\le q_{G-2}(\hV^{\rm f})^{1/2}\right\}\ge1-\alpha,
        \label{eq:coverage}
    \end{equation}
    and $\bbE(\htheta^{\rm f})-\theta\to0$. Under Assumption \ref{ass:duration}, in addition, the first inequality in \eqref{eq:coverage} holds for $\htheta=\htheta^{\rm c}$.
\end{theorem}

The proof of Theorem \ref{thm:coverage} is in Section \ref{sec:coverage} of the Supplemental Material, which also gives a lower bound on the coverage for every $G\ge4C_\omega$. The pooled estimator is a weighted sum of independent bounded block estimates. It is asymptotically normal when no block dominates, and by Proposition \ref{prop:variance}, $\bbE(\hV)$ is at least $\Var(\htheta)$. With sufficient blocks, $\hV$ is unlikely to fall far below it. The interval therefore has asymptotic coverage of at least $1-\alpha$. When the block effects differ, $\bbE(\hV)$ exceeds $\Var(\htheta)$, and the interval can be conservative. The pooled-anchor estimator requires Assumption \ref{ass:anticipation} only, so the effect may build up over the sessions after the onset. Notably, it is a difference of two ratios, with a bias at most a constant divided by $G$, and an error $\htheta^{\rm f}-\theta$ equaling the linear term $L$ plus a remainder of smaller order. The variance estimator $\hV^{\rm f}$ has the form of $\hV$ with the residuals $\widehat e_i$ in place of the $\htheta_i$, and it uses the factor $(G-1)/(G-2)$ and the quantile $q_{G-2}$ to allow for the two estimated anchors $\hmu_1$ and $\hmu_0$. The replication unit in Theorem \ref{thm:coverage} is the block, and therefore a synthesis of studies that each permute a list of onsets has as many blocks as studies. Suppose that the cases of each such study have the same weight. The weighted sum of the $X_{iw}$ over all cases then equals 1 for every permutation, and $\htheta^{\rm b}$ equals $\htheta^{\rm f}$ (see Section \ref{sec:permuted} of the Supplemental Material for details).

\section{Planning}

This section works under a model for the outcomes. It gives the variances of $\htheta^{\rm b}$ and $\htheta^{\rm c}$ and of the linear term $L$ of $\htheta^{\rm f}$, for use in planning a study or a synthesis. It also shows that a standard deviation computed from a few sessions, when used as the scale, inflates the standardized effect on average. Consider the model
\begin{equation}
    Y_{it}(s)=\phi_i+\tau_iW_{it}(s)+\varepsilon_{it},
    \label{eq:working}
\end{equation}
where the $\varepsilon_{it}$ are uncorrelated across sessions and cases, with mean 0 and variance $\sigma_i^2$, and are independent of the onsets. The onsets of different cases are independent, and the scales $c_i$ are fixed positive numbers. In this section, expectations and variances are taken over the onsets and the $\varepsilon_{it}$. The linear term $L$ of the pooled-anchor estimator depends on $\mu_w$ in Definition \ref{def:estimators}(f), where the expectation is now taken over the onsets and the $\varepsilon_{it}$. Under model \eqref{eq:working}, a treated window outcome exceeds the pre-window mean of its case by $\tau_i$ on average, and an untreated window outcome equals it on average. Hence $\mu_1=\theta$ and $\mu_0=0$. Model \eqref{eq:working} also satisfies Assumptions \ref{ass:anticipation} and \ref{ass:duration}. Let
\begin{equation}
    \psi_i=\frac1{m_i^2}\sum_{t\in\cW_i}\frac1{\pi_{it}(1-\pi_{it})}.
    \label{eq:psi}
\end{equation}

\begin{proposition}
    \label{prop:planning}
    Under model \eqref{eq:working} with independent onsets across cases and fixed scales $c_i$, the pooled estimators $\htheta^{\rm b}$ and $\htheta^{\rm c}$ and the linear term $L$ of the pooled-anchor estimator satisfy
    \begin{align}
        \Var(\htheta^{\rm b})&=\sum_{i=1}^N\omega_i^2\left[\theta_i^2\Var(X_{i1})+\frac{\sigma_i^2}{c_i^2}\left\{\psi_i+\frac{\Var(X_{i1}-X_{i0})}{a_i}\right\}\right],\label{eq:planb}\displaybreak[0]\\
        \Var(\htheta^{\rm c})&=\sum_{i=1}^N\omega_i^2\,\frac{\sigma_i^2}{c_i^2}\left\{\psi_i+\frac{\Var(X_{i1})}{n_i}+\frac{\Var(X_{i0})}{a_i}\right\},\label{eq:planc}\displaybreak[0]\\
        \Var(L)&=\sum_{i=1}^N\omega_i^2\left[(\theta_i-\theta)^2\Var(X_{i1})+\frac{\sigma_i^2}{c_i^2}\left\{\psi_i+\frac{\Var(X_{i1}-X_{i0})}{a_i}\right\}\right].\label{eq:planf}
    \end{align}
\end{proposition}

The proof is in Section \ref{sec:planning} of the Supplemental Material. The term $\psi_i$ comes from the noise at the window sessions, and for equally likely consecutive onsets it decreases as the window grows. The terms divided by $a_i$ and $n_i$ in \eqref{eq:planb}--\eqref{eq:planf} come from the noise in the pre-window and post-window means, so several sessions on each side of the window reduce them. The sum of treated weights $X_{i1}$ varies with the onset. The average treated weight $X_{i1}$ varies with the onset and has mean 1. In $\htheta^{\rm b}_i$, every treated window outcome exceeds the pre-window mean by $\theta_i$ on average, and thus $\htheta^{\rm b}_i$ equals $\theta_iX_{i1}$ plus noise, which gives the term $\theta_i^2\Var(X_{i1})$ in \eqref{eq:planb}. The pooled-anchor estimator replaces $\theta_i$ by $\theta_i-\theta$ in this term in \eqref{eq:planf}, and $\Var(\htheta^{\rm c})$ in \eqref{eq:planc} does not have it. The variances of $X_{i1}$ and $X_{i0}$ depend on the design only. For $K$ consecutive and equally likely onsets, with $h_K=\sum_{j=1}^{K-1}1/j$, we have
\begin{equation}
    \psi_i=\frac{2Kh_K}{(K-1)^2},\qquad\Var(X_{i1})=\Var(X_{i0})=\frac{K\{2(K-1)-h_K\}}{(K-1)^2}-1.
    \label{eq:uniform}
\end{equation}

Table \ref{tab:plan} illustrates \eqref{eq:planc} and \eqref{eq:planf} for one case under model \eqref{eq:working}. The case has $K$ consecutive and equally likely onsets and the same number of sessions before and after the window, i.e., $a_i=n_i$. The scale is $c_i=\sigma_i$, so effects are measured in units of $\sigma_i$, and all cases have the same standardized effect, resulting in $\theta_i=\theta$. With $c_i=\sigma_i$ and $\theta_i=\theta$, \eqref{eq:planc} and \eqref{eq:planf} become
\begin{align*}
    \Var(\htheta^{\rm c})&=\sum_{i=1}^N\omega_i^2\left\{\psi_i+\frac{\Var(X_{i1})}{n_i}+\frac{\Var(X_{i0})}{a_i}\right\},\displaybreak[0]\\
    \Var(L)&=\sum_{i=1}^N\omega_i^2\left\{\psi_i+\frac{\Var(X_{i1}-X_{i0})}{a_i}\right\}.
\end{align*}
The table lists $\psi_i$ and the square roots of the two terms, and Table \ref{tab:planfull} of the Supplemental Material adds $\Var(X_{i1})$ and $\Var(X_{i1}-X_{i0})$ for these and more designs. With $N$ such cases and equal weights, the standard error of the pooled estimate is the table entry divided by $N^{1/2}$. For example, with eight admissible onsets and five sessions on each side of the window, 20 cases give a standard error of $1.09/20^{1/2}=0.24$ for $\htheta^{\rm c}$ and a standard deviation of $1.21/20^{1/2}=0.27$ for the linear term of the pooled-anchor estimator. Doubling the number of onsets from 4 to 8 reduces the entry of $\htheta^{\rm c}$ from 1.40 to 1.09. The case then has four more sessions, since the table holds $a_i$ and $n_i$ fixed. The last column gives the standard deviation of the phase contrast in the same setting and in units of $\sigma_i$, i.e., $[\bbE\{1/(S_i-1)+1/(T_i-S_i+1)\}]^{1/2}$, which is about half of the entries for $\htheta^{\rm c}$ and $\htheta^{\rm f}$ or less. The phase contrast uses every session of the case in its two means, and its expectation includes the drift of Proposition \ref{prop:contrast}.

\begin{table}[t]
    \caption{Variance Term $\psi_i$ and Standard Deviation Contributed by One Case Under Model \eqref{eq:working}}
    \label{tab:plan}
    \centering
    \small
    \begin{tabular}{cccccc}
        \toprule
        Onsets & Sessions on each side & & \multicolumn{3}{c}{Standard deviation, in units of $\sigma_i$} \\
        \cmidrule(lr){4-6}
        $K$ & of the window, $a_i=n_i$ & $\psi_i$ & $\htheta^{\rm c}$ & $\htheta^{\rm f}$ & Phase contrast \\
        \midrule
        4 & 3 & 1.63 & 1.48 & 1.64 & 0.69 \\
        4 & 5 & 1.63 & 1.40 & 1.51 & 0.56 \\
        4 & 8 & 1.63 & 1.36 & 1.43 & 0.46 \\
        8 & 3 & 0.85 & 1.19 & 1.37 & 0.60 \\
        8 & 5 & 0.85 & 1.09 & 1.21 & 0.50 \\
        8 & 8 & 0.85 & 1.03 & 1.11 & 0.43 \\
        16 & 3 & 0.47 & 1.03 & 1.23 & 0.50 \\
        16 & 5 & 0.47 & 0.91 & 1.05 & 0.44 \\
        16 & 8 & 0.47 & 0.83 & 0.93 & 0.38 \\
        \bottomrule
    \end{tabular}
    \par\smallskip
    \parbox{0.9\textwidth}{\small\textit{Note.} One case under model \eqref{eq:working} with $K$ consecutive and equally likely onsets, $a_i=n_i$ sessions before and after the window, $T_i=2a_i+K-1$ sessions in total, $c_i=\sigma_i$, and the same standardized effect for all cases. The entry for $\htheta^{\rm c}$ is $\{\psi_i+\Var(X_{i1})/n_i+\Var(X_{i0})/a_i\}^{1/2}$, from \eqref{eq:planc}. The entry for $\htheta^{\rm f}$ is $\{\psi_i+\Var(X_{i1}-X_{i0})/a_i\}^{1/2}$, from \eqref{eq:planf}, and it is the standard deviation of the linear term $L$. The variance of $\htheta^{\rm b}$ also depends on the size of the effect and is not listed.}
\end{table}

The estimators $\htheta^{\rm b}$, $\htheta^{\rm c}$ and $\htheta^{\rm f}$ suit different interventions. When all cases have the same standardized effect and $n_i\ge a_i$, the variances are ordered as $\Var(\htheta^{\rm c})\le\Var(L)\le\Var(\htheta^{\rm b})$ (Section \ref{sec:planning} of the Supplemental Material). The estimator $\htheta^{\rm c}$ requires Assumption \ref{ass:duration}, and thus it suits interventions whose effect is present in full at the onset session. When the effect grows with the number of sessions since the onset, an earlier onset gives a larger post-window mean, and $\htau^{\rm c}_i$ has the bias $-\Cov(X_{i1},\overline Y_{i\cB})\le0$ (Section \ref{sec:biasc} of the Supplemental Material). The pooled-anchor estimator requires Assumption \ref{ass:anticipation} only, so we recommend it for such interventions. The estimator $\htheta^{\rm b}$ is unbiased for every $G$ under Assumption \ref{ass:anticipation}, and its variance includes the term $\theta_i^2\Var(X_{i1})$. For $K=8$, $\theta_i=2$ and $c_i=\sigma_i$, this term is $4\times0.862=3.45$ by \eqref{eq:uniform}, and the remaining terms of \eqref{eq:planb} sum to 1.47 with five sessions on each side of the window.

Proposition \ref{prop:planning} treats the scales as fixed numbers. A standard deviation computed from pre-window outcomes is random under model \eqref{eq:working}. Let $\cN_i$ be the set of cases in the same study as case $i$, with case $i$ included. Let $c_i$ be the standard deviation of the pre-window outcomes pooled over the cases in $\cN_i$, $\nu_i$ be its number of degrees of freedom, and $\iota_i$ be a factor that depends on $\nu_i$ only. They are given by
\begin{equation}
    \begin{aligned}
        c_i^2&=\frac1{\nu_i}\sum_{i'\in\cN_i}\sum_{t\in\cA_{i'}}\left(Y_{i't}-\overline Y_{i'\cA}\right)^2,\qquad\nu_i=\sum_{i'\in\cN_i}(a_{i'}-1),\\
        \iota_i&=\frac{\Gamma(\nu_i/2)}{(\nu_i/2)^{1/2}\,\Gamma\{(\nu_i-1)/2\}},
    \end{aligned}
    \label{eq:pooledsd}
\end{equation}
where $\Gamma$ is the gamma function.

\begin{proposition}
    \label{prop:scale}
    Let the $\varepsilon_{i't}$ in model \eqref{eq:working} be independent and normal with the same variance $\sigma^2$ for all cases $i'\in\cN_i$, and let $\nu_i\ge2$. Then, $\bbE(\tau_i/c_i)=\tau_i/(\sigma\iota_i)$.
\end{proposition}

The proof is in Section \ref{sec:scale} of the Supplemental Material. The factor $\iota_i$ is the correction factor of \citet[Equation 6e]{Hedges1981} for a standardized mean difference whose standard deviation has $\nu_i$ degrees of freedom. It is smaller than 1, so by Proposition \ref{prop:scale}, the standardized effect $\theta_i=\tau_i/c_i$ is on average $1/\iota_i$ times as large as $\tau_i/\sigma$. This is because $c_i$ varies around $\sigma$ when computed from a few sessions, and a small $c_i$ raises $1/c_i$ by more than a large $c_i$ lowers it. The inflation is large when $\nu_i$ is small. A single case with three pre-window sessions has $\nu_i=2$ and $1/\iota_i=1.77$. Six cases with three pre-window sessions each have $\nu_i=12$ and $1/\iota_i=1.07$. Dividing $c_i$ by $\iota_i$ removes the inflation, because the mean of $\iota_i\tau_i/c_i$ is $\tau_i/\sigma$ by Proposition \ref{prop:scale}. The scale $c_i/\iota_i$ is computed from pre-window outcomes only, so Corollary \ref{cor:standardized} holds for it. For a synthesis in standard deviation units, we therefore use the scale $c_i/\iota_i$.

\section{Simulation Study}

This section reports three simulation studies. Section \ref{sec:simulations} of the Supplemental Material gives their full results and four more studies, on permuted onsets, standardized effects, the number of cases under a common shift, and the scale of Proposition \ref{prop:scale}. In every study, each of the 1,000 replications draws a new population and new onsets.

\subsection{Study 1: Baselines That Change Over Sessions}

This study examines Theorem \ref{thm:unbiased} and Proposition \ref{prop:contrast} by comparing the estimators in Definition \ref{def:estimators} with the phase contrast and with multilevel models when the untreated outcome changes over sessions. There are 24 cases with 20 sessions each, and the onsets are drawn independently and uniformly from sessions 6 to 13. The untreated outcome is a case level drawn from a normal distribution with mean 10 and standard deviation 2, plus a baseline pattern and a first-order autoregressive noise with unit variance and autocorrelation 0.3. The pattern is the same for all cases and is flat, a linear trend of 0.1 per session, or a shift of 1.5 from session 10 onward. The effect of case $i$ is drawn from a normal distribution with mean 2 and standard deviation 0.5 and is present from the onset session on, so Assumptions \ref{ass:anticipation} and \ref{ass:duration} hold. We compute $\htheta^{\rm b}$, $\htheta^{\rm c}$, the pooled-anchor estimator $\htheta^{\rm f}$, and the mean of the phase contrasts with the variance estimator in Definition \ref{def:estimators}(e), all with $c_i=1$ and equal weights. We also fit two multilevel models by restricted maximum likelihood, with a random intercept for each case and first-order autoregressive errors \citep{PinheiroBates2000}. One model includes only the treatment indicator as a fixed effect, and the other adds a linear time trend. Intervals have nominal level 95\%. Table \ref{tab:sim1} reports the bias, the standard deviation of the estimation error, and the coverage of the intervals.

The results validate Theorem \ref{thm:unbiased} and Proposition \ref{prop:contrast}. The bias of $\htheta^{\rm b}$, $\htheta^{\rm c}$, and $\htheta^{\rm f}$ is at most 0.03 in absolute value under every pattern, and the phase contrast is biased when the baseline changes. Under the linear trend, the bias of the phase contrast is 1.00, and Proposition \ref{prop:contrast} gives $\beta_iT_i/2=0.1\times20/2=1$. Under the common shift, it is 1.26. The multilevel model without a trend has biases of 0.88 and 1.08 under these two patterns. The model with a linear trend has a bias of 0.00 under the linear trend and 0.38 under the common shift, and its intervals cover the effect in 56.4\% of replications. The intervals of $\htheta^{\rm b}$, $\htheta^{\rm c}$ and $\htheta^{\rm f}$ cover the effect in 92.3\% to 96.5\% of the replications. The lowest rate is for $\htheta^{\rm f}$ under the common shift, with a bias of 0.02 with 24 cases. Section \ref{sec:sim6} of the Supplemental Material reports this pattern with 6 to 96 cases. The standard deviations of $\htheta^{\rm b}$, $\htheta^{\rm c}$, and $\htheta^{\rm f}$ are larger than those of the phase contrast and the multilevel models under all three baselines, but their biases stay near 0 under every pattern.

\begin{table}[t]
    \caption{Study 1: Bias, Standard Deviation (SD) and Coverage of 95\% Intervals (\%) Under Three Baseline Patterns}
    \label{tab:sim1}
    \centering
    \footnotesize
    \setlength{\tabcolsep}{4pt}
    \begin{tabular}{lccccccccc}
        \toprule
        & \multicolumn{3}{c}{Flat} & \multicolumn{3}{c}{Linear trend} & \multicolumn{3}{c}{Common shift} \\
        \cmidrule(lr){2-4}\cmidrule(lr){5-7}\cmidrule(lr){8-10}
        Method & Bias & SD & Cov. & Bias & SD & Cov. & Bias & SD & Cov. \\
        \midrule
        $\htheta^{\rm b}$ & 0.00 & 0.50 & 94.3 & $-0.01$ & 0.69 & 94.7 & $-0.03$ & 0.67 & 95.8 \\
        $\htheta^{\rm c}$ & 0.01 & 0.26 & 96.5 & 0.00 & 0.29 & 95.1 & 0.00 & 0.37 & 93.6 \\
        $\htheta^{\rm f}$ & 0.00 & 0.31 & 94.8 & 0.01 & 0.32 & 94.4 & 0.02 & 0.37 & 92.3 \\
        Phase contrast & 0.00 & 0.12 & 99.0 & 1.00 & 0.12 & 0.0 & 1.26 & 0.13 & 0.0 \\
        Multilevel, no trend & 0.00 & 0.12 & 96.7 & 0.88 & 0.12 & 0.0 & 1.08 & 0.14 & 0.0 \\
        Multilevel, linear trend & 0.00 & 0.20 & 95.5 & 0.00 & 0.20 & 95.1 & 0.38 & 0.23 & 56.4 \\
        \bottomrule
    \end{tabular}
    \par\smallskip
    \parbox{0.95\textwidth}{\small\textit{Note.} 24 cases, 20 sessions, onsets uniform on sessions 6 to 13, over 1,000 replications. The average pooled effect is 2.00.}
\end{table}

\subsection{Study 2: Number of Cases and Number of Onsets}

This study examines Proposition \ref{prop:planning} and Theorem \ref{thm:coverage} by varying the number of cases and admissible onsets. The outcomes follow model \eqref{eq:working} with $\tau_i=2$ and $\sigma_i=1$ for every case. There are $N=4$, 6, 8, 16, or 32 cases with 20 sessions each, and the onsets are drawn independently and uniformly from $K=4$ or 8 consecutive sessions starting at session 6. We compute $\htheta^{\rm c}$ and $\htheta^{\rm f}$ with $c_i=1$ and equal weights, and intervals with nominal level 95\%. Table \ref{tab:sim2} reports the standard deviation of the estimation error, the standard deviation predicted by Proposition \ref{prop:planning}, and the coverage of the intervals.

The results match Proposition \ref{prop:planning} and Theorem \ref{thm:coverage}. The simulated standard deviations decrease with the number of cases and the number of onsets, as \eqref{eq:planc} and \eqref{eq:planf} state. They are within 5\% of the predicted ones, except for $\htheta^{\rm f}$ with $N=4$ and $K=4$, where the simulated value is 0.805, and the predicted value for its linear term is 0.754. With $K=8$, there are eight post-window sessions, and \eqref{eq:planc} and \eqref{eq:uniform} give $0.847+0.862/8+0.862/5=1.127$ for one case. For $N=8$, the predicted standard deviation of $\htheta^{\rm c}$ is $(1.127/8)^{1/2}=0.375$, and the simulated one is 0.366. For $\htheta^{\rm f}$ with $N=32$ and $K=4$, the two values are 0.266 and 0.261, respectively. The intervals cover the effect in 93.4\% to 97.7\% of the replications, and in 95.2\% to 97.7\% with four or six cases. The Monte Carlo standard error of a coverage rate near 95\% is 0.7 percentage points.

\begin{table}[t]
    \caption{Study 2: Simulated and Predicted Standard Deviation (SD) and Coverage of 95\% Intervals (\%)}
    \label{tab:sim2}
    \centering
    \small
    \begin{tabular}{cccccccc}
        \toprule
        & & \multicolumn{3}{c}{$\htheta^{\rm c}$} & \multicolumn{3}{c}{$\htheta^{\rm f}$} \\
        \cmidrule(lr){3-5}\cmidrule(lr){6-8}
        $K$ & $N$ & SD & Predicted & Cov. & SD & Predicted & Cov. \\
        \midrule
        4 & 4 & 0.684 & 0.684 & 96.1 & 0.805 & 0.754 & 97.7 \\
        4 & 6 & 0.561 & 0.558 & 96.6 & 0.637 & 0.615 & 95.2 \\
        4 & 8 & 0.468 & 0.484 & 95.3 & 0.528 & 0.533 & 94.0 \\
        4 & 16 & 0.336 & 0.342 & 95.3 & 0.375 & 0.377 & 95.5 \\
        4 & 32 & 0.238 & 0.242 & 95.4 & 0.261 & 0.266 & 95.4 \\
        8 & 4 & 0.537 & 0.531 & 96.7 & 0.580 & 0.606 & 97.1 \\
        8 & 6 & 0.428 & 0.433 & 95.6 & 0.471 & 0.494 & 95.6 \\
        8 & 8 & 0.366 & 0.375 & 95.7 & 0.421 & 0.428 & 95.0 \\
        8 & 16 & 0.274 & 0.265 & 94.9 & 0.306 & 0.303 & 93.4 \\
        8 & 32 & 0.196 & 0.188 & 94.0 & 0.224 & 0.214 & 94.3 \\
        \bottomrule
    \end{tabular}
    \par\smallskip
    \parbox{0.9\textwidth}{\small\textit{Note.} 20 sessions, five pre-window sessions, $K$ onsets starting at session 6, over 1,000 replications. Predicted values are the square roots of \eqref{eq:planc} and \eqref{eq:planf}.}
\end{table}

\subsection{Study 3: Effects That Build Up Over Sessions}

This study checks Theorems \ref{thm:unbiased} and \ref{thm:coverage}, and the bias of $\htheta^{\rm c}$ when the effect builds up over the sessions after the onset. The cases, sessions, onsets, and untreated outcomes are those of Study 1 with the flat baseline. Under onset $s$, the effect of case $i$ at session $t\ge s$ is its full effect times $\min\{1,(t-s+1)/R\}$, with $R=1$, 3 or 6, so the effect reaches its full size at the $R$th treated session. Assumption \ref{ass:anticipation} holds, and Assumption \ref{ass:duration} holds only for $R=1$. We compute $\htheta^{\rm b}$, $\htheta^{\rm c}$ and $\htheta^{\rm f}$ with $c_i=1$ and equal weights, and intervals with nominal level 95\%. The bias of $\htheta^{\rm c}$ given by Proposition \ref{prop:biasc} of the Supplemental Material is computed for each population and averaged over the replications. Table \ref{tab:sim3} reports the mean pooled effect, the bias with its Monte Carlo standard error, this predicted bias, the standard deviation of the estimation error, and the coverage of the intervals.

The results corroborate Theorems \ref{thm:unbiased} and \ref{thm:coverage} and Proposition \ref{prop:biasc}. The estimators $\htheta^{\rm b}$ and $\htheta^{\rm f}$, which require Assumption \ref{ass:anticipation} only, have biases within two Monte Carlo standard errors of 0 for every $R$, and $\htheta^{\rm c}$ is biased downward when the effect builds up. The bias of $\htheta^{\rm c}$ is $-0.038$ for $R=3$ and $-0.143$ for $R=6$, and Proposition \ref{prop:biasc} gives $-0.040$ and $-0.149$. The intervals of $\htheta^{\rm f}$ cover the pooled effect in 95.1\% to 95.2\% of the replications, and those of $\htheta^{\rm c}$ in 93.9\% for $R=6$. For $R=6$, $\htheta^{\rm f}$ also has the smallest standard deviation, 0.296 versus 0.319 for $\htheta^{\rm c}$.

\begin{table}[t]
    \caption{Study 3: Bias With Its Monte Carlo Standard Error (MCSE), Bias of $\htheta^{\rm c}$ Given by Proposition \ref{prop:biasc}, Standard Deviation (SD) and Coverage of 95\% Intervals (\%)}
    \label{tab:sim3}
    \centering
    \small
    \begin{tabular}{cclccccc}
        \toprule
        $R$ & Mean $\theta$ & Method & Bias & MCSE & Predicted bias & SD & Cov. \\
        \midrule
        1 & 1.996 & $\htheta^{\rm b}$ & 0.016 & 0.015 & & 0.485 & 95.4 \\
        & & $\htheta^{\rm c}$ & 0.013 & 0.008 & 0.000 & 0.248 & 96.5 \\
        & & $\htheta^{\rm f}$ & 0.013 & 0.010 & & 0.305 & 95.1 \\
        3 & 1.354 & $\htheta^{\rm b}$ & 0.013 & 0.013 & & 0.407 & 95.5 \\
        & & $\htheta^{\rm c}$ & $-0.038$ & 0.009 & $-0.040$ & 0.274 & 95.9 \\
        & & $\htheta^{\rm f}$ & 0.005 & 0.009 & & 0.295 & 95.2 \\
        6 & 0.827 & $\htheta^{\rm b}$ & $-0.005$ & 0.011 & & 0.345 & 94.7 \\
        & & $\htheta^{\rm c}$ & $-0.143$ & 0.010 & $-0.149$ & 0.319 & 93.9 \\
        & & $\htheta^{\rm f}$ & $-0.004$ & 0.009 & & 0.296 & 95.2 \\
        \bottomrule
    \end{tabular}
    \par\smallskip
    \parbox{0.9\textwidth}{\small\textit{Note.} 24 cases, 20 sessions, flat baseline, onsets uniform on sessions 6 to 13, 1,000 replications. The effect reaches its full size at the $R$th treated session.}
\end{table}

\section{Three Studies of Reading and Writing Interventions}

\subsection{Design and Data}

We use three multiple-baseline studies that drew the onset of each student at random. Section \ref{sec:application} of the Supplemental Material presents their designs, estimates for each student, and analyses with equal study weights. \citet{Grunke2013} taught story mapping to six students aged 10 to 14 with reading comprehension problems and recorded the number of correct answers to ten comprehension questions on 18 school days. The onset of each student was drawn from sessions 4 to 9. \citet{Grunke2015} replicated the study with five elementary students with learning disabilities. We used its 15 baseline and intervention sessions, with onsets drawn from sessions 4 to 11. \citet{Grunke2019} taught a planning strategy for writing to four fourth graders and recorded the total number of words written on 12 school days, with onsets drawn from sessions 6 to 10. We refer to these three studies as A, B, and C, respectively. We took the outcomes of every session from the tables of the articles. Each article states that the onset was drawn at random for each student. We take the onsets to be independent across students and equally likely over the admissible sessions, so the 15 students are 15 blocks. In Study C, two students started at session 6 and two at session 8, and a second analysis takes these two pairs as two blocks. There are three pre-window sessions in Studies A and B and five in Study C. Figure \ref{fig:data} shows the 15 series.

\begin{figure}[t]
    \caption{Outcomes of the 15 Students in Studies A, B, and C}
    \label{fig:data}
    \centering
    \begin{tikzpicture}[-,thin,x=1cm,y=1cm,font=\linespread{1}\selectfont\sffamily\scriptsize,line cap=round,line join=round]
        \coordinate (m0) at (0.12,0.95);
        \filldraw[fill=white,draw=vizsec,line width=0.7pt] (m0) circle (1.7pt);
        \node[right,text=vizsec,inner sep=0pt,font=\sffamily\footnotesize] (l0) at ([xshift=0.2cm]m0) {\strut Baseline session};
        \path (l0.east) ++(0.75,0) coordinate (m1);
        \filldraw[fill=vizblue,draw=white,line width=0.5pt] (m1) circle (2.1pt);
        \node[right,text=vizsec,inner sep=0pt,font=\sffamily\footnotesize] (l1) at ([xshift=0.2cm]m1) {\strut Intervention session};
        \path (l1.east) ++(0.9,0) coordinate (m2);
        \fill[vizband,rounded corners=1.2pt] ([xshift=-0.2cm,yshift=-0.13cm]m2) rectangle ([xshift=0.2cm,yshift=0.13cm]m2);
        \node[right,text=vizsec,inner sep=0pt,font=\sffamily\footnotesize] (l2) at ([xshift=0.3cm]m2) {\strut Window sessions};
        \path (l2.east) ++(0.75,0) coordinate (m3);
        \draw[vizorange,line width=1.1pt] ([yshift=-0.15cm]m3) -- ([yshift=0.15cm]m3);
        \node[right,text=vizsec,inner sep=0pt,font=\sffamily\footnotesize] (l3) at ([xshift=0.2cm]m3) {\strut Onset};
        \path (l3.east) ++(0.75,0) coordinate (m4);
        \node[right,text=vizink,inner sep=0pt,font=\sffamily\footnotesize] at (0,-0.14) {\textbf{Study A}\enspace\textcolor{vizsec}{Correct answers out of 10}};
        \fill[vizband,rounded corners=1.5pt] (0.381,-2.68) rectangle (1.015,-0.78);
        \draw[vizaxis,line width=0.35pt] (0,-2.6) -- (2.283,-2.6);
        \node[left,text=vizmuted,inner sep=1.5pt] at (-0.02,-2.6) {0};
        \draw[vizgrid,line width=0.35pt] (0,-1.82) -- (2.283,-1.82);
        \node[left,text=vizmuted,inner sep=1.5pt] at (-0.02,-1.82) {5};
        \draw[vizgrid,line width=0.35pt] (0,-1.04) -- (2.283,-1.04);
        \node[left,text=vizmuted,inner sep=1.5pt] at (-0.02,-1.04) {10};
        \draw[vizorange,line width=1.1pt] (0.507,-2.6) -- (0.507,-0.78);
        \draw[vizmuted,line width=0.6pt] (0.063,-1.976) -- (0.19,-2.132) -- (0.317,-2.132) -- (0.444,-2.288) -- (0.571,-1.508) -- (0.698,-1.352) -- (0.825,-1.04) -- (0.951,-1.196) -- (1.078,-1.04) -- (1.205,-1.04) -- (1.332,-1.04) -- (1.459,-1.04) -- (1.586,-1.196) -- (1.713,-1.196) -- (1.839,-1.196) -- (1.966,-1.352) -- (2.093,-1.196) -- (2.22,-1.352);
        \foreach \p in {(0.063,-1.976),(0.19,-2.132),(0.317,-2.132),(0.444,-2.288)} { \filldraw[fill=white,draw=vizsec,line width=0.7pt] \p circle (1.7pt); }
        \foreach \p in {(0.571,-1.508),(0.698,-1.352),(0.825,-1.04),(0.951,-1.196),(1.078,-1.04),(1.205,-1.04),(1.332,-1.04),(1.459,-1.04),(1.586,-1.196),(1.713,-1.196),(1.839,-1.196),(1.966,-1.352),(2.093,-1.196),(2.22,-1.352)} { \filldraw[fill=vizblue,draw=white,line width=0.5pt] \p circle (2.1pt); }
        \node[anchor=base west,text=vizink,inner sep=0pt] at (0.03,-0.66) {Anna};
        \node[below,text=vizmuted,inner sep=2pt] at (0.063,-2.7) {1};
        \node[below,text=vizmuted,inner sep=2pt] at (0.571,-2.7) {5};
        \node[below,text=vizmuted,inner sep=2pt] at (2.22,-2.7) {18};
        \fill[vizband,rounded corners=1.5pt] (3.004,-2.68) rectangle (3.638,-0.78);
        \draw[vizaxis,line width=0.35pt] (2.623,-2.6) -- (4.907,-2.6);
        \draw[vizgrid,line width=0.35pt] (2.623,-1.82) -- (4.907,-1.82);
        \draw[vizgrid,line width=0.35pt] (2.623,-1.04) -- (4.907,-1.04);
        \draw[vizorange,line width=1.1pt] (3.511,-2.6) -- (3.511,-0.78);
        \draw[vizmuted,line width=0.6pt] (2.687,-2.132) -- (2.814,-1.976) -- (2.94,-2.288) -- (3.067,-2.288) -- (3.194,-2.132) -- (3.321,-1.976) -- (3.448,-2.132) -- (3.575,-1.196) -- (3.702,-1.196) -- (3.828,-1.352) -- (3.955,-1.196) -- (4.082,-1.352) -- (4.209,-1.04) -- (4.336,-1.196) -- (4.463,-1.508) -- (4.59,-1.196) -- (4.716,-1.04) -- (4.843,-1.196);
        \foreach \p in {(2.687,-2.132),(2.814,-1.976),(2.94,-2.288),(3.067,-2.288),(3.194,-2.132),(3.321,-1.976),(3.448,-2.132)} { \filldraw[fill=white,draw=vizsec,line width=0.7pt] \p circle (1.7pt); }
        \foreach \p in {(3.575,-1.196),(3.702,-1.196),(3.828,-1.352),(3.955,-1.196),(4.082,-1.352),(4.209,-1.04),(4.336,-1.196),(4.463,-1.508),(4.59,-1.196),(4.716,-1.04),(4.843,-1.196)} { \filldraw[fill=vizblue,draw=white,line width=0.5pt] \p circle (2.1pt); }
        \node[anchor=base west,text=vizink,inner sep=0pt] at (2.653,-0.66) {Bella};
        \node[below,text=vizmuted,inner sep=2pt] at (2.687,-2.7) {1};
        \node[below,text=vizmuted,inner sep=2pt] at (3.575,-2.7) {8};
        \node[below,text=vizmuted,inner sep=2pt] at (4.843,-2.7) {18};
        \fill[vizband,rounded corners=1.5pt] (5.627,-2.68) rectangle (6.261,-0.78);
        \draw[vizaxis,line width=0.35pt] (5.247,-2.6) -- (7.53,-2.6);
        \draw[vizgrid,line width=0.35pt] (5.247,-1.82) -- (7.53,-1.82);
        \draw[vizgrid,line width=0.35pt] (5.247,-1.04) -- (7.53,-1.04);
        \draw[vizorange,line width=1.1pt] (5.754,-2.6) -- (5.754,-0.78);
        \draw[vizmuted,line width=0.6pt] (5.31,-1.664) -- (5.437,-1.82) -- (5.564,-1.976) -- (5.691,-1.82) -- (5.817,-1.352) -- (5.944,-1.04) -- (6.071,-1.04) -- (6.198,-1.352) -- (6.325,-1.352) -- (6.452,-1.04) -- (6.579,-1.196) -- (6.705,-1.196) -- (6.832,-1.04) -- (6.959,-1.196) -- (7.086,-1.196) -- (7.213,-1.04) -- (7.34,-1.196) -- (7.467,-1.196);
        \foreach \p in {(5.31,-1.664),(5.437,-1.82),(5.564,-1.976),(5.691,-1.82)} { \filldraw[fill=white,draw=vizsec,line width=0.7pt] \p circle (1.7pt); }
        \foreach \p in {(5.817,-1.352),(5.944,-1.04),(6.071,-1.04),(6.198,-1.352),(6.325,-1.352),(6.452,-1.04),(6.579,-1.196),(6.705,-1.196),(6.832,-1.04),(6.959,-1.196),(7.086,-1.196),(7.213,-1.04),(7.34,-1.196),(7.467,-1.196)} { \filldraw[fill=vizblue,draw=white,line width=0.5pt] \p circle (2.1pt); }
        \node[anchor=base west,text=vizink,inner sep=0pt] at (5.277,-0.66) {Christina};
        \node[below,text=vizmuted,inner sep=2pt] at (5.31,-2.7) {1};
        \node[below,text=vizmuted,inner sep=2pt] at (5.817,-2.7) {5};
        \node[below,text=vizmuted,inner sep=2pt] at (7.467,-2.7) {18};
        \fill[vizband,rounded corners=1.5pt] (8.251,-2.68) rectangle (8.885,-0.78);
        \draw[vizaxis,line width=0.35pt] (7.87,-2.6) -- (10.153,-2.6);
        \draw[vizgrid,line width=0.35pt] (7.87,-1.82) -- (10.153,-1.82);
        \draw[vizgrid,line width=0.35pt] (7.87,-1.04) -- (10.153,-1.04);
        \draw[vizorange,line width=1.1pt] (8.631,-2.6) -- (8.631,-0.78);
        \draw[vizmuted,line width=0.6pt] (7.933,-1.976) -- (8.06,-1.82) -- (8.187,-1.976) -- (8.314,-2.132) -- (8.441,-1.976) -- (8.568,-2.132) -- (8.695,-1.352) -- (8.821,-1.196) -- (8.948,-1.196) -- (9.075,-1.196) -- (9.202,-1.352) -- (9.329,-1.04) -- (9.456,-1.196) -- (9.582,-1.352) -- (9.709,-1.04) -- (9.836,-1.196) -- (9.963,-1.352) -- (10.09,-1.196);
        \foreach \p in {(7.933,-1.976),(8.06,-1.82),(8.187,-1.976),(8.314,-2.132),(8.441,-1.976),(8.568,-2.132)} { \filldraw[fill=white,draw=vizsec,line width=0.7pt] \p circle (1.7pt); }
        \foreach \p in {(8.695,-1.352),(8.821,-1.196),(8.948,-1.196),(9.075,-1.196),(9.202,-1.352),(9.329,-1.04),(9.456,-1.196),(9.582,-1.352),(9.709,-1.04),(9.836,-1.196),(9.963,-1.352),(10.09,-1.196)} { \filldraw[fill=vizblue,draw=white,line width=0.5pt] \p circle (2.1pt); }
        \node[anchor=base west,text=vizink,inner sep=0pt] at (7.9,-0.66) {Dunja};
        \node[below,text=vizmuted,inner sep=2pt] at (7.933,-2.7) {1};
        \node[below,text=vizmuted,inner sep=2pt] at (8.695,-2.7) {7};
        \node[below,text=vizmuted,inner sep=2pt] at (10.09,-2.7) {18};
        \fill[vizband,rounded corners=1.5pt] (10.874,-2.68) rectangle (11.508,-0.78);
        \draw[vizaxis,line width=0.35pt] (10.493,-2.6) -- (12.777,-2.6);
        \draw[vizgrid,line width=0.35pt] (10.493,-1.82) -- (12.777,-1.82);
        \draw[vizgrid,line width=0.35pt] (10.493,-1.04) -- (12.777,-1.04);
        \draw[vizorange,line width=1.1pt] (11.128,-2.6) -- (11.128,-0.78);
        \draw[vizmuted,line width=0.6pt] (10.557,-1.976) -- (10.684,-1.82) -- (10.81,-1.82) -- (10.937,-1.976) -- (11.064,-1.82) -- (11.191,-1.352) -- (11.318,-1.196) -- (11.445,-1.196) -- (11.572,-1.352) -- (11.698,-1.04) -- (11.825,-1.04) -- (11.952,-1.04) -- (12.079,-1.352) -- (12.206,-1.04) -- (12.333,-1.196) -- (12.46,-1.196) -- (12.586,-1.04) -- (12.713,-1.352);
        \foreach \p in {(10.557,-1.976),(10.684,-1.82),(10.81,-1.82),(10.937,-1.976),(11.064,-1.82)} { \filldraw[fill=white,draw=vizsec,line width=0.7pt] \p circle (1.7pt); }
        \foreach \p in {(11.191,-1.352),(11.318,-1.196),(11.445,-1.196),(11.572,-1.352),(11.698,-1.04),(11.825,-1.04),(11.952,-1.04),(12.079,-1.352),(12.206,-1.04),(12.333,-1.196),(12.46,-1.196),(12.586,-1.04),(12.713,-1.352)} { \filldraw[fill=vizblue,draw=white,line width=0.5pt] \p circle (2.1pt); }
        \node[anchor=base west,text=vizink,inner sep=0pt] at (10.523,-0.66) {Egor};
        \node[below,text=vizmuted,inner sep=2pt] at (10.557,-2.7) {1};
        \node[below,text=vizmuted,inner sep=2pt] at (11.191,-2.7) {6};
        \node[below,text=vizmuted,inner sep=2pt] at (12.713,-2.7) {18};
        \fill[vizband,rounded corners=1.5pt] (13.497,-2.68) rectangle (14.131,-0.78);
        \draw[vizaxis,line width=0.35pt] (13.117,-2.6) -- (15.4,-2.6);
        \draw[vizgrid,line width=0.35pt] (13.117,-1.82) -- (15.4,-1.82);
        \draw[vizgrid,line width=0.35pt] (13.117,-1.04) -- (15.4,-1.04);
        \draw[vizorange,line width=1.1pt] (14.131,-2.6) -- (14.131,-0.78);
        \draw[vizmuted,line width=0.6pt] (13.18,-1.976) -- (13.307,-1.82) -- (13.434,-1.82) -- (13.561,-2.132) -- (13.687,-1.976) -- (13.814,-1.976) -- (13.941,-2.132) -- (14.068,-1.82) -- (14.195,-1.04) -- (14.322,-1.352) -- (14.449,-1.196) -- (14.575,-1.04) -- (14.702,-1.196) -- (14.829,-1.04) -- (14.956,-1.196) -- (15.083,-1.508) -- (15.21,-1.196) -- (15.337,-1.352);
        \foreach \p in {(13.18,-1.976),(13.307,-1.82),(13.434,-1.82),(13.561,-2.132),(13.687,-1.976),(13.814,-1.976),(13.941,-2.132),(14.068,-1.82)} { \filldraw[fill=white,draw=vizsec,line width=0.7pt] \p circle (1.7pt); }
        \foreach \p in {(14.195,-1.04),(14.322,-1.352),(14.449,-1.196),(14.575,-1.04),(14.702,-1.196),(14.829,-1.04),(14.956,-1.196),(15.083,-1.508),(15.21,-1.196),(15.337,-1.352)} { \filldraw[fill=vizblue,draw=white,line width=0.5pt] \p circle (2.1pt); }
        \node[anchor=base west,text=vizink,inner sep=0pt] at (13.147,-0.66) {Fabian};
        \node[below,text=vizmuted,inner sep=2pt] at (13.18,-2.7) {1};
        \node[below,text=vizmuted,inner sep=2pt] at (14.195,-2.7) {9};
        \node[below,text=vizmuted,inner sep=2pt] at (15.337,-2.7) {18};
        \node[right,text=vizink,inner sep=0pt,font=\sffamily\footnotesize] at (0,-3.62) {\textbf{Study B}\enspace\textcolor{vizsec}{Correct answers out of 10}};
        \fill[vizband,rounded corners=1.5pt] (0.562,-6.16) rectangle (1.872,-4.26);
        \draw[vizaxis,line width=0.35pt] (0,-6.08) -- (2.808,-6.08);
        \node[left,text=vizmuted,inner sep=1.5pt] at (-0.02,-6.08) {0};
        \draw[vizgrid,line width=0.35pt] (0,-5.3) -- (2.808,-5.3);
        \node[left,text=vizmuted,inner sep=1.5pt] at (-0.02,-5.3) {5};
        \draw[vizgrid,line width=0.35pt] (0,-4.52) -- (2.808,-4.52);
        \node[left,text=vizmuted,inner sep=1.5pt] at (-0.02,-4.52) {10};
        \draw[vizorange,line width=1.1pt] (0.749,-6.08) -- (0.749,-4.26);
        \draw[vizmuted,line width=0.6pt] (0.094,-5.612) -- (0.281,-6.08) -- (0.468,-5.924) -- (0.655,-5.924) -- (0.842,-5.456) -- (1.03,-5.612) -- (1.217,-4.52) -- (1.404,-4.832) -- (1.591,-4.832) -- (1.778,-5.144) -- (1.966,-4.676) -- (2.153,-4.676) -- (2.34,-4.832) -- (2.527,-5.144) -- (2.714,-4.52);
        \foreach \p in {(0.094,-5.612),(0.281,-6.08),(0.468,-5.924),(0.655,-5.924)} { \filldraw[fill=white,draw=vizsec,line width=0.7pt] \p circle (1.7pt); }
        \foreach \p in {(0.842,-5.456),(1.03,-5.612),(1.217,-4.52),(1.404,-4.832),(1.591,-4.832),(1.778,-5.144),(1.966,-4.676),(2.153,-4.676),(2.34,-4.832),(2.527,-5.144),(2.714,-4.52)} { \filldraw[fill=vizblue,draw=white,line width=0.5pt] \p circle (2.1pt); }
        \node[anchor=base west,text=vizink,inner sep=0pt] at (0.03,-4.14) {Asena};
        \node[below,text=vizmuted,inner sep=2pt] at (0.094,-6.18) {1};
        \node[below,text=vizmuted,inner sep=2pt] at (0.842,-6.18) {5};
        \node[below,text=vizmuted,inner sep=2pt] at (2.714,-6.18) {15};
        \fill[vizband,rounded corners=1.5pt] (3.71,-6.16) rectangle (5.02,-4.26);
        \draw[vizaxis,line width=0.35pt] (3.148,-6.08) -- (5.956,-6.08);
        \draw[vizgrid,line width=0.35pt] (3.148,-5.3) -- (5.956,-5.3);
        \draw[vizgrid,line width=0.35pt] (3.148,-4.52) -- (5.956,-4.52);
        \draw[vizorange,line width=1.1pt] (4.646,-6.08) -- (4.646,-4.26);
        \draw[vizmuted,line width=0.6pt] (3.242,-5.768) -- (3.429,-5.612) -- (3.616,-6.08) -- (3.803,-5.924) -- (3.99,-5.612) -- (4.178,-5.612) -- (4.365,-5.924) -- (4.552,-5.768) -- (4.739,-4.52) -- (4.926,-4.52) -- (5.114,-4.52) -- (5.301,-4.52) -- (5.488,-4.52) -- (5.675,-4.676) -- (5.862,-4.52);
        \foreach \p in {(3.242,-5.768),(3.429,-5.612),(3.616,-6.08),(3.803,-5.924),(3.99,-5.612),(4.178,-5.612),(4.365,-5.924),(4.552,-5.768)} { \filldraw[fill=white,draw=vizsec,line width=0.7pt] \p circle (1.7pt); }
        \foreach \p in {(4.739,-4.52),(4.926,-4.52),(5.114,-4.52),(5.301,-4.52),(5.488,-4.52),(5.675,-4.676),(5.862,-4.52)} { \filldraw[fill=vizblue,draw=white,line width=0.5pt] \p circle (2.1pt); }
        \node[anchor=base west,text=vizink,inner sep=0pt] at (3.178,-4.14) {Eman};
        \node[below,text=vizmuted,inner sep=2pt] at (3.242,-6.18) {1};
        \node[below,text=vizmuted,inner sep=2pt] at (4.739,-6.18) {9};
        \node[below,text=vizmuted,inner sep=2pt] at (5.862,-6.18) {15};
        \fill[vizband,rounded corners=1.5pt] (6.858,-6.16) rectangle (8.168,-4.26);
        \draw[vizaxis,line width=0.35pt] (6.296,-6.08) -- (9.104,-6.08);
        \draw[vizgrid,line width=0.35pt] (6.296,-5.3) -- (9.104,-5.3);
        \draw[vizgrid,line width=0.35pt] (6.296,-4.52) -- (9.104,-4.52);
        \draw[vizorange,line width=1.1pt] (7.045,-6.08) -- (7.045,-4.26);
        \draw[vizmuted,line width=0.6pt] (6.39,-5.768) -- (6.577,-6.08) -- (6.764,-5.3) -- (6.951,-5.924) -- (7.138,-5.612) -- (7.326,-5.144) -- (7.513,-4.832) -- (7.7,-4.676) -- (7.887,-4.988) -- (8.074,-5.144) -- (8.262,-4.52) -- (8.449,-4.832) -- (8.636,-4.832) -- (8.823,-5.3) -- (9.01,-4.52);
        \foreach \p in {(6.39,-5.768),(6.577,-6.08),(6.764,-5.3),(6.951,-5.924)} { \filldraw[fill=white,draw=vizsec,line width=0.7pt] \p circle (1.7pt); }
        \foreach \p in {(7.138,-5.612),(7.326,-5.144),(7.513,-4.832),(7.7,-4.676),(7.887,-4.988),(8.074,-5.144),(8.262,-4.52),(8.449,-4.832),(8.636,-4.832),(8.823,-5.3),(9.01,-4.52)} { \filldraw[fill=vizblue,draw=white,line width=0.5pt] \p circle (2.1pt); }
        \node[anchor=base west,text=vizink,inner sep=0pt] at (6.326,-4.14) {Julia};
        \node[below,text=vizmuted,inner sep=2pt] at (6.39,-6.18) {1};
        \node[below,text=vizmuted,inner sep=2pt] at (7.138,-6.18) {5};
        \node[below,text=vizmuted,inner sep=2pt] at (9.01,-6.18) {15};
        \fill[vizband,rounded corners=1.5pt] (10.006,-6.16) rectangle (11.316,-4.26);
        \draw[vizaxis,line width=0.35pt] (9.444,-6.08) -- (12.252,-6.08);
        \draw[vizgrid,line width=0.35pt] (9.444,-5.3) -- (12.252,-5.3);
        \draw[vizgrid,line width=0.35pt] (9.444,-4.52) -- (12.252,-4.52);
        \draw[vizorange,line width=1.1pt] (10.193,-6.08) -- (10.193,-4.26);
        \draw[vizmuted,line width=0.6pt] (9.538,-5.456) -- (9.725,-6.08) -- (9.912,-5.3) -- (10.099,-5.768) -- (10.286,-5.3) -- (10.474,-4.988) -- (10.661,-4.676) -- (10.848,-4.832) -- (11.035,-4.52) -- (11.222,-4.832) -- (11.41,-4.832) -- (11.597,-4.52) -- (11.784,-4.676) -- (11.971,-4.832) -- (12.158,-4.52);
        \foreach \p in {(9.538,-5.456),(9.725,-6.08),(9.912,-5.3),(10.099,-5.768)} { \filldraw[fill=white,draw=vizsec,line width=0.7pt] \p circle (1.7pt); }
        \foreach \p in {(10.286,-5.3),(10.474,-4.988),(10.661,-4.676),(10.848,-4.832),(11.035,-4.52),(11.222,-4.832),(11.41,-4.832),(11.597,-4.52),(11.784,-4.676),(11.971,-4.832),(12.158,-4.52)} { \filldraw[fill=vizblue,draw=white,line width=0.5pt] \p circle (2.1pt); }
        \node[anchor=base west,text=vizink,inner sep=0pt] at (9.474,-4.14) {Leon};
        \node[below,text=vizmuted,inner sep=2pt] at (9.538,-6.18) {1};
        \node[below,text=vizmuted,inner sep=2pt] at (10.286,-6.18) {5};
        \node[below,text=vizmuted,inner sep=2pt] at (12.158,-6.18) {15};
        \fill[vizband,rounded corners=1.5pt] (13.154,-6.16) rectangle (14.464,-4.26);
        \draw[vizaxis,line width=0.35pt] (12.592,-6.08) -- (15.4,-6.08);
        \draw[vizgrid,line width=0.35pt] (12.592,-5.3) -- (15.4,-5.3);
        \draw[vizgrid,line width=0.35pt] (12.592,-4.52) -- (15.4,-4.52);
        \draw[vizorange,line width=1.1pt] (13.528,-6.08) -- (13.528,-4.26);
        \draw[vizmuted,line width=0.6pt] (12.686,-5.768) -- (12.873,-5.924) -- (13.06,-5.612) -- (13.247,-6.08) -- (13.434,-5.768) -- (13.622,-4.52) -- (13.809,-4.832) -- (13.996,-4.52) -- (14.183,-4.52) -- (14.37,-4.52) -- (14.558,-4.52) -- (14.745,-4.52) -- (14.932,-4.52) -- (15.119,-4.52) -- (15.306,-4.52);
        \foreach \p in {(12.686,-5.768),(12.873,-5.924),(13.06,-5.612),(13.247,-6.08),(13.434,-5.768)} { \filldraw[fill=white,draw=vizsec,line width=0.7pt] \p circle (1.7pt); }
        \foreach \p in {(13.622,-4.52),(13.809,-4.832),(13.996,-4.52),(14.183,-4.52),(14.37,-4.52),(14.558,-4.52),(14.745,-4.52),(14.932,-4.52),(15.119,-4.52),(15.306,-4.52)} { \filldraw[fill=vizblue,draw=white,line width=0.5pt] \p circle (2.1pt); }
        \node[anchor=base west,text=vizink,inner sep=0pt] at (12.622,-4.14) {Marvin};
        \node[below,text=vizmuted,inner sep=2pt] at (12.686,-6.18) {1};
        \node[below,text=vizmuted,inner sep=2pt] at (13.622,-6.18) {6};
        \node[below,text=vizmuted,inner sep=2pt] at (15.306,-6.18) {15};
        \node[right,text=vizink,inner sep=0pt,font=\sffamily\footnotesize] at (0,-7.1) {\textbf{Study C}\enspace\textcolor{vizsec}{Words written}};
        \fill[vizband,rounded corners=1.5pt] (1.498,-9.64) rectangle (2.696,-7.74);
        \draw[vizaxis,line width=0.35pt] (0,-9.56) -- (3.595,-9.56);
        \node[left,text=vizmuted,inner sep=1.5pt] at (-0.02,-9.56) {0};
        \draw[vizgrid,line width=0.35pt] (0,-8.78) -- (3.595,-8.78);
        \node[left,text=vizmuted,inner sep=1.5pt] at (-0.02,-8.78) {100};
        \draw[vizgrid,line width=0.35pt] (0,-8) -- (3.595,-8);
        \node[left,text=vizmuted,inner sep=1.5pt] at (-0.02,-8) {200};
        \draw[vizorange,line width=1.1pt] (1.498,-9.56) -- (1.498,-7.74);
        \draw[vizmuted,line width=0.6pt] (0.15,-9.123) -- (0.449,-9.201) -- (0.749,-9.225) -- (1.049,-9.217) -- (1.348,-9.17) -- (1.648,-8.944) -- (1.947,-8.593) -- (2.247,-8.616) -- (2.546,-8.554) -- (2.846,-8.499) -- (3.146,-8.086) -- (3.445,-8.437);
        \foreach \p in {(0.15,-9.123),(0.449,-9.201),(0.749,-9.225),(1.049,-9.217),(1.348,-9.17)} { \filldraw[fill=white,draw=vizsec,line width=0.7pt] \p circle (1.7pt); }
        \foreach \p in {(1.648,-8.944),(1.947,-8.593),(2.247,-8.616),(2.546,-8.554),(2.846,-8.499),(3.146,-8.086),(3.445,-8.437)} { \filldraw[fill=vizblue,draw=white,line width=0.5pt] \p circle (2.1pt); }
        \node[anchor=base west,text=vizink,inner sep=0pt] at (0.03,-7.62) {Adal};
        \node[below,text=vizmuted,inner sep=2pt] at (0.15,-9.66) {1};
        \node[below,text=vizmuted,inner sep=2pt] at (1.648,-9.66) {6};
        \node[below,text=vizmuted,inner sep=2pt] at (3.445,-9.66) {12};
        \fill[vizband,rounded corners=1.5pt] (5.433,-9.64) rectangle (6.631,-7.74);
        \draw[vizaxis,line width=0.35pt] (3.935,-9.56) -- (7.53,-9.56);
        \draw[vizgrid,line width=0.35pt] (3.935,-8.78) -- (7.53,-8.78);
        \draw[vizgrid,line width=0.35pt] (3.935,-8) -- (7.53,-8);
        \draw[vizorange,line width=1.1pt] (6.032,-9.56) -- (6.032,-7.74);
        \draw[vizmuted,line width=0.6pt] (4.085,-9.24) -- (4.384,-9.248) -- (4.684,-9.326) -- (4.984,-9.24) -- (5.283,-9.279) -- (5.583,-9.256) -- (5.882,-9.318) -- (6.182,-8.398) -- (6.481,-8.086) -- (6.781,-8.125) -- (7.081,-8.265) -- (7.38,-8.078);
        \foreach \p in {(4.085,-9.24),(4.384,-9.248),(4.684,-9.326),(4.984,-9.24),(5.283,-9.279),(5.583,-9.256),(5.882,-9.318)} { \filldraw[fill=white,draw=vizsec,line width=0.7pt] \p circle (1.7pt); }
        \foreach \p in {(6.182,-8.398),(6.481,-8.086),(6.781,-8.125),(7.081,-8.265),(7.38,-8.078)} { \filldraw[fill=vizblue,draw=white,line width=0.5pt] \p circle (2.1pt); }
        \node[anchor=base west,text=vizink,inner sep=0pt] at (3.965,-7.62) {Berta};
        \node[below,text=vizmuted,inner sep=2pt] at (4.085,-9.66) {1};
        \node[below,text=vizmuted,inner sep=2pt] at (6.182,-9.66) {8};
        \node[below,text=vizmuted,inner sep=2pt] at (7.38,-9.66) {12};
        \fill[vizband,rounded corners=1.5pt] (9.368,-9.64) rectangle (10.566,-7.74);
        \draw[vizaxis,line width=0.35pt] (7.87,-9.56) -- (11.465,-9.56);
        \draw[vizgrid,line width=0.35pt] (7.87,-8.78) -- (11.465,-8.78);
        \draw[vizgrid,line width=0.35pt] (7.87,-8) -- (11.465,-8);
        \draw[vizorange,line width=1.1pt] (9.368,-9.56) -- (9.368,-7.74);
        \draw[vizmuted,line width=0.6pt] (8.02,-9.248) -- (8.319,-9.209) -- (8.619,-9.201) -- (8.919,-9.295) -- (9.218,-9.256) -- (9.518,-8.92) -- (9.817,-8.842) -- (10.117,-8.796) -- (10.416,-8.725) -- (10.716,-8.702) -- (11.016,-8.749) -- (11.315,-8.686);
        \foreach \p in {(8.02,-9.248),(8.319,-9.209),(8.619,-9.201),(8.919,-9.295),(9.218,-9.256)} { \filldraw[fill=white,draw=vizsec,line width=0.7pt] \p circle (1.7pt); }
        \foreach \p in {(9.518,-8.92),(9.817,-8.842),(10.117,-8.796),(10.416,-8.725),(10.716,-8.702),(11.016,-8.749),(11.315,-8.686)} { \filldraw[fill=vizblue,draw=white,line width=0.5pt] \p circle (2.1pt); }
        \node[anchor=base west,text=vizink,inner sep=0pt] at (7.9,-7.62) {Channa};
        \node[below,text=vizmuted,inner sep=2pt] at (8.02,-9.66) {1};
        \node[below,text=vizmuted,inner sep=2pt] at (9.518,-9.66) {6};
        \node[below,text=vizmuted,inner sep=2pt] at (11.315,-9.66) {12};
        \fill[vizband,rounded corners=1.5pt] (13.303,-9.64) rectangle (14.501,-7.74);
        \draw[vizaxis,line width=0.35pt] (11.805,-9.56) -- (15.4,-9.56);
        \draw[vizgrid,line width=0.35pt] (11.805,-8.78) -- (15.4,-8.78);
        \draw[vizgrid,line width=0.35pt] (11.805,-8) -- (15.4,-8);
        \draw[vizorange,line width=1.1pt] (13.902,-9.56) -- (13.902,-7.74);
        \draw[vizmuted,line width=0.6pt] (11.955,-9.31) -- (12.254,-9.279) -- (12.554,-9.256) -- (12.854,-9.303) -- (13.153,-9.349) -- (13.453,-9.342) -- (13.752,-9.248) -- (14.052,-8.819) -- (14.351,-8.858) -- (14.651,-8.928) -- (14.951,-8.889) -- (15.25,-8.811);
        \foreach \p in {(11.955,-9.31),(12.254,-9.279),(12.554,-9.256),(12.854,-9.303),(13.153,-9.349),(13.453,-9.342),(13.752,-9.248)} { \filldraw[fill=white,draw=vizsec,line width=0.7pt] \p circle (1.7pt); }
        \foreach \p in {(14.052,-8.819),(14.351,-8.858),(14.651,-8.928),(14.951,-8.889),(15.25,-8.811)} { \filldraw[fill=vizblue,draw=white,line width=0.5pt] \p circle (2.1pt); }
        \node[anchor=base west,text=vizink,inner sep=0pt] at (11.835,-7.62) {Daniel};
        \node[below,text=vizmuted,inner sep=2pt] at (11.955,-9.66) {1};
        \node[below,text=vizmuted,inner sep=2pt] at (14.052,-9.66) {8};
        \node[below,text=vizmuted,inner sep=2pt] at (15.25,-9.66) {12};
        \node[text=vizink,inner sep=2pt,font=\sffamily\footnotesize] at (7.7,-10.22) {Session};
    \end{tikzpicture}
    \par\smallskip
    \parbox{0.95\textwidth}{\small\textit{Note.} Each panel shows one student. The shaded band covers the window sessions, and the vertical line marks the onset. The horizontal axis labels the first session, the onset session, and the last session.}
\end{figure}
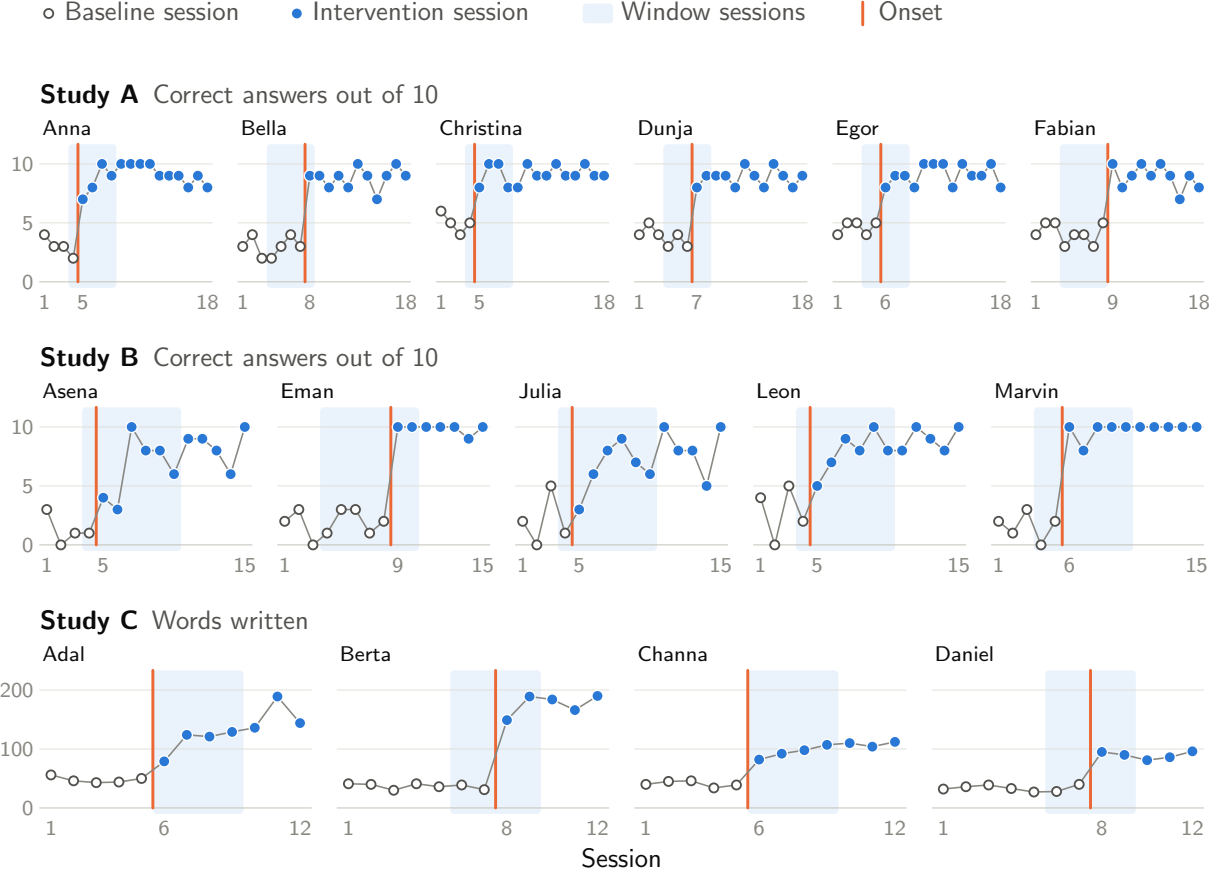

\subsection{Effects by Study and Pooled Effects}

This analysis estimates the effect of each study on its original scale and the pooled standardized effect of the 15 students. Within a study, we use $c_i=1$ and equal weights. For the pooled effect, we use equal weights and two scales. The first is the pre-window mean of the student. The second is the standard deviation of the pre-window outcomes pooled over the students of a study and divided by the factor $\iota_i$ of Proposition \ref{prop:scale}. We compute $\htheta^{\rm b}$, $\htheta^{\rm c}$, the pooled-anchor estimator $\htheta^{\rm f}$, and the mean of the phase contrasts with the variance estimator in Definition \ref{def:estimators}(e). We also test the sharp null hypothesis of no effect on any student at any session by a one-sided randomization test with the pooled-anchor estimate as the test statistic (Section \ref{sec:test} of the Supplemental Material). The test is exact within each study and uses 100,000 draws of the onsets for the 15 students together. Table \ref{tab:app} reports the estimates, standard errors, $p$-values, and 95\% intervals for the pooled-anchor estimator.

The estimators $\htheta^{\rm c}$ and $\htheta^{\rm f}$ and the phase contrast give similar estimates in Study A, where the effect is immediate, and different estimates in Studies B and C, where it builds up over sessions. The bias of $\htheta^{\rm c}$ in Proposition \ref{prop:biasc} of the Supplemental Material predicts its smaller estimates in Studies B and C. In Study A, the outcomes rose close to their intervention level at the onset session. The estimators $\htheta^{\rm c}$ and $\htheta^{\rm f}$ and the phase contrast yield 4.83, 4.75, and 5.04 correct answers, respectively, and the standard error of $\htheta^{\rm b}$ is about three times larger than theirs. In Studies B and C, the outcomes of several students kept rising for several sessions after the onset. The estimator $\htheta^{\rm c}$ then gives the smallest estimates, 4.55 correct answers and 44.79 words, and the pooled-anchor estimates are 5.08 and 62.00, respectively. The estimates of $\htheta^{\rm b}$, 7.00 and 102.21, exceed the pooled-anchor estimates by less than two of their standard errors. The phase contrasts are larger, 6.64 and 84.65. They place more weight on late sessions, where outcomes are highest, and they include any drift. The sharp null hypothesis is rejected in every study ($p=0.00006, 0.002, 0.013$, respectively) and for the 15 students together ($p=0.0002$). The 95\% intervals of the pooled-anchor estimates by study rest on four to six blocks. They are wide, and the interval of Study C contains 0. The pooled-anchor estimate of the pooled effect is 1.86 pre-window means, a 186\% increase over the pre-window mean, with a 95\% interval of $[1.06, 2.66]$. In standard deviation units, it is 6.98, with interval $[3.13, 10.84]$. The pooled standard deviations before the intervention are small in these studies, 0.75 and 1.95 correct answers and 4.86 words, so effects in standard deviation units are large. With the two pairs of Study C as two blocks, the estimates are the same, and the intervals become $[0.95, 2.77]$ and $[2.05, 11.92]$. The $p$-value of Study C is then 0.08, and the $p$-value for the 15 students is 0.0008 (Section \ref{sec:pairs} of the Supplemental Material).

\begin{table}[t]
    \caption{Estimates (Standard Errors) for Three Randomized Multiple-Baseline Studies}
    \label{tab:app}
    \centering
    \footnotesize
    \setlength{\tabcolsep}{3pt}
    \begin{tabular}{@{}lcccccc@{}}
        \toprule
        & & & & Phase & 95\% interval & \\
        & $\htheta^{\rm b}$ & $\htheta^{\rm c}$ & $\htheta^{\rm f}$ & contrast & for $\htheta^{\rm f}$ & $p$-value \\
        \midrule
        \multicolumn{7}{l}{\textit{Original scale, by study}} \\
        A: correct answers & 3.90 (1.04) & 4.83 (0.35) & 4.75 (0.31) & 5.04 (0.30) & $[3.89, 5.61]$ & 0.00006 \\
        B: correct answers & 7.00 (1.22) & 4.55 (1.33) & 5.08 (1.16) & 6.64 (0.61) & $[1.39, 8.78]$ & 0.002 \\
        C: words written & 102.21 (22.21) & 44.79 (35.01) & 62.00 (21.93) & 84.65 (19.06) & $[-32.34, 156.33]$ & 0.013 \\
        \multicolumn{7}{l}{\textit{Standardized by scale, 15 students}} \\
        Pre-window mean & 2.25 (0.43) & 1.64 (0.36) & 1.86 (0.37) & 2.26 (0.35) & $[1.06, 2.66]$ & \\
        Pre-window SD & 8.41 (2.21) & 5.49 (1.80) & 6.98 (1.78) & 8.00 (1.69) & $[3.13, 10.84]$ & 0.0002 \\
        \bottomrule
    \end{tabular}
    \par\smallskip
    \parbox{0.95\textwidth}{\small\textit{Note.} The $p$-values are for the sharp null hypothesis of no effect, from the one-sided randomization test with the pooled-anchor estimate as the statistic. The pre-window SD is pooled within study and divided by $\iota_i$. The $p$-value for the 15 students together uses the pre-window SD scale. The intervals by study rest on six, five, and four blocks.}
\end{table}

\section{Discussion}

In this paper, we developed design-based estimators of effect sizes for multiple-baseline experiments with randomized onsets, and we showed how to pool them across cases and studies. We found that weighted comparisons within the window of admissible onsets are unbiased for a case-level effect. They can be standardized without bias, and their pooled value has a variance estimator and an asymptotically valid confidence interval that use no model for the series. The pooled-anchor estimator allows the effect to build up over sessions and is asymptotically unbiased. The phase contrast is generally biased when the untreated outcome changes over sessions, and randomizing the onset does not remove this bias. When planning a study, we recommend many admissible onsets and several sessions before and after them. In Table \ref{tab:plan}, both choices reduce the standard deviations of $\htheta^{\rm c}$ and $\htheta^{\rm f}$.

The proposed estimators have two limitations. First, they have larger variances than the phase contrast and the multilevel models, while their biases stay near zero across different baselines. This suggests that they are robust to baseline patterns at the cost of losing precision. Second, the confidence intervals are asymptotic in the number of blocks, and the pooled-anchor estimator is asymptotically unbiased. With 6 and 12 cases under a common shift, its intervals could under-cover (92.7\% and 92.5\%, shown in Section \ref{sec:sim6} of the Supplemental Material). The intervals by study in Table \ref{tab:app} rest on four to six blocks. In practice, a synthesis pools the cases of several studies and therefore has more blocks than any one of them. The 15 students of the three studies together gave an interval of $[1.06, 2.66]$ in pre-window means.

We leave two extensions for future work. In many multiple-baseline studies, the onset is chosen after inspecting the baseline, a practice known as response-guided experimentation \citep{FerronJooLevin2017}. The onset probabilities then depend on the outcomes, and the estimators of this paper would need to be extended to such designs. The estimators of this paper use each case as its own comparison. Cases of the same study that are still in baseline at a session carry further information about the untreated outcome at that session \citep{Ferron2014}, and a design-based estimator that combines the two comparisons could have a smaller variance.

\bibliographystyle{jasa3}
\bibliography{mb_meta_refs}

\clearpage
\setcounter{section}{0}
\setcounter{equation}{0}
\setcounter{table}{0}
\setcounter{figure}{0}
\renewcommand{\thesection}{S\arabic{section}}
\renewcommand{\theequation}{S\arabic{equation}}
\renewcommand{\thetable}{S\arabic{table}}
\renewcommand{\thefigure}{S\arabic{figure}}
\renewcommand{\theHsection}{S\arabic{section}}
\renewcommand{\theHequation}{S\arabic{equation}}
\renewcommand{\theHtable}{S\arabic{table}}
\renewcommand{\theHfigure}{S\arabic{figure}}
\addtocontents{toc}{\protect\setcounter{tocdepth}{2}}

\begin{center}
    {\Large\bfseries Supplemental Material}
\end{center}

\noindent References to Definition, Proposition, Theorem, Corollary, and Assumption numbers and to unprefixed equation, table, and section numbers are to the main text, Sections 1 to 9. Equations, lemmas, propositions, tables and sections numbered with the prefix S are in this supplement.

{\singlespacing\tableofcontents}
\clearpage

\section{Notation and Window Sums}
\label{sec:notation}

Throughout, $W_{it}=W_{it}(S_i)$, and $\bbE$, $\Var$ and $\Cov$ refer to the distribution of the onsets unless stated otherwise, i.e., Section \ref{sec:planning}. For a case $i$ and an onset $s\in\cS_i$, write $\overline Y_{i\cA}(s)=a_i^{-1}\sum_{t\in\cA_i}Y_{it}(s)$ and $\overline Y_{i\cB}(s)=n_i^{-1}\sum_{t\in\cB_i}Y_{it}(s)$, so that $\overline Y_{i\cA}=\overline Y_{i\cA}(S_i)$ and $\overline Y_{i\cB}=\overline Y_{i\cB}(S_i)$. Every $t\in\cA_i$ has $t\le a_i<s$ for every $s\in\cS_i$. Under Assumption \ref{ass:anticipation},
\begin{equation}
    Y_{it}(s)=Y^{(0)}_{it}\quad\text{for all }t\in\cA_i\text{ and }s\in\cS_i,\quad\text{so that}\quad\overline Y_{i\cA}(s)=a_i^{-1}\sum_{t\in\cA_i}Y^{(0)}_{it}\quad\text{for all }s\in\cS_i.
    \label{eq:prefixed}
\end{equation}
The scale $c_i$ is a function of the outcomes $\{Y_{i't}:t\in\cA_{i'},\ 1\le i'\le N\}$, so by \eqref{eq:prefixed}, it takes the same value under every vector of onsets. The centered and scaled potential outcome is
\begin{equation*}
    \widetilde Y_{it}(s)=\frac{Y_{it}(s)-\overline Y_{i\cA}(s)}{c_i},\qquad\widetilde Y_{it}=\widetilde Y_{it}(S_i),
\end{equation*}
and the window sums in Definition \ref{def:estimators}(f) are
\begin{equation}
    \begin{aligned}
        H_{i1}&=\frac1{m_i}\sum_{t\in\cW_i}\frac{W_{it}\widetilde Y_{it}}{\pi_{it}},&\qquad X_{i1}&=\frac1{m_i}\sum_{t\in\cW_i}\frac{W_{it}}{\pi_{it}},\\
        H_{i0}&=\frac1{m_i}\sum_{t\in\cW_i}\frac{(1-W_{it})\widetilde Y_{it}}{1-\pi_{it}},&\qquad X_{i0}&=\frac1{m_i}\sum_{t\in\cW_i}\frac{1-W_{it}}{1-\pi_{it}}.
    \end{aligned}
    \label{eq:sums}
\end{equation}
The sums $X_{i1}$ and $X_{i0}$ are functions of the onset $S_i$ alone. We write $X_{iw}(s)$ for the value of $X_{iw}$ when $S_i=s$. Under Assumption \ref{ass:anticipation}, $\overline Y_{i\cA}$ and $c_i$ are constants by \eqref{eq:prefixed}, so $H_{i1}$, $H_{i0}$ and the estimators of Definition \ref{def:estimators}(a)--(d) for case $i$ are functions of $S_i$ alone. The block averages of the window sums are
\begin{equation}
    H_{(g)w}=\frac1{\Omega_g}\sum_{i\in\cC_g}\omega_iH_{iw},\qquad X_{(g)w}=\frac1{\Omega_g}\sum_{i\in\cC_g}\omega_iX_{iw},
    \label{eq:blocksums}
\end{equation}
and the pooled sums are
\begin{equation*}
    \overline H_w=\sum_{i=1}^N\omega_iH_{iw}=\sum_{g=1}^G\Omega_gH_{(g)w},\qquad\overline X_w=\sum_{i=1}^N\omega_iX_{iw}=\sum_{g=1}^G\Omega_gX_{(g)w}.
\end{equation*}

\begin{lemma}
    \label{lem:weights}
    Let $t\in\cW_i$ and let $\kappa$ be a real function on $\cS_i$. Then,
    \begin{equation}
        \bbE\left\{\frac{W_{it}\,\kappa(S_i)}{\pi_{it}}\right\}=\frac1{\pi_{it}}\sum_{s\in\cS_i,\,s\le t}p_i(s)\kappa(s),\quad\bbE\left\{\frac{(1-W_{it})\,\kappa(S_i)}{1-\pi_{it}}\right\}=\frac1{1-\pi_{it}}\sum_{s\in\cS_i,\,s>t}p_i(s)\kappa(s).
        \label{eq:w1}
    \end{equation}
    In particular, $\bbE(W_{it}/\pi_{it})=\bbE\{(1-W_{it})/(1-\pi_{it})\}=1$, and $\bbE(X_{i1})=\bbE(X_{i0})=1$. Under Assumption \ref{ass:anticipation},
    \begin{equation}
        \bbE\left(\frac{W_{it}Y_{it}}{\pi_{it}}\right)=\overline Y^{(1)}_{it},\qquad\bbE\left\{\frac{(1-W_{it})Y_{it}}{1-\pi_{it}}\right\}=Y^{(0)}_{it}.
        \label{eq:w2}
    \end{equation}
\end{lemma}

\begin{proof}
    Since $W_{it}=\bbone(S_i\le t)$ and $\bbP(S_i=s)=p_i(s)$,
    \begin{equation*}
        \bbE\{W_{it}\kappa(S_i)\}=\sum_{s\in\cS_i}p_i(s)\bbone(s\le t)\kappa(s),\qquad\bbE\{(1-W_{it})\kappa(S_i)\}=\sum_{s\in\cS_i}p_i(s)\bbone(s>t)\kappa(s),
    \end{equation*}
    which is \eqref{eq:w1}. Taking $\kappa\equiv1$ and using $\pi_{it}=\sum_{s\in\cS_i,\,s\le t}p_i(s)$ gives the two expectations equal to 1, and averaging them over $t\in\cW_i$ gives $\bbE(X_{i1})=\bbE(X_{i0})=1$. For \eqref{eq:w2}, $Y_{it}=Y_{it}(S_i)$. Taking $\kappa(s)=Y_{it}(s)$ in the first part of \eqref{eq:w1} gives $\overline Y^{(1)}_{it}$ by its definition in \eqref{eq:celleffect}. Under Assumption \ref{ass:anticipation}, $Y_{it}(s)=Y^{(0)}_{it}$ for every $s>t$, so the second part of \eqref{eq:w1} with $\kappa(s)=Y_{it}(s)$ equals
    \begin{equation*}
        \frac{Y^{(0)}_{it}}{1-\pi_{it}}\sum_{s\in\cS_i,\,s>t}p_i(s)=Y^{(0)}_{it}.
    \end{equation*}
\end{proof}

\section{Proofs of Theorem \ref{thm:unbiased} and Corollary \ref{cor:standardized}}
\label{sec:unbiased}

The three estimators $\htau^{\rm a}_i$, $\htau^{\rm b}_i$ and $\htau^{\rm c}_i$ of Definition \ref{def:estimators} are members of one family. For offsets $\ell^1_{it}$ and $\ell^0_{it}$, $t\in\cW_i$, define
\begin{equation}
    \htau_i(\ell)=\frac1{m_i}\sum_{t\in\cW_i}\left\{\frac{W_{it}(Y_{it}-\ell^1_{it})}{\pi_{it}}-\frac{(1-W_{it})(Y_{it}-\ell^0_{it})}{1-\pi_{it}}\right\}+\frac1{m_i}\sum_{t\in\cW_i}\left(\ell^1_{it}-\ell^0_{it}\right).
    \label{eq:family}
\end{equation}

\begin{lemma}
    \label{lem:offsets}
    Let Assumption \ref{ass:anticipation} hold, and let the offsets $\ell^1_{it}$ and $\ell^0_{it}$, $t\in\cW_i$, take the same values under every onset in $\cS_i$. Then, $\bbE\{\htau_i(\ell)\}=\tau_i$.
\end{lemma}

\begin{proof}
    The offsets are constants. By \eqref{eq:w2} and $\bbE(W_{it}/\pi_{it})=\bbE\{(1-W_{it})/(1-\pi_{it})\}=1$ from Lemma \ref{lem:weights}, for each $t\in\cW_i$,
    \begin{equation*}
        \bbE\left\{\frac{W_{it}(Y_{it}-\ell^1_{it})}{\pi_{it}}-\frac{(1-W_{it})(Y_{it}-\ell^0_{it})}{1-\pi_{it}}\right\}=\overline Y^{(1)}_{it}-\ell^1_{it}-Y^{(0)}_{it}+\ell^0_{it}=\tau_{it}-\left(\ell^1_{it}-\ell^0_{it}\right).
    \end{equation*}
    Averaging over $t\in\cW_i$ and adding the last term of \eqref{eq:family} gives $\bbE\{\htau_i(\ell)\}=m_i^{-1}\sum_{t\in\cW_i}\tau_{it}=\tau_i$.
\end{proof}

\begin{proof}[Proof of Theorem \ref{thm:unbiased}]
    The estimator $\htau^{\rm a}_i$ is \eqref{eq:family} with $\ell^1_{it}=\ell^0_{it}=0$. The estimator $\htau^{\rm b}_i$ is \eqref{eq:family} with $\ell^1_{it}=\ell^0_{it}=\overline Y_{i\cA}$, which takes the same value under every onset by \eqref{eq:prefixed}. Lemma \ref{lem:offsets} gives $\bbE(\htau^{\rm a}_i)=\bbE(\htau^{\rm b}_i)=\tau_i$ under Assumption \ref{ass:anticipation}. The estimator $\htau^{\rm c}_i$ is \eqref{eq:family} with $\ell^0_{it}=\overline Y_{i\cA}$ and $\ell^1_{it}=\overline Y_{i\cB}$. Every $t\in\cB_i$ has $t\ge b_i\ge s$ for every $s\in\cS_i$. Under Assumption \ref{ass:duration},
    \begin{equation}
        Y_{it}(s)=Y^{(1)}_{it}\quad\text{for all }t\in\cB_i\text{ and }s\in\cS_i,\quad\text{so that}\quad\overline Y_{i\cB}(s)=n_i^{-1}\sum_{t\in\cB_i}Y^{(1)}_{it}\quad\text{for all }s\in\cS_i.
        \label{eq:postfixed}
    \end{equation}
    By \eqref{eq:prefixed} and \eqref{eq:postfixed}, both offsets take the same values under every onset, and Lemma \ref{lem:offsets} gives $\bbE(\htau^{\rm c}_i)=\tau_i$ under Assumptions \ref{ass:anticipation} and \ref{ass:duration}.
\end{proof}

\begin{proof}[Proof of Corollary \ref{cor:standardized}]
    By \eqref{eq:prefixed}, $c_i$ is a constant. Hence $\bbE(\htheta_i)=\bbE(\htau_i)/c_i=\tau_i/c_i=\theta_i$ by Theorem \ref{thm:unbiased}, and $\bbE(\htheta)=\sum_{i=1}^N\omega_i\bbE(\htheta_i)=\sum_{i=1}^N\omega_i\theta_i=\theta$ because the weights are constants.
\end{proof}

The scaled estimators are functions of the window sums. By Definition \ref{def:estimators}(d), $\htheta^{\rm b}_i=\htau^{\rm b}_i/c_i$ and $\htheta^{\rm c}_i=\htau^{\rm c}_i/c_i$. Dividing \eqref{eq:family} by $c_i$ and using $\widetilde Y_{it}=(Y_{it}-\overline Y_{i\cA})/c_i$ in \eqref{eq:sums},
\begin{equation}
    \htheta^{\rm b}_i=H_{i1}-H_{i0},\qquad\htheta^{\rm c}_i=\widetilde Y_{i\cB}\left(1-X_{i1}\right)+H_{i1}-H_{i0},\qquad\widetilde Y_{i\cB}=\frac{\overline Y_{i\cB}-\overline Y_{i\cA}}{c_i}.
    \label{eq:bc}
\end{equation}
The first identity holds because the two offsets of $\htau^{\rm b}_i$ are equal. For the second, $Y_{it}-\overline Y_{i\cB}=(Y_{it}-\overline Y_{i\cA})-(\overline Y_{i\cB}-\overline Y_{i\cA})$ on the treated side, which contributes $-(\overline Y_{i\cB}-\overline Y_{i\cA})X_{i1}$ to the first sum of \eqref{eq:family}, and the second sum of \eqref{eq:family} equals $\overline Y_{i\cB}-\overline Y_{i\cA}$.

\section{Proof of Proposition \ref{prop:contrast}}
\label{sec:contrast}

Let Assumptions \ref{ass:anticipation} and \ref{ass:duration} hold. For $t\in\cW_i$, Assumption \ref{ass:duration} gives $\overline Y^{(1)}_{it}=Y^{(1)}_{it}$ in \eqref{eq:celleffect}, so $\tau_{it}=Y^{(1)}_{it}-Y^{(0)}_{it}$ for every $t\ge a_i+1$. For every $s\in\cS_i$ and every session $t$,
\begin{equation}
    Y_{it}(s)=Y^{(0)}_{it}\quad\text{if }t<s,\qquad Y_{it}(s)=Y^{(0)}_{it}+\tau_{it}\quad\text{if }t\ge s.
    \label{eq:pc1}
\end{equation}
The first part is Assumption \ref{ass:anticipation}. The second part is Assumption \ref{ass:duration}, because $t\ge s$ implies $t\ge a_i+1$ and $Y_{it}(s)=Y^{(1)}_{it}$. By \eqref{eq:pc1} and \eqref{eq:contrast},
\begin{equation*}
    \gamma_i(s)=\frac1{T_i-s+1}\sum_{t=s}^{T_i}\tau_{it}+\left\{\frac1{T_i-s+1}\sum_{t=s}^{T_i}Y^{(0)}_{it}-\frac1{s-1}\sum_{t=1}^{s-1}Y^{(0)}_{it}\right\}.
\end{equation*}
Taking the expectation over $S_i$, using the definition of $\zeta_i$ in \eqref{eq:drift}, and exchanging the order of summation in the first term,
\begin{equation*}
    \bbE\{\gamma_i(S_i)\}=\sum_{s\in\cS_i}\frac{p_i(s)}{T_i-s+1}\sum_{t=s}^{T_i}\tau_{it}+\zeta_i=\sum_{t=a_i+1}^{T_i}\tau_{it}\sum_{s\in\cS_i,\,s\le t}\frac{p_i(s)}{T_i-s+1}+\zeta_i=\sum_{t=a_i+1}^{T_i}\lambda_{it}\tau_{it}+\zeta_i,
\end{equation*}
because every $s\in\cS_i$ is at least $a_i+1$. The weights are nonnegative, and
\begin{equation*}
    \sum_{t=a_i+1}^{T_i}\lambda_{it}=\sum_{s\in\cS_i}\frac{p_i(s)}{T_i-s+1}\sum_{t=s}^{T_i}1=\sum_{s\in\cS_i}p_i(s)=1.
\end{equation*}
If $Y^{(0)}_{it}=\phi_i+\beta_it$ for every $t$, then for every $s\in\cS_i$,
\begin{equation*}
    \frac1{T_i-s+1}\sum_{t=s}^{T_i}\left(\phi_i+\beta_it\right)-\frac1{s-1}\sum_{t=1}^{s-1}\left(\phi_i+\beta_it\right)=\beta_i\left(\frac{s+T_i}2-\frac s2\right)=\frac{\beta_iT_i}2,
\end{equation*}
so $\zeta_i=\sum_{s\in\cS_i}p_i(s)\beta_iT_i/2=\beta_iT_i/2$. \qed

\section{Proof of Proposition \ref{prop:variance}}
\label{sec:variance}

For numbers or random variables $U_1,\dots,U_G$ write $\overline U=\sum_{g=1}^G\Omega_gU_g$ and
\begin{equation}
    \cV(U)=\frac1{1+\rho}\sum_{g=1}^G\frac{\Omega_g^2\left(U_g-\overline U\right)^2}{1-2\Omega_g},\qquad\rho=\sum_{g=1}^G\frac{\Omega_g^2}{1-2\Omega_g},
    \label{eq:calV}
\end{equation}
so that $\hV=\cV(\hvartheta)$ with $\hvartheta=(\hvartheta_1,\dots,\hvartheta_G)$.

\begin{lemma}
    \label{lem:between}
    Let $\max_g\Omega_g<1/2$, and let $U_1,\dots,U_G$ be independent random variables with finite variances. Then, $\Var(\overline U)=\sum_{g=1}^G\Omega_g^2\Var(U_g)$ and
    \begin{equation}
        \bbE\{\cV(U)\}=\Var(\overline U)+\frac1{1+\rho}\sum_{g=1}^G\frac{\Omega_g^2\left\{\bbE(U_g)-\bbE(\overline U)\right\}^2}{1-2\Omega_g}.
        \label{eq:between}
    \end{equation}
\end{lemma}

\begin{proof}
    The first claim holds because the $U_g$ are independent. Fix $g$. Since $\sum_{g'=1}^G\Omega_{g'}=1$,
    \begin{equation*}
        U_g-\overline U=(1-\Omega_g)U_g-\sum_{g'\ne g}\Omega_{g'}U_{g'},
    \end{equation*}
    and by independence
    \begin{equation*}
        \Var(U_g-\overline U)=(1-\Omega_g)^2\Var(U_g)+\sum_{g'\ne g}\Omega_{g'}^2\Var(U_{g'})=(1-2\Omega_g)\Var(U_g)+\Var(\overline U).
    \end{equation*}
    Hence
    \begin{equation*}
        \bbE\left(U_g-\overline U\right)^2=(1-2\Omega_g)\Var(U_g)+\Var(\overline U)+\left\{\bbE(U_g)-\bbE(\overline U)\right\}^2.
    \end{equation*}
    Multiplying by $\Omega_g^2/(1-2\Omega_g)$ and summing over $g$,
    \begin{equation*}
        \sum_{g=1}^G\frac{\Omega_g^2\,\bbE\left(U_g-\overline U\right)^2}{1-2\Omega_g}=\Var(\overline U)+\rho\Var(\overline U)+\sum_{g=1}^G\frac{\Omega_g^2\left\{\bbE(U_g)-\bbE(\overline U)\right\}^2}{1-2\Omega_g}.
    \end{equation*}
    Dividing by $1+\rho$ gives \eqref{eq:between}.
\end{proof}

\begin{proof}[Proof of Proposition \ref{prop:variance}]
    By Section \ref{sec:notation}, each $\htheta_i$ is a function of $S_i$ alone, so the block estimate $\hvartheta_g$ is a function of the onsets of block $g$. Onsets of different blocks are independent, and thus are $\hvartheta_1,\dots,\hvartheta_G$. They have finite variances because each takes finitely many values. By Corollary \ref{cor:standardized}, $\bbE(\hvartheta_g)=\Omega_g^{-1}\sum_{i\in\cC_g}\omega_i\theta_i=\vartheta_g$ and $\bbE(\htheta)=\theta$. Lemma \ref{lem:between} with $U_g=\hvartheta_g$ and $\overline U=\htheta$ gives \eqref{eq:variance}. With $\Omega_g=1/G$ for every $g$, $\rho=G\cdot G^{-2}/(1-2/G)=1/(G-2)$ and
    \begin{equation*}
        \frac{\Omega_g^2}{(1+\rho)(1-2\Omega_g)}=\frac{G^{-2}}{\{(G-1)/(G-2)\}\{(G-2)/G\}}=\frac1{G(G-1)},
    \end{equation*}
    which gives the stated form of $\hV$ for equal block weights.
\end{proof}

\section{Proof of Theorem \ref{thm:coverage}}
\label{sec:coverage}

\subsection{Three Facts}

Let $\Phi$ be the standard normal distribution function and $z_\alpha=\Phi^{-1}(1-\alpha/2)$, so that $z_\alpha>0$.

\begin{lemma}[\citealp{Feller1971}, Section XVI.5, Theorem 2]
    \label{lem:be}
    Let $\xi_1,\dots,\xi_G$ be independent random variables with $\bbE(\xi_g)=0$ and $\bbE|\xi_g|^3<\infty$, and let $V_\xi=\sum_{g=1}^G\bbE(\xi_g^2)>0$. Then, for every real $x$,
    \begin{equation*}
        \left|\bbP\left(\sum_{g=1}^G\xi_g\le xV_\xi^{1/2}\right)-\Phi(x)\right|\le\frac{6\sum_{g=1}^G\bbE|\xi_g|^3}{V_\xi^{3/2}}.
    \end{equation*}
\end{lemma}

\begin{lemma}
    \label{lem:tquantile}
    For every integer $\nu\ge1$, $q_\nu\ge z_\alpha$.
\end{lemma}

\begin{proof}
    Let $Z$ be standard normal, and let $\chi^2$ be independent of $Z$ with a chi-squared distribution with $\nu$ degrees of freedom. Put $J=(\chi^2/\nu)^{1/2}$, so that $Z/J$ has the $t$ distribution with $\nu$ degrees of freedom and $\bbE(J)\le\{\bbE(J^2)\}^{1/2}=1$. Let $x>0$. The function $\Phi$ is concave and increasing on $[0,\infty)$, and $xJ\ge0$. Conditioning on $J$ and using Jensen's inequality,
    \begin{equation*}
        \bbP\left(\frac ZJ\le x\right)=\bbE\{\Phi(xJ)\}\le\Phi\{x\,\bbE(J)\}\le\Phi(x).
    \end{equation*}
    Take $x=z_\alpha$. Then, $\bbP(Z/J\le z_\alpha)\le1-\alpha/2$, so the $1-\alpha/2$ quantile $q_\nu$ of $Z/J$ is at least $z_\alpha$.
\end{proof}

The next lemma collects what the proof needs about a weighted sum of independent bounded terms. Let $C_\omega$ and $\lv$ be the constants of Assumptions \ref{ass:weights} and \ref{ass:nondegenerate}. Since $1=\sum_{g=1}^G\Omega_g\le C_\omega$ under Assumption \ref{ass:weights}, $C_\omega\ge1$.

\begin{lemma}
    \label{lem:linear}
    Let Assumption \ref{ass:weights} hold and $G\ge4C_\omega$. Let $U_1,\dots,U_G$ be independent with $|U_g|\le M_\ast$ for a constant $M_\ast$, and suppose that $G\Var(\overline U)\ge\lv>0$. Then,
    \begin{align}
        &\Var(\overline U)\le\frac{C_\omega^2M_\ast^2}G,\label{eq:l1} \displaybreak[0]\\
        &\bbP\left\{\left|\overline U-\bbE(\overline U)\right|\le x\Var(\overline U)^{1/2}\right\}\ge2\Phi(x)-1-2\Lambda_G\quad\text{for every }x\ge0,\quad\Lambda_G=\frac{48C_\omega^3M_\ast^3}{\lv^{3/2}G^{1/2}},\label{eq:l2}\displaybreak[0]\\
        &\bbP\left\{\cV(U)<(1-2\delta)\Var(\overline U)\right\}\le\eta_G(\delta)\quad\text{for every }\delta\in(0,1/2),\label{eq:l3}
    \end{align}
    where
    \begin{equation*}
        \eta_G(\delta)=\frac{64C_\omega^4M_\ast^4}{\delta^2\lv^2G}+\frac{16C_\omega^3M_\ast^2}{\delta\lv G^{1/2}}+\frac{4C_\omega^4M_\ast^2}{\delta\lv G}.
    \end{equation*}
\end{lemma}

\begin{proof}
    Write $u_g=\bbE(U_g)$ and $u=\bbE(\overline U)=\sum_{g=1}^G\Omega_gu_g$, so that $|u_g|\le M_\ast$ and $|u|\le M_\ast$. By Lemma \ref{lem:between}, $\Var(U_g)\le\bbE(U_g^2)\le M_\ast^2$ and Assumption \ref{ass:weights},
    \begin{equation*}
        \Var(\overline U)=\sum_{g=1}^G\Omega_g^2\Var(U_g)\le M_\ast^2\sum_{g=1}^G\frac{C_\omega^2}{G^2}=\frac{C_\omega^2M_\ast^2}G,
    \end{equation*}
    which is \eqref{eq:l1}. For \eqref{eq:l2}, let $\xi_g=\Omega_g(U_g-u_g)$. Then, $\overline U-u=\sum_{g=1}^G\xi_g$, $V_\xi=\Var(\overline U)\ge\lv/G$, and $|\xi_g|\le2C_\omega M_\ast/G$, so
    \begin{equation}
        \frac{6\sum_{g=1}^G\bbE|\xi_g|^3}{V_\xi^{3/2}}\le\frac{6\,G\,(2C_\omega M_\ast/G)^3}{(\lv/G)^{3/2}}=\frac{48C_\omega^3M_\ast^3}{\lv^{3/2}G^{1/2}}=\Lambda_G.
        \label{eq:be2}
    \end{equation}
    Let $F$ be the distribution function of $(\overline U-u)/V_\xi^{1/2}$. By Lemma \ref{lem:be} and \eqref{eq:be2}, $F(x)\ge\Phi(x)-\Lambda_G$, and $F(y)\le\Phi(y)+\Lambda_G$ for every $y<-x$. Letting $y$ increase to $-x$ and using the continuity of $\Phi$,
    \begin{equation*}
        \bbP\left\{\left|\overline U-u\right|\le xV_\xi^{1/2}\right\}=F(x)-\lim_{y\uparrow-x}F(y)\ge\Phi(x)-\Lambda_G-\Phi(-x)-\Lambda_G=2\Phi(x)-1-2\Lambda_G.
    \end{equation*}

    For \eqref{eq:l3}, let $\widetilde\Omega_g=\Omega_g^2/\{(1+\rho)(1-2\Omega_g)\}$, so that $\cV(U)=\sum_{g=1}^G\widetilde\Omega_g(U_g-\overline U)^2$. Since $G\ge4C_\omega$, Assumption \ref{ass:weights} gives $\Omega_g\le1/4$ and $1-2\Omega_g\ge1/2$. With $\rho\ge0$,
    \begin{equation}
        \widetilde\Omega_g\le2\Omega_g^2\le\frac{2C_\omega^2}{G^2},\qquad\sum_{g=1}^G\widetilde\Omega_g\le\frac{2C_\omega^2}G.
        \label{eq:nu}
    \end{equation}
    Expanding $(U_g-\overline U)^2=\{(U_g-u)-(\overline U-u)\}^2$, we have
    \begin{equation}
        \cV(U)=Q_1+Q_2,\quad Q_1=\sum_{g=1}^G\widetilde\Omega_g(U_g-u)^2,\quad Q_2=-2(\overline U-u)\sum_{g=1}^G\widetilde\Omega_g(U_g-u)+(\overline U-u)^2\sum_{g=1}^G\widetilde\Omega_g.
        \label{eq:QR}
    \end{equation}
    The terms of $Q_1$ are independent and $|U_g-u|\le2M_\ast$. By \eqref{eq:nu},
    \begin{equation*}
        \Var(Q_1)=\sum_{g=1}^G\widetilde\Omega_g^2\Var\{(U_g-u)^2\}\le G\left(\frac{2C_\omega^2}{G^2}\right)^2(2M_\ast)^4=\frac{64C_\omega^4M_\ast^4}{G^3},
    \end{equation*}
    and Chebyshev's inequality gives
    \begin{equation}
        \bbP\left\{|Q_1-\bbE(Q_1)|>\frac{\delta\lv}G\right\}\le\frac{G^2\Var(Q_1)}{\delta^2\lv^2}\le\frac{64C_\omega^4M_\ast^4}{\delta^2\lv^2G}.
        \label{eq:Qtail}
    \end{equation}
    By \eqref{eq:nu}, $|\sum_{g=1}^G\widetilde\Omega_g(U_g-u)|\le4C_\omega^2M_\ast/G$. With $\bbE|\overline U-u|\le\Var(\overline U)^{1/2}$, $\bbE(\overline U-u)^2=\Var(\overline U)$ and \eqref{eq:l1},
    \begin{equation*}
        \bbE|Q_2|\le\frac{8C_\omega^2M_\ast}G\cdot\frac{C_\omega M_\ast}{G^{1/2}}+\frac{2C_\omega^2}G\cdot\frac{C_\omega^2M_\ast^2}G=\frac{8C_\omega^3M_\ast^2}{G^{3/2}}+\frac{2C_\omega^4M_\ast^2}{G^2}.
    \end{equation*}
    Markov's inequality and $\bbE|Q_2-\bbE(Q_2)|\le2\bbE|Q_2|$ give
    \begin{equation}
        \bbP\left\{|Q_2-\bbE(Q_2)|>\frac{\delta\lv}G\right\}\le\frac{2G\,\bbE|Q_2|}{\delta\lv}\le\frac{16C_\omega^3M_\ast^2}{\delta\lv G^{1/2}}+\frac{4C_\omega^4M_\ast^2}{\delta\lv G}.
        \label{eq:Rtail}
    \end{equation}
    On the complement of the events in \eqref{eq:Qtail} and \eqref{eq:Rtail}, \eqref{eq:QR}, Lemma \ref{lem:between} and $\lv/G\le\Var(\overline U)$ give
    \begin{equation*}
        \cV(U)\ge\bbE\{\cV(U)\}-\frac{2\delta\lv}G\ge\Var(\overline U)-2\delta\Var(\overline U).
    \end{equation*}
    The sum of the bounds in \eqref{eq:Qtail} and \eqref{eq:Rtail} is $\eta_G(\delta)$, which proves \eqref{eq:l3}.
\end{proof}

\subsection{Bounds on the Case-Level Quantities}

Let $M$ and $\lpi$ be the constants of Assumptions \ref{ass:bounded} and \ref{ass:overlap}, and put $M_\ast=\max(4M,1)/\lpi$. Every case has a window session $t$, and $\lpi\le\pi_{it}\le1-\lpi$ gives $\lpi\le1/2$, so $M_\ast\ge2$.

\begin{lemma}
    \label{lem:bounds}
    Let Assumptions \ref{ass:anticipation}, \ref{ass:bounded} and \ref{ass:overlap} hold. For every case $i$, every $w\in\{0,1\}$ and every vector of onsets in the support of the design,
    \begin{equation}
        |H_{iw}|\le\frac M\lpi,\qquad0\le X_{iw}\le\frac1\lpi,\qquad\left|\htheta^{\rm b}_i\right|\le M_\ast,\qquad\left|\htheta^{\rm c}_i\right|\le M_\ast.
        \label{eq:bd1}
    \end{equation}
    Let $\mu_w=\bbE(\overline H_w)$, as in Definition \ref{def:estimators}(f). Then, $|\mu_w|\le M$, $\mu_1-\mu_0=\theta$, and the quantity
    \begin{equation}
        e_i=\left(H_{i1}-\mu_1X_{i1}\right)-\left(H_{i0}-\mu_0X_{i0}\right)
        \label{eq:ei}
    \end{equation}
    satisfies $|e_i|\le M_\ast$ and $\bbE(e_i)=\theta_i-\theta$. The bounds on $H_{iw}$ and $X_{iw}$ hold for $H_{(g)w}$ and $X_{(g)w}$, and the bounds on $\htheta^{\rm b}_i$, $\htheta^{\rm c}_i$ and $e_i$ hold for their block averages.
\end{lemma}

\begin{proof}
    By \eqref{eq:prefixed} and Assumption \ref{ass:bounded}, $|\widetilde Y_{it}(s)|\le M$ for all $i$, $t$ and $s\in\cS_i$. Each term of the four averages in \eqref{eq:sums} has a weight in $[0,1/\lpi]$ by Assumption \ref{ass:overlap}, which gives the first two bounds of \eqref{eq:bd1}. By \eqref{eq:bc}, $|\htheta^{\rm b}_i|\le2M/\lpi\le M_\ast$. Since $|\widetilde Y_{i\cB}|\le M$ as an average of the $\widetilde Y_{it}$, $t\in\cB_i$, and $1\le1/\lpi$,
    \begin{equation*}
        \left|\htheta^{\rm c}_i\right|\le M\left(1+\frac1\lpi\right)+\frac{2M}\lpi\le\frac{4M}\lpi\le M_\ast.
    \end{equation*}
    The outcome $\widetilde Y_{it}$ is $Y_{it}$ shifted and scaled by constants. By \eqref{eq:w2} and $\bbE(W_{it}/\pi_{it})=\bbE\{(1-W_{it})/(1-\pi_{it})\}=1$,
    \begin{equation}
        \bbE(H_{i1})=\frac1{m_i}\sum_{t\in\cW_i}\frac{\overline Y^{(1)}_{it}-\overline Y_{i\cA}}{c_i},\qquad\bbE(H_{i0})=\frac1{m_i}\sum_{t\in\cW_i}\frac{Y^{(0)}_{it}-\overline Y_{i\cA}}{c_i},
        \label{eq:EH}
    \end{equation}
    so $\bbE(H_{i1})-\bbE(H_{i0})=\tau_i/c_i=\theta_i$. For $t\in\cW_i$, $t<b_i$ and $b_i\in\cS_i$, so Assumption \ref{ass:anticipation} gives $Y^{(0)}_{it}=Y_{it}(b_i)$, and Assumption \ref{ass:bounded} gives $|Y^{(0)}_{it}-\overline Y_{i\cA}|\le Mc_i$. By \eqref{eq:celleffect}, $\overline Y^{(1)}_{it}$ is a weighted average of the $Y_{it}(s)$, so both expectations in \eqref{eq:EH} are averages of numbers bounded by $M$ in absolute value. Hence $|\mu_w|\le M$, and $\mu_1-\mu_0=\sum_{i=1}^N\omega_i\theta_i=\theta$. By the bounds already shown, $|e_i|\le2(M/\lpi+M/\lpi)=4M/\lpi\le M_\ast$. By \eqref{eq:EH} and $\bbE(X_{iw})=1$ from Lemma \ref{lem:weights}, $\bbE(e_i)=\bbE(H_{i1})-\mu_1-\bbE(H_{i0})+\mu_0=\theta_i-\theta$. A block average is a weighted average with nonnegative weights that sum to 1, so it satisfies the same bounds.
\end{proof}

\subsection{The Pooled Estimator}

\begin{propositionS}
    \label{prop:rate}
    Let Assumptions \ref{ass:anticipation} and \ref{ass:weights}--\ref{ass:nondegenerate} hold, let $G\ge4C_\omega$, and let $\Lambda_G$ and $\eta_G$ be as in Lemma \ref{lem:linear} with $M_\ast=\max(4M,1)/\lpi$. For the pooled estimator $\htheta=\htheta^{\rm b}$ in Definition \ref{def:estimators}(e), and for every $\delta\in(0,1/2)$,
    \begin{align}
        &\sup_{x\in\bbR}\left|\bbP\left\{\htheta-\theta\le x\Var(\htheta)^{1/2}\right\}-\Phi(x)\right|\le\Lambda_G,\label{eq:r1}\\
        &\bbP\left\{\left|\htheta-\theta\right|\le q_{G-1}\hV^{1/2}\right\}\ge2\Phi\left\{z_\alpha(1-2\delta)^{1/2}\right\}-1-2\Lambda_G-\eta_G(\delta).\label{eq:r2}
    \end{align}
    Under Assumption \ref{ass:duration} in addition, the same holds for $\htheta=\htheta^{\rm c}$.
\end{propositionS}

\begin{proof}
    Let $U_g=\hvartheta_g$. The $U_g$ are independent as in the proof of Proposition \ref{prop:variance}, $|U_g|\le M_\ast$ by Lemma \ref{lem:bounds}, $\overline U=\htheta$, and $\bbE(\overline U)=\theta$ by Corollary \ref{cor:standardized}. Assumption \ref{ass:nondegenerate} is $G\Var(\overline U)\ge\lv$. The bound \eqref{eq:r1} is Lemma \ref{lem:be} with $\xi_g=\Omega_g\{U_g-\bbE(U_g)\}$ and \eqref{eq:be2}. For \eqref{eq:r2}, recall that $\hV=\cV(U)$. On the event $\{\cV(U)\ge(1-2\delta)\Var(\overline U)\}$,
    \begin{equation*}
        \left|\htheta-\theta\right|\le z_\alpha(1-2\delta)^{1/2}\Var(\htheta)^{1/2}\quad\text{implies}\quad\left|\htheta-\theta\right|\le z_\alpha\hV^{1/2}\le q_{G-1}\hV^{1/2},
    \end{equation*}
    where the last inequality is Lemma \ref{lem:tquantile}. By \eqref{eq:l2} with $x=z_\alpha(1-2\delta)^{1/2}$ and \eqref{eq:l3},
    \begin{equation*}
        \bbP\left\{\left|\htheta-\theta\right|\le q_{G-1}\hV^{1/2}\right\}\ge2\Phi\left\{z_\alpha(1-2\delta)^{1/2}\right\}-1-2\Lambda_G-\eta_G(\delta).
    \end{equation*}
    Under Assumption \ref{ass:duration}, $\bbE(\htheta^{\rm c})=\theta$ by Corollary \ref{cor:standardized}, and Lemma \ref{lem:bounds} bounds $\htheta^{\rm c}_i$ by the same $M_\ast$, so the same displays hold.
\end{proof}

\subsection{The Pooled-Anchor Estimator}

By Definition \ref{def:estimators}(f), $\hmu_w=\overline H_w/\overline X_w$ when $\overline X_w>0$ and $\hmu_w=0$ otherwise. With $\mu_w$ and $e_i$ as in Lemma \ref{lem:bounds}, let
\begin{equation}
    L_w=\overline H_w-\mu_w\overline X_w,\qquad r_g=\frac1{\Omega_g}\sum_{i\in\cC_g}\omega_ie_i,\qquad L=L_1-L_0=\sum_{i=1}^N\omega_ie_i=\sum_{g=1}^G\Omega_gr_g.
    \label{eq:Lz}
\end{equation}
This $L$ is the linear term of the main text, and Assumption \ref{ass:nondegenerate} for the pooled-anchor estimator is $G\Var(L)\ge\lv$.

\begin{propositionS}
    \label{prop:anchor}
    Let Assumptions \ref{ass:anticipation} and \ref{ass:weights}--\ref{ass:nondegenerate} hold, let $G\ge4C_\omega$, and let $M_\ast$, $\Lambda_G$ and $\eta_G$ be as in Proposition \ref{prop:rate}. Then,
    \begin{equation}
        \left|\bbE(\htheta^{\rm f})-\theta\right|\le\frac{4C_\omega^2M_\ast^2(1+4M_\ast)}G.
        \label{eq:a0}
    \end{equation}
    Let $\delta\in(0,1/3)$ with $z_\alpha(1-3\delta)^{1/2}>\delta$, and $\varsigma_\delta=\min[1,\ \delta\lv/\{2C_\omega^2(4M_\ast+1)\}]$. Then,
    \begin{equation}
        \begin{aligned}
            \bbP\left\{\left|\htheta^{\rm f}-\theta\right|\le q_{G-2}\left(\hV^{\rm f}\right)^{1/2}\right\}&\ge2\Phi\left\{z_\alpha(1-3\delta)^{1/2}-\delta\right\}-1-2\Lambda_G-\eta_G(\delta)\\
            &\quad-\frac{8C_\omega^2M_\ast^2}G-\frac{4C_\omega^2M_\ast^2}{\delta\lv^{1/2}G^{1/2}}-\frac{6C_\omega M_\ast^2}{\varsigma_\delta G^{1/2}}.
        \end{aligned}
        \label{eq:a1}
    \end{equation}
\end{propositionS}

\begin{proof}
    Since $C_\omega\ge1$, $G\ge4$. By Lemma \ref{lem:bounds}, $|H_{(g)w}-\mu_wX_{(g)w}|\le M/\lpi+M/\lpi\le M_\ast$, $X_{(g)w}\in[0,1/\lpi]\subseteq[0,M_\ast]$ and $|r_g|\le M_\ast$. By Section \ref{sec:notation}, each of these block quantities is a function of the onsets of block $g$, so each family is independent across blocks. By \eqref{eq:blocksums}, $L_w=\sum_{g=1}^G\Omega_g(H_{(g)w}-\mu_wX_{(g)w})$. Since $\bbE(\overline X_w)=1$ by Lemma \ref{lem:weights}, $\bbE(L_w)=\mu_w-\mu_w=0$. The argument for \eqref{eq:l1}, which uses only independence, the bound $M_\ast$ and Assumption \ref{ass:weights}, gives
    \begin{equation}
        \bbE(L_w^2)=\Var(L_w)\le\frac{C_\omega^2M_\ast^2}G,\qquad\bbE\left(\overline X_w-1\right)^2=\Var(\overline X_w)\le\frac{C_\omega^2M_\ast^2}G,\qquad w\in\{0,1\}.
        \label{eq:a2}
    \end{equation}
    Let $\cE_0=\{|\overline X_1-1|\le1/2\}\cap\{|\overline X_0-1|\le1/2\}$. By Chebyshev's inequality and \eqref{eq:a2},
    \begin{equation}
        \bbP(\cE_0^c)\le4\,\bbE\left(\overline X_1-1\right)^2+4\,\bbE\left(\overline X_0-1\right)^2\le\frac{8C_\omega^2M_\ast^2}G.
        \label{eq:a3}
    \end{equation}
    On $\cE_0$, $\overline X_w\ge1/2$ and $\hmu_w-\mu_w=(\overline H_w-\mu_w\overline X_w)/\overline X_w=L_w/\overline X_w$. By Lemma \ref{lem:bounds}, $\theta=\mu_1-\mu_0$, so on $\cE_0$,
    \begin{equation}
        \begin{gathered}
            \htheta^{\rm f}-\theta=\frac{L_1}{\overline X_1}-\frac{L_0}{\overline X_0}=L+\Delta,\qquad\Delta=\frac{L_1(1-\overline X_1)}{\overline X_1}-\frac{L_0(1-\overline X_0)}{\overline X_0},\\
            |\Delta|\le2|L_1||\overline X_1-1|+2|L_0||\overline X_0-1|.
        \end{gathered}
        \label{eq:a4}
    \end{equation}
    Off $\cE_0$, put $\Delta=0$. By the Cauchy--Schwarz inequality and \eqref{eq:a2},
    \begin{equation}
        \bbE\left(|\Delta|\bbone_{\cE_0}\right)\le2\left\{\bbE(L_1^2)\bbE(\overline X_1-1)^2\right\}^{1/2}+2\left\{\bbE(L_0^2)\bbE(\overline X_0-1)^2\right\}^{1/2}\le\frac{4C_\omega^2M_\ast^2}G.
        \label{eq:a5}
    \end{equation}
    When $\overline X_w>0$, $\hmu_w$ is a weighted average of numbers $\widetilde Y_{it}$ with nonnegative weights, and $\hmu_w=0$ otherwise, so $|\hmu_w|\le M$ for every vector of onsets and $|\htheta^{\rm f}|\le2M$. With $|\theta|=|\mu_1-\mu_0|\le2M$ and $|L|\le M_\ast$, $|\htheta^{\rm f}-\theta-L|\le4M+M_\ast\le2M_\ast$. Since $\bbE(L)=0$, \eqref{eq:a3}, \eqref{eq:a4} and \eqref{eq:a5} give
    \begin{equation*}
        \left|\bbE(\htheta^{\rm f})-\theta\right|=\left|\bbE\left(\htheta^{\rm f}-\theta-L\right)\right|\le\bbE\left(|\Delta|\bbone_{\cE_0}\right)+2M_\ast\bbP(\cE_0^c)\le\frac{4C_\omega^2M_\ast^2}G+\frac{16C_\omega^2M_\ast^3}G,
    \end{equation*}
    which is \eqref{eq:a0}.

    For \eqref{eq:a1}, recall the residuals $\widehat e_i$ and their block averages $\widehat r_g$ from Definition \ref{def:estimators}(f). On $\cE_0$, the definition of $\hmu_w$ gives $\sum_{g=1}^G\Omega_g\widehat r_g=(\overline H_1-\hmu_1\overline X_1)-(\overline H_0-\hmu_0\overline X_0)=0$. With the notation \eqref{eq:calV} and $\widehat r=(\widehat r_1,\dots,\widehat r_G)$, and since $G\ge4$,
    \begin{equation}
        \hV^{\rm f}=\frac{G-1}{G-2}\,\cV(\widehat r)\ge\cV(\widehat r)\quad\text{on }\cE_0.
        \label{eq:a6}
    \end{equation}
    Let $r=(r_1,\dots,r_G)$, so that $\overline r=L$ by \eqref{eq:Lz}. On $\cE_0$, by \eqref{eq:ei},
    \begin{equation*}
        \widehat r_g-(r_g-L)=-(\hmu_1-\mu_1)X_{(g)1}+(\hmu_0-\mu_0)X_{(g)0}+L.
    \end{equation*}
    Call this difference $d_g$. On $\cE_0$, $|\hmu_w-\mu_w|\le2|L_w|$. With $0\le X_{(g)w}\le M_\ast$, $|L|\le|L_1|+|L_0|$ and $M_\ast\ge1$,
    \begin{equation*}
        |d_g|\le2M_\ast\left(|L_1|+|L_0|\right)+|L|\le D,\qquad D=3M_\ast\left(|L_1|+|L_0|\right).
    \end{equation*}
    With $\widetilde\Omega_g$ as in the proof of Lemma \ref{lem:linear}, $|r_g-L|\le2M_\ast$ and \eqref{eq:nu}, on $\cE_0\cap\{D\le\varsigma_\delta\}$,
    \begin{equation}
        \begin{split}\left|\cV(\widehat r)-\cV(r)\right|&=\left|\sum_{g=1}^G\widetilde\Omega_g\left\{2(r_g-L)d_g+d_g^2\right\}\right|\le\frac{2C_\omega^2}G\left(4M_\ast D+D^2\right)\\&\le\frac{2C_\omega^2(4M_\ast+1)\varsigma_\delta}G\le\frac{\delta\lv}G,\end{split}
        \label{eq:a7}
    \end{equation}
    where $D^2\le D$ because $D\le\varsigma_\delta\le1$. By Markov's inequality, $\bbE|L_w|\le\{\bbE(L_w^2)\}^{1/2}$ and \eqref{eq:a2},
    \begin{equation}
        \bbP(D>\varsigma_\delta)\le\frac{3M_\ast\left\{\bbE|L_1|+\bbE|L_0|\right\}}{\varsigma_\delta}\le\frac{6C_\omega M_\ast^2}{\varsigma_\delta G^{1/2}}.
        \label{eq:a8}
    \end{equation}
    The $r_g$ are independent and bounded by $M_\ast$, and $G\Var(L)\ge\lv$ by Assumption \ref{ass:nondegenerate}. Lemma \ref{lem:linear} with $U_g=r_g$, for which $\overline U=L$ and $\bbE(\overline U)=0$, gives
    \begin{equation}
        \begin{gathered}\bbP\left\{\cV(r)<(1-2\delta)\Var(L)\right\}\le\eta_G(\delta),\\\bbP\left\{|L|\le x\Var(L)^{1/2}\right\}\ge2\Phi(x)-1-2\Lambda_G\quad\text{for }x\ge0.\end{gathered}
        \label{eq:a9}
    \end{equation}
    By Markov's inequality, \eqref{eq:a5} and $\Var(L)\ge\lv/G$,
    \begin{equation}
        \bbP\left\{|\Delta|\bbone_{\cE_0}>\delta\Var(L)^{1/2}\right\}\le\frac{4C_\omega^2M_\ast^2}{G\,\delta\,(\lv/G)^{1/2}}=\frac{4C_\omega^2M_\ast^2}{\delta\lv^{1/2}G^{1/2}}.
        \label{eq:a10}
    \end{equation}
    Let $\cE$ be the intersection of $\cE_0$, $\{D\le\varsigma_\delta\}$, $\{\cV(r)\ge(1-2\delta)\Var(L)\}$ and $\{|\Delta|\bbone_{\cE_0}\le\delta\Var(L)^{1/2}\}$. On $\cE$, \eqref{eq:a6}, \eqref{eq:a7} and $\lv/G\le\Var(L)$ give
    \begin{equation*}
        \hV^{\rm f}\ge\cV(r)-\frac{\delta\lv}G\ge(1-3\delta)\Var(L),
    \end{equation*}
    and \eqref{eq:a4} gives $|\htheta^{\rm f}-\theta|\le|L|+\delta\Var(L)^{1/2}$. Hence, on $\cE$,
    \begin{align*}
        &|L|\le\left\{z_\alpha(1-3\delta)^{1/2}-\delta\right\}\Var(L)^{1/2}\quad\text{implies}\displaybreak[0]\\
        &\left|\htheta^{\rm f}-\theta\right|\le z_\alpha(1-3\delta)^{1/2}\Var(L)^{1/2}\le z_\alpha\left(\hV^{\rm f}\right)^{1/2}\le q_{G-2}\left(\hV^{\rm f}\right)^{1/2},
    \end{align*}
    where the last inequality is Lemma \ref{lem:tquantile}. Therefore,
    \begin{equation*}
        \bbP\left\{\left|\htheta^{\rm f}-\theta\right|\le q_{G-2}\left(\hV^{\rm f}\right)^{1/2}\right\}\ge\bbP\left[|L|\le\left\{z_\alpha(1-3\delta)^{1/2}-\delta\right\}\Var(L)^{1/2}\right]-\bbP(\cE^c).
    \end{equation*}
    The second part of \eqref{eq:a9} with $x=z_\alpha(1-3\delta)^{1/2}-\delta$ bounds the first probability from below. The union bound with \eqref{eq:a3}, \eqref{eq:a8}, the first part of \eqref{eq:a9} and \eqref{eq:a10} bounds $\bbP(\cE^c)$ from above. Together they give \eqref{eq:a1}.
\end{proof}

\begin{proof}[Proof of Theorem \ref{thm:coverage}]
    Along the sequence of the theorem, the constants $C_\omega$, $M$, $\lpi$ and $\lv$ are fixed and $G\to\infty$, so $G\ge4C_\omega$ for all large $G$, and
    \begin{equation*}
        \lim_{G\to\infty}\Lambda_G=0,\qquad\lim_{G\to\infty}\eta_G(\delta)=0\quad\text{for every fixed }\delta.
    \end{equation*}
    For $\htheta=\htheta^{\rm b}$, and for $\htheta=\htheta^{\rm c}$ under Assumption \ref{ass:duration}, \eqref{eq:r2} gives
    \begin{equation*}
        \liminf_{G\to\infty}\bbP\left\{\left|\htheta-\theta\right|\le q_{G-1}\hV^{1/2}\right\}\ge2\Phi\left\{z_\alpha(1-2\delta)^{1/2}\right\}-1\quad\text{for every }\delta\in(0,1/2).
    \end{equation*}
    The left side does not depend on $\delta$. Letting $\delta\to0$ and using the continuity of $\Phi$, the right side tends to $2\Phi(z_\alpha)-1=1-\alpha$. For the pooled-anchor estimator, the last three terms of \eqref{eq:a1} tend to 0 for every fixed $\delta$, so
    \begin{equation*}
        \liminf_{G\to\infty}\bbP\left\{\left|\htheta^{\rm f}-\theta\right|\le q_{G-2}\left(\hV^{\rm f}\right)^{1/2}\right\}\ge2\Phi\left\{z_\alpha(1-3\delta)^{1/2}-\delta\right\}-1
    \end{equation*}
    for every sufficiently small $\delta>0$, and letting $\delta\to0$ gives $1-\alpha$. By \eqref{eq:a0},
    \begin{equation*}
        \left|\bbE(\htheta^{\rm f})-\theta\right|\le\frac{4C_\omega^2M_\ast^2(1+4M_\ast)}G\to0\quad\text{as }G\to\infty.
    \end{equation*}
\end{proof}

\section{Proof of Proposition \ref{prop:planning}}
\label{sec:planning}

In this section $\bbE$, $\Var$ and $\Cov$ are taken over the onsets and the noise $\varepsilon_{it}$ in model \eqref{eq:working}, and $\bm S=(S_1,\dots,S_N)$ is the vector of onsets. The $\varepsilon_{it}$ are independent of $\bm S$, so conditionally on $\bm S$ they are uncorrelated with mean 0 and variance $\sigma_i^2$. Let
\begin{equation}
    \overline\varepsilon_{i\cA}=\frac1{a_i}\sum_{t\in\cA_i}\varepsilon_{it},\qquad\overline\varepsilon_{i\cB}=\frac1{n_i}\sum_{t\in\cB_i}\varepsilon_{it},\qquad\varepsilon^\ast_i=\frac1{m_i}\sum_{t\in\cW_i}\left\{\frac{W_{it}}{\pi_{it}}-\frac{1-W_{it}}{1-\pi_{it}}\right\}\varepsilon_{it}.
    \label{eq:noise}
\end{equation}
The three quantities in \eqref{eq:noise} use disjoint sets of sessions. Conditionally on $\bm S$ they are uncorrelated with mean 0, and since $W_{it}(1-W_{it})=0$,
\begin{equation}
    \Var(\overline\varepsilon_{i\cA}\mid\bm S)=\frac{\sigma_i^2}{a_i},\qquad\Var(\overline\varepsilon_{i\cB}\mid\bm S)=\frac{\sigma_i^2}{n_i},\qquad\Var(\varepsilon^\ast_i\mid\bm S)=\frac{\sigma_i^2}{m_i^2}\sum_{t\in\cW_i}\left\{\frac{W_{it}}{\pi_{it}^2}+\frac{1-W_{it}}{(1-\pi_{it})^2}\right\}.
    \label{eq:noisevar}
\end{equation}
Since $\bbE(W_{it})=\pi_{it}$, the definition of $\psi_i$ in \eqref{eq:psi} gives
\begin{equation}
    \bbE\{\Var(\varepsilon^\ast_i\mid\bm S)\}=\frac{\sigma_i^2}{m_i^2}\sum_{t\in\cW_i}\left\{\frac1{\pi_{it}}+\frac1{1-\pi_{it}}\right\}=\frac{\sigma_i^2}{m_i^2}\sum_{t\in\cW_i}\frac1{\pi_{it}(1-\pi_{it})}=\sigma_i^2\psi_i.
    \label{eq:psiE}
\end{equation}

Under \eqref{eq:working}, $W_{it}=0$ for $t\in\cA_i$ and $W_{it}=1$ for $t\in\cB_i$, so $\overline Y_{i\cA}=\phi_i+\overline\varepsilon_{i\cA}$ and $\overline Y_{i\cB}=\phi_i+\tau_i+\overline\varepsilon_{i\cB}$. With $\theta_i=\tau_i/c_i$, for $t\in\cW_i$,
\begin{equation*}
    \widetilde Y_{it}=\frac{Y_{it}-\overline Y_{i\cA}}{c_i}=\theta_iW_{it}+\frac{\varepsilon_{it}-\overline\varepsilon_{i\cA}}{c_i},\qquad\widetilde Y_{i\cB}=\frac{\overline Y_{i\cB}-\overline Y_{i\cA}}{c_i}=\theta_i+\frac{\overline\varepsilon_{i\cB}-\overline\varepsilon_{i\cA}}{c_i}.
\end{equation*}
Substituting into \eqref{eq:sums} and using $W_{it}^2=W_{it}$ and $(1-W_{it})W_{it}=0$,
\begin{equation}
    H_{i1}=\theta_iX_{i1}+\frac1{c_im_i}\sum_{t\in\cW_i}\frac{W_{it}\varepsilon_{it}}{\pi_{it}}-\frac{\overline\varepsilon_{i\cA}X_{i1}}{c_i},\qquad H_{i0}=\frac1{c_im_i}\sum_{t\in\cW_i}\frac{(1-W_{it})\varepsilon_{it}}{1-\pi_{it}}-\frac{\overline\varepsilon_{i\cA}X_{i0}}{c_i}.
    \label{eq:Hmodel}
\end{equation}
Taking expectations in \eqref{eq:Hmodel} conditionally on $\bm S$ and then over $\bm S$, with $\bbE(X_{i1})=1$ from Lemma \ref{lem:weights}, gives $\bbE(H_{i1})=\theta_i$ and $\bbE(H_{i0})=0$. Hence $\mu_1=\sum_{i=1}^N\omega_i\theta_i=\theta$ and $\mu_0=0$, as stated in the main text. By \eqref{eq:bc}, \eqref{eq:ei}, \eqref{eq:noise} and \eqref{eq:Hmodel},
\begin{align}
    \htheta^{\rm b}_i&=H_{i1}-H_{i0}=\theta_iX_{i1}+\frac1{c_i}\left\{\varepsilon^\ast_i-\overline\varepsilon_{i\cA}\left(X_{i1}-X_{i0}\right)\right\},\label{eq:pb}\\
    \htheta^{\rm c}_i&=\widetilde Y_{i\cB}\left(1-X_{i1}\right)+\htheta^{\rm b}_i=\theta_i+\frac1{c_i}\left\{\varepsilon^\ast_i+\overline\varepsilon_{i\cB}\left(1-X_{i1}\right)-\overline\varepsilon_{i\cA}\left(1-X_{i0}\right)\right\},\label{eq:pc}\\
    e_i&=H_{i1}-\theta X_{i1}-H_{i0}=\left(\theta_i-\theta\right)X_{i1}+\frac1{c_i}\left\{\varepsilon^\ast_i-\overline\varepsilon_{i\cA}\left(X_{i1}-X_{i0}\right)\right\}.\label{eq:pe}
\end{align}
In \eqref{eq:pc}, the terms in $\theta_i$ sum to $\theta_i(1-X_{i1})+\theta_iX_{i1}=\theta_i$, and the terms in $\overline\varepsilon_{i\cA}$ sum to $-\overline\varepsilon_{i\cA}\{(1-X_{i1})+(X_{i1}-X_{i0})\}=-\overline\varepsilon_{i\cA}(1-X_{i0})$.

Let $\Theta_i$ denote any one of the three quantities in \eqref{eq:pb}--\eqref{eq:pe}. The sums $X_{i1}$ and $X_{i0}$ are functions of $S_i$. By \eqref{eq:noisevar}, the conditional means and variances are
\begin{equation*}
    \bbE(\htheta^{\rm b}_i\mid\bm S)=\theta_iX_{i1},\qquad\bbE(\htheta^{\rm c}_i\mid\bm S)=\theta_i,\qquad\bbE(e_i\mid\bm S)=(\theta_i-\theta)X_{i1},
\end{equation*}
and
\begin{align*}
    \Var(\htheta^{\rm b}_i\mid\bm S)=\Var(e_i\mid\bm S)&=\frac1{c_i^2}\left\{\Var(\varepsilon^\ast_i\mid\bm S)+\frac{\sigma_i^2(X_{i1}-X_{i0})^2}{a_i}\right\},\\
    \Var(\htheta^{\rm c}_i\mid\bm S)&=\frac1{c_i^2}\left\{\Var(\varepsilon^\ast_i\mid\bm S)+\frac{\sigma_i^2(1-X_{i1})^2}{n_i}+\frac{\sigma_i^2(1-X_{i0})^2}{a_i}\right\}.
\end{align*}
By Lemma \ref{lem:weights}, $\bbE(X_{i1})=\bbE(X_{i0})=1$, so $\bbE(1-X_{iw})^2=\Var(X_{iw})$ and $\bbE(X_{i1}-X_{i0})^2=\Var(X_{i1}-X_{i0})$. The decomposition $\Var(\Theta_i)=\bbE\{\Var(\Theta_i\mid\bm S)\}+\Var\{\bbE(\Theta_i\mid\bm S)\}$ and \eqref{eq:psiE} give the three variances
\begin{align}
    \Var(\htheta^{\rm b}_i)&=\theta_i^2\Var(X_{i1})+\frac{\sigma_i^2}{c_i^2}\left\{\psi_i+\frac{\Var(X_{i1}-X_{i0})}{a_i}\right\},\notag\\
    \Var(\htheta^{\rm c}_i)&=\frac{\sigma_i^2}{c_i^2}\left\{\psi_i+\frac{\Var(X_{i1})}{n_i}+\frac{\Var(X_{i0})}{a_i}\right\},\label{eq:pvar}\\
    \Var(e_i)&=(\theta_i-\theta)^2\Var(X_{i1})+\frac{\sigma_i^2}{c_i^2}\left\{\psi_i+\frac{\Var(X_{i1}-X_{i0})}{a_i}\right\}.\notag
\end{align}
For two cases $i\ne i'$, $\Cov(\Theta_i,\Theta_{i'}\mid\bm S)=0$ because the noise of different cases is uncorrelated given $\bm S$. The conditional means $\bbE(\Theta_i\mid\bm S)$ and $\bbE(\Theta_{i'}\mid\bm S)$ are functions of $S_i$ and of $S_{i'}$, which are independent. Hence
\begin{equation*}
    \Cov(\Theta_i,\Theta_{i'})=\bbE\{\Cov(\Theta_i,\Theta_{i'}\mid\bm S)\}+\Cov\{\bbE(\Theta_i\mid\bm S),\bbE(\Theta_{i'}\mid\bm S)\}=0,
\end{equation*}
and therefore $\Var(\sum_{i=1}^N\omega_i\Theta_i)=\sum_{i=1}^N\omega_i^2\Var(\Theta_i)$.
With the three variances in \eqref{eq:pvar}, $\htheta^{\rm b}=\sum_{i=1}^N\omega_i\htheta^{\rm b}_i$, $\htheta^{\rm c}=\sum_{i=1}^N\omega_i\htheta^{\rm c}_i$ and $L=\sum_{i=1}^N\omega_ie_i$, this is \eqref{eq:planb}--\eqref{eq:planf}. \qed

\subsection{Order of the Three Variances}

\begin{lemma}
    \label{lem:order}
    For every case $i$, $\Cov(X_{i1},X_{i0})\le0$ and $\Var(X_{i1}-X_{i0})\ge\Var(X_{i1})+\Var(X_{i0})$. If $n_i\ge a_i$, then
    \begin{equation}
        \frac{\Var(X_{i1})}{n_i}+\frac{\Var(X_{i0})}{a_i}\le\frac{\Var(X_{i1}-X_{i0})}{a_i}.
        \label{eq:order}
    \end{equation}
\end{lemma}

\begin{proof}
    The indicator $W_{it}(s)=\bbone(t\ge s)$ is nonincreasing in $s$, so by \eqref{eq:sums}, $X_{i1}(s)\ge X_{i1}(s')$ and $X_{i0}(s)\le X_{i0}(s')$ for $s<s'$. Let $S'_i$ be an independent copy of $S_i$. The two differences below have opposite signs or one of them is 0, so
    \begin{equation*}
        \Cov\left(X_{i1},X_{i0}\right)=\frac12\,\bbE\left[\left\{X_{i1}(S_i)-X_{i1}(S'_i)\right\}\left\{X_{i0}(S_i)-X_{i0}(S'_i)\right\}\right]\le0.
    \end{equation*}
    Hence $\Var(X_{i1}-X_{i0})=\Var(X_{i1})+\Var(X_{i0})-2\Cov(X_{i1},X_{i0})\ge\Var(X_{i1})+\Var(X_{i0})$. If $n_i\ge a_i$,
    \begin{equation*}
        \frac{\Var(X_{i1})}{n_i}+\frac{\Var(X_{i0})}{a_i}\le\frac{\Var(X_{i1})+\Var(X_{i0})}{a_i}\le\frac{\Var(X_{i1}-X_{i0})}{a_i}.
    \end{equation*}
\end{proof}

Suppose that $\theta_i=\theta$ and $n_i\ge a_i$ for every case. Then the first term in brackets in \eqref{eq:planf} is 0, and \eqref{eq:planb}--\eqref{eq:planf} with \eqref{eq:order} give
\begin{equation*}
    \Var(\htheta^{\rm c})\le\Var(L)\le\Var(L)+\sum_{i=1}^N\omega_i^2\theta^2\Var(X_{i1})=\Var(\htheta^{\rm b}).
\end{equation*}

\subsection{Equally Likely Consecutive Onsets}

Let $\cS_i=\{a_i+1,\dots,a_i+K\}$ with $p_i(s)=1/K$, and let $h_k=\sum_{j=1}^{k-1}1/j$ for integers $k\ge1$, so that $h_1=0$ and $h_K$ is as in the main text. The window sessions are $t=a_i+j$ with $j=1,\dots,K-1$, so $m_i=K-1$ and $\pi_{i,a_i+j}=j/K$. Since $K^2/\{j(K-j)\}=K\{1/j+1/(K-j)\}$,
\begin{equation*}
    \psi_i=\frac1{(K-1)^2}\sum_{j=1}^{K-1}\frac{K^2}{j(K-j)}=\frac K{(K-1)^2}\sum_{j=1}^{K-1}\left(\frac1j+\frac1{K-j}\right)=\frac{2Kh_K}{(K-1)^2}.
\end{equation*}
When $S_i=a_i+k$ with $1\le k\le K$, session $a_i+j$ is treated if and only if $j\ge k$, so
\begin{equation}
    \begin{aligned}
        X_{i1}(a_i+k)&=\frac1{K-1}\sum_{j=k}^{K-1}\frac Kj=\frac{K(h_K-h_k)}{K-1},\\
        X_{i0}(a_i+k)&=\frac1{K-1}\sum_{j=1}^{k-1}\frac K{K-j}=\frac{K(h_K-h_{K-k+1})}{K-1}.
    \end{aligned}
    \label{eq:xuniform}
\end{equation}

\begin{lemma}
    \label{lem:harmonic}
    For every integer $K\ge1$, $\sum_{k=1}^K(h_K-h_k)^2=2(K-1)-h_K$.
\end{lemma}

\begin{proof}
    The claim holds for $K=1$ because both sides are 0. Suppose that it holds for some $K\ge1$. Exchanging the order of summation,
    \begin{equation*}
        \sum_{k=1}^K(h_K-h_k)=\sum_{k=1}^K\sum_{j=k}^{K-1}\frac1j=\sum_{j=1}^{K-1}\frac1j\cdot j=K-1.
    \end{equation*}
    Since $h_{K+1}=h_K+1/K$ and the term $k=K+1$ of the sum for $K+1$ is 0,
    \begin{align*}
        \sum_{k=1}^{K+1}(h_{K+1}-h_k)^2&=\sum_{k=1}^K\left(h_K-h_k+\frac1K\right)^2=\sum_{k=1}^K(h_K-h_k)^2+\frac2K\sum_{k=1}^K(h_K-h_k)+\frac1K\\
        &=2(K-1)-h_K+\frac{2(K-1)}K+\frac1K=2K-h_K-\frac1K=2K-h_{K+1},
    \end{align*}
    which is the claim for $K+1$.
\end{proof}

By \eqref{eq:xuniform}, Lemma \ref{lem:harmonic} and $\bbE(X_{i1})=1$,
\begin{equation*}
    \Var(X_{i1})=\frac1K\sum_{k=1}^K\frac{K^2(h_K-h_k)^2}{(K-1)^2}-1=\frac{K\{2(K-1)-h_K\}}{(K-1)^2}-1.
\end{equation*}
The map $k\mapsto K-k+1$ is a bijection of $\{1,\dots,K\}$, and by \eqref{eq:xuniform} it carries the values of $X_{i0}$ to the values of $X_{i1}$. The onsets are equally likely, so $X_{i0}$ and $X_{i1}$ have the same distribution and $\Var(X_{i0})=\Var(X_{i1})$. This proves \eqref{eq:uniform}. The variance of $X_{i1}-X_{i0}$ is computed from \eqref{eq:xuniform} by summing over the $K$ onsets.

Under model \eqref{eq:working}, the phase contrast \eqref{eq:contrast} at the realized onset is $\tau_i$ plus the mean of the $\varepsilon_{it}$ over sessions $S_i,\dots,T_i$ minus their mean over sessions $1,\dots,S_i-1$. Its conditional mean given $\bm S$ is $\tau_i$, so
\begin{equation}
    \Var\{\gamma_i(S_i)\}=\sigma_i^2\,\bbE\left(\frac1{T_i-S_i+1}+\frac1{S_i-1}\right).
    \label{eq:pcvar}
\end{equation}
Table \ref{tab:planfull} lists $\psi_i$, the variances of $X_{i1}$ and of $X_{i1}-X_{i0}$, and the standard deviations that one case contributes, for a grid of designs that contains the designs of Table \ref{tab:plan}.

\begin{table}[ht]
    \centering
    \caption{Planning quantities for one case under model \eqref{eq:working} with $K$ equally likely consecutive onsets, $a_i=n_i$ sessions before and after the window, $c_i=\sigma_i$, and the same standardized effect for all cases. The last three columns are the square roots of the terms in braces in \eqref{eq:planc} and \eqref{eq:planf} and of \eqref{eq:pcvar} divided by $\sigma_i^2$, which are standard deviations in units of $\sigma_i$.}
    \label{tab:planfull}
    \small
    \begin{tabular}{cccccccc}
        \toprule
        $K$ & $a_i=n_i$ & $\psi_i$ & $\Var(X_{i1})$ & $\Var(X_{i1}-X_{i0})$ & $\htheta^{\rm c}$ & $\htheta^{\rm f}$ & Phase contrast \\
        \midrule
        4 & 3 & 1.630 & 0.852 & 3.210 & 1.48 & 1.64 & 0.69 \\
        4 & 5 & 1.630 & 0.852 & 3.210 & 1.40 & 1.51 & 0.56 \\
        4 & 8 & 1.630 & 0.852 & 3.210 & 1.36 & 1.43 & 0.46 \\
        6 & 3 & 1.096 & 0.852 & 3.119 & 1.29 & 1.46 & 0.64 \\
        6 & 5 & 1.096 & 0.852 & 3.119 & 1.20 & 1.31 & 0.53 \\
        6 & 8 & 1.096 & 0.852 & 3.119 & 1.14 & 1.22 & 0.44 \\
        8 & 3 & 0.847 & 0.862 & 3.103 & 1.19 & 1.37 & 0.60 \\
        8 & 5 & 0.847 & 0.862 & 3.103 & 1.09 & 1.21 & 0.50 \\
        8 & 8 & 0.847 & 0.862 & 3.103 & 1.03 & 1.11 & 0.43 \\
        12 & 3 & 0.599 & 0.882 & 3.109 & 1.09 & 1.28 & 0.54 \\
        12 & 5 & 0.599 & 0.882 & 3.109 & 0.98 & 1.10 & 0.47 \\
        12 & 8 & 0.599 & 0.882 & 3.109 & 0.91 & 0.99 & 0.40 \\
        16 & 3 & 0.472 & 0.897 & 3.124 & 1.03 & 1.23 & 0.50 \\
        16 & 5 & 0.472 & 0.897 & 3.124 & 0.91 & 1.05 & 0.44 \\
        16 & 8 & 0.472 & 0.897 & 3.124 & 0.83 & 0.93 & 0.38 \\
        \bottomrule
    \end{tabular}
\end{table}

\section{Proof of Proposition \ref{prop:scale}}
\label{sec:scale}

As in Section \ref{sec:planning}, $\bbE$ is taken over the onsets and the noise of model \eqref{eq:working}. Let $i'\in\cN_i$. For $t\in\cA_{i'}$, $W_{i't}=0$ and \eqref{eq:working} gives $Y_{i't}=\phi_{i'}+\varepsilon_{i't}$, so that $Y_{i't}-\overline Y_{i'\cA}=\varepsilon_{i't}-\overline\varepsilon_{i'\cA}$ with the notation \eqref{eq:noise}. The $\varepsilon_{i't}$, $t\in\cA_{i'}$, are independent and normal with mean 0 and variance $\sigma^2$, so $\sigma^{-2}\sum_{t\in\cA_{i'}}(\varepsilon_{i't}-\overline\varepsilon_{i'\cA})^2$ has the chi-squared distribution with $a_{i'}-1$ degrees of freedom. These sums are independent across the cases in $\cN_i$. By \eqref{eq:pooledsd},
\begin{equation*}
    \chi^2=\frac{\nu_ic_i^2}{\sigma^2}=\frac1{\sigma^2}\sum_{i'\in\cN_i}\sum_{t\in\cA_{i'}}\left(\varepsilon_{i't}-\overline\varepsilon_{i'\cA}\right)^2
\end{equation*}
has the chi-squared distribution with $\nu_i=\sum_{i'\in\cN_i}(a_{i'}-1)$ degrees of freedom, whose density at $x>0$ is $x^{\nu_i/2-1}e^{-x/2}/\{2^{\nu_i/2}\Gamma(\nu_i/2)\}$. Since $\nu_i\ge2$, the integral below is finite, and the definition of the gamma function gives
\begin{equation*}
    \begin{split}\bbE\left\{(\chi^2)^{-1/2}\right\}&=\frac1{2^{\nu_i/2}\Gamma(\nu_i/2)}\int_0^\infty x^{(\nu_i-1)/2-1}e^{-x/2}\,dx\\&=\frac{2^{(\nu_i-1)/2}\,\Gamma\{(\nu_i-1)/2\}}{2^{\nu_i/2}\,\Gamma(\nu_i/2)}=\frac{\Gamma\{(\nu_i-1)/2\}}{2^{1/2}\,\Gamma(\nu_i/2)}.\end{split}
\end{equation*}
Since $1/c_i=\nu_i^{1/2}/\{\sigma(\chi^2)^{1/2}\}$ and $\tau_i$ is a constant,
\begin{equation*}
    \bbE\left(\frac{\tau_i}{c_i}\right)=\frac{\tau_i\nu_i^{1/2}}\sigma\cdot\frac{\Gamma\{(\nu_i-1)/2\}}{2^{1/2}\,\Gamma(\nu_i/2)}=\frac{\tau_i}\sigma\cdot\frac{(\nu_i/2)^{1/2}\,\Gamma\{(\nu_i-1)/2\}}{\Gamma(\nu_i/2)}=\frac{\tau_i}{\sigma\iota_i}.
\end{equation*}
\qed

The scale $c_i/\iota_i$ is a function of the pre-window outcomes, so it takes the same value under every vector of onsets by \eqref{eq:prefixed}, and Corollary \ref{cor:standardized} holds with this scale. With it, $\theta_i=\iota_i\tau_i/c_i$, and the display above gives $\bbE(\theta_i)=\tau_i/\sigma$.

The factor satisfies $\iota_i<1$. The two displays above give $\bbE(1/c_i)=1/(\sigma\iota_i)$. The chi-squared distribution with $\nu_i$ degrees of freedom has mean $\nu_i$, so $\bbE(c_i^2)=\sigma^2$ and $\bbE(c_i)\le\{\bbE(c_i^2)\}^{1/2}=\sigma$. The function $x\mapsto1/x$ is strictly convex on $(0,\infty)$ and $c_i$ is not constant, so Jensen's inequality gives
\begin{equation*}
    \frac1{\sigma\iota_i}=\bbE\left(\frac1{c_i}\right)>\frac1{\bbE(c_i)}\ge\frac1\sigma.
\end{equation*}

\section{Further Results}
\label{sec:further}

\subsection{The Two-Anchor Estimator Without Assumption \ref{ass:duration}}
\label{sec:biasc}

\begin{propositionS}
    \label{prop:biasc}
    Under Assumption \ref{ass:anticipation},
    \begin{equation}
        \bbE(\htau^{\rm c}_i)-\tau_i=-\Cov\left(X_{i1},\overline Y_{i\cB}\right),\quad\left|\bbE(\htau^{\rm c}_i)-\tau_i\right|\le\frac{\Var(X_{i1})^{1/2}}2\left\{\max_{s\in\cS_i}\overline Y_{i\cB}(s)-\min_{s\in\cS_i}\overline Y_{i\cB}(s)\right\}.
        \label{eq:biasc}
    \end{equation}
    If $\overline Y_{i\cB}(s)\ge\overline Y_{i\cB}(s')$ for all $s,s'\in\cS_i$ with $s<s'$, then $\bbE(\htau^{\rm c}_i)\le\tau_i$.
\end{propositionS}

\begin{proof}
    Multiplying the second identity of \eqref{eq:bc} by $c_i$ gives $\htau^{\rm c}_i=\htau^{\rm b}_i+(\overline Y_{i\cB}-\overline Y_{i\cA})(1-X_{i1})$. By Theorem \ref{thm:unbiased}, $\bbE(\htau^{\rm b}_i)=\tau_i$. By \eqref{eq:prefixed}, $\overline Y_{i\cA}$ is a constant, and $\bbE(1-X_{i1})=0$ by Lemma \ref{lem:weights}. Hence
    \begin{equation*}
        \bbE(\htau^{\rm c}_i)-\tau_i=\bbE\left\{\overline Y_{i\cB}\left(1-X_{i1}\right)\right\}=-\Cov\left(X_{i1},\overline Y_{i\cB}\right).
    \end{equation*}
    Let $y_{\min}$ and $y_{\max}$ be the smallest and the largest value of $\overline Y_{i\cB}(s)$ over $s\in\cS_i$. Then, $|\overline Y_{i\cB}-(y_{\min}+y_{\max})/2|\le(y_{\max}-y_{\min})/2$, so
    \begin{equation*}
        \Var(\overline Y_{i\cB})\le\bbE\left[\left\{\overline Y_{i\cB}-\frac{y_{\min}+y_{\max}}2\right\}^2\right]\le\frac{(y_{\max}-y_{\min})^2}4,
    \end{equation*}
    and the Cauchy--Schwarz inequality gives the bound in \eqref{eq:biasc}. For the last claim, $W_{it}(s)=\bbone(t\ge s)$ is nonincreasing in $s$, so $X_{i1}(s)\ge X_{i1}(s')$ for $s<s'$. Let $S'_i$ be an independent copy of $S_i$. The two differences below have the same sign or one of them is 0, so
    \begin{equation*}
        \Cov\left(X_{i1},\overline Y_{i\cB}\right)=\frac12\,\bbE\left[\left\{X_{i1}(S_i)-X_{i1}(S'_i)\right\}\left\{\overline Y_{i\cB}(S_i)-\overline Y_{i\cB}(S'_i)\right\}\right]\ge0.
    \end{equation*}
\end{proof}

Suppose that the treated outcome at each post-window session is larger when the intervention started earlier. Then an earlier onset gives a larger post-window mean, and the last claim of Proposition \ref{prop:biasc} gives a bias that is negative or 0. The bias of $\htheta^{\rm c}$ is $-\sum_{i=1}^N\omega_i\Cov(X_{i1},\overline Y_{i\cB})/c_i$.

\subsection{Permuted Onsets}
\label{sec:permuted}

\begin{propositionS}
    \label{prop:permuted}
    Let the cases of block $g$ have the same weight $\omega_i=\Omega_g/|\cC_g|$, and let their onsets be a uniformly distributed permutation of a fixed list of $|\cC_g|$ sessions that contains at least two distinct sessions. Then, $X_{(g)w}=1$ for $w\in\{0,1\}$ and every permutation. If this holds for every block, then $\overline X_1=\overline X_0=1$ and $\htheta^{\rm f}=\htheta^{\rm b}$.
\end{propositionS}

\begin{proof}
    Every case of the block has the same marginal distribution of the onset, which puts probability on each session in proportion to the number of times it appears in the list. The cases therefore share $\cS_i$, $a_i$, $b_i$, $m_i$ and the probabilities $\pi_{it}$, and the function $s\mapsto X_{iw}(s)$ is the same for all of them. The sum $\sum_{i\in\cC_g}X_{iw}(S_i)$ is the sum of this function over the list, which is the same for every permutation. It equals its expectation, which is $|\cC_g|$ by Lemma \ref{lem:weights}. By \eqref{eq:blocksums},
    \begin{equation*}
        X_{(g)w}=\frac1{\Omega_g}\sum_{i\in\cC_g}\frac{\Omega_g}{|\cC_g|}X_{iw}(S_i)=1.
    \end{equation*}
    If this holds for every block, $\overline X_w=\sum_{g=1}^G\Omega_gX_{(g)w}=1$, so $\hmu_w=\overline H_w$ and, by \eqref{eq:bc},
    \begin{equation*}
        \htheta^{\rm f}=\overline H_1-\overline H_0=\sum_{i=1}^N\omega_i\left(H_{i1}-H_{i0}\right)=\sum_{i=1}^N\omega_i\htheta^{\rm b}_i.
    \end{equation*}
\end{proof}

\subsection{Randomization Test of the Sharp Null Hypothesis}
\label{sec:test}

The sharp null hypothesis states that $Y_{it}(s)=Y_{it}(s')$ for every case $i$, every session $t$ and all $s,s'\in\cS_i$. Under it, the observed outcomes are the outcomes under every vector of onsets. A test statistic computed from the onsets and the outcomes is then a known function $\cT(\bm s)$ of the vector of onsets $\bm s$. Let $\Psi(x)=\bbP\{\cT(\bm S)\ge x\}$, and let $\bm S^{(1)},\dots,\bm S^{(B)}$ be independent draws from the design, independent of $\bm S$. The exact and the Monte Carlo $p$-values are
\begin{equation}
    P_{\rm ex}=\Psi\{\cT(\bm S)\},\qquad P_{\rm mc}=\frac{1+\#\{1\le j\le B:\cT(\bm S^{(j)})\ge\cT(\bm S)\}}{1+B}.
    \label{eq:pvalues}
\end{equation}

\begin{propositionS}
    \label{prop:test}
    Under the sharp null hypothesis, $\bbP(P_{\rm ex}\le\alpha)\le\alpha$ and $\bbP(P_{\rm mc}\le\alpha)\le\alpha$ for every $\alpha\in(0,1)$.
\end{propositionS}

\begin{proof}
    The statistic $\cT(\bm S)$ takes finitely many values. If none of these values $x$ has $\Psi(x)\le\alpha$, then $\bbP(P_{\rm ex}\le\alpha)=0$. Otherwise let $x_\ast$ be the smallest of these values with $\Psi(x_\ast)\le\alpha$. On the event $\{P_{\rm ex}\le\alpha\}$, the value $x=\cT(\bm S)$ has $\Psi(x)\le\alpha$, so $\cT(\bm S)\ge x_\ast$. Hence
    \begin{equation*}
        \bbP(P_{\rm ex}\le\alpha)\le\bbP\{\cT(\bm S)\ge x_\ast\}=\Psi(x_\ast)\le\alpha.
    \end{equation*}
    For the Monte Carlo $p$-value, write $\bm S^{(0)}=\bm S$ and, for $0\le j\le B$, let $I_j$ be the number of indices $j'\in\{0,\dots,B\}$ with $\cT(\bm S^{(j')})\ge\cT(\bm S^{(j)})$. Then $P_{\rm mc}=I_0/(1+B)$. Fix an integer $k\ge0$. If some index has $I_j\le k$, let $j_\ast$ be such an index with the smallest value of $\cT(\bm S^{(j)})$. Every index $j$ with $I_j\le k$ has $\cT(\bm S^{(j)})\ge\cT(\bm S^{(j_\ast)})$, so there are at most $I_{j_\ast}\le k$ of them. The vectors $\bm S^{(0)},\dots,\bm S^{(B)}$ are independent and identically distributed, so $I_0,\dots,I_B$ are exchangeable, and
    \begin{equation*}
        \bbP(I_0\le k)=\frac1{1+B}\sum_{j=0}^B\bbP(I_j\le k)=\frac{\bbE\,\#\{0\le j\le B:I_j\le k\}}{1+B}\le\frac k{1+B}.
    \end{equation*}
    Let $k$ be the integer part of $\alpha(1+B)$. Since $I_0$ is an integer,
    \begin{equation*}
        \bbP(P_{\rm mc}\le\alpha)=\bbP\{I_0\le\alpha(1+B)\}=\bbP(I_0\le k)\le\frac k{1+B}\le\alpha.
    \end{equation*}
\end{proof}

\section{Simulation Details and Additional Studies}
\label{sec:simulations}

In Studies 1 and 3--6, the untreated outcome of a case is a level drawn from a normal distribution with mean 10 and standard deviation 2, plus a baseline pattern and a first-order autoregressive normal noise with unit variance and autocorrelation 0.3. The effect of a case is drawn from a normal distribution with mean 2 and standard deviation 0.5. Unless a study states otherwise, there are 20 sessions, and the onsets are drawn independently and uniformly from sessions 6 to 13. The effect is present in full from the onset session onward, with $c_i=1$, equal weights, and each case as a block. Every replication draws a new population and new onsets, and every study has 1,000 replications. The bias is the mean of the estimate minus the pooled effect $\theta$ of the replication, and its Monte Carlo standard error is the standard deviation of this difference divided by $1000^{1/2}$. Intervals have nominal level 95\%, and the Monte Carlo standard error of a coverage rate near 95\% is 0.7 percentage points.

\subsection{Studies 1 to 3: Full Results}

Table \ref{tab:sim1full} adds the Monte Carlo standard error of the bias and the mean standard error to Table \ref{tab:sim1}. Table \ref{tab:sim2full} adds the bias with its Monte Carlo standard error and the mean standard error to Table \ref{tab:sim2}. Table \ref{tab:sim3full} adds the mean standard error to Table \ref{tab:sim3}.

\begin{table}[ht]
    \centering
    \caption{Study 1: bias with its Monte Carlo standard error (MCSE), standard deviation of the estimation error (SD), mean standard error (SE), and coverage of 95\% intervals (\%).}
    \label{tab:sim1full}
    \small
    \begin{tabular}{llccccc}
        \toprule
        Baseline & Method & Bias & MCSE & SD & Mean SE & Coverage \\
        \midrule
        Flat & $\htheta^{\rm b}$ & 0.005 & 0.016 & 0.501 & 0.493 & 94.3 \\
        & $\htheta^{\rm c}$ & 0.005 & 0.008 & 0.257 & 0.275 & 96.5 \\
        & $\htheta^{\rm f}$ & 0.002 & 0.010 & 0.313 & 0.316 & 94.8 \\
        & Phase contrast & 0.005 & 0.004 & 0.123 & 0.158 & 99.0 \\
        & Multilevel, no trend & 0.004 & 0.004 & 0.119 & 0.126 & 96.7 \\
        & Multilevel, linear trend & 0.000 & 0.006 & 0.197 & 0.203 & 95.5 \\
        Linear trend & $\htheta^{\rm b}$ & $-0.014$ & 0.022 & 0.687 & 0.654 & 94.7 \\
        & $\htheta^{\rm c}$ & $-0.004$ & 0.009 & 0.292 & 0.301 & 95.1 \\
        & $\htheta^{\rm f}$ & 0.007 & 0.010 & 0.318 & 0.315 & 94.4 \\
        & Phase contrast & 1.000 & 0.004 & 0.116 & 0.158 & 0.0 \\
        & Multilevel, no trend & 0.881 & 0.004 & 0.124 & 0.138 & 0.0 \\
        & Multilevel, linear trend & $-0.001$ & 0.006 & 0.201 & 0.202 & 95.1 \\
        Common shift & $\htheta^{\rm b}$ & $-0.030$ & 0.021 & 0.672 & 0.669 & 95.8 \\
        & $\htheta^{\rm c}$ & $-0.004$ & 0.012 & 0.374 & 0.362 & 93.6 \\
        & $\htheta^{\rm f}$ & 0.022 & 0.012 & 0.369 & 0.350 & 92.3 \\
        & Phase contrast & 1.263 & 0.004 & 0.126 & 0.159 & 0.0 \\
        & Multilevel, no trend & 1.078 & 0.005 & 0.143 & 0.144 & 0.0 \\
        & Multilevel, linear trend & 0.378 & 0.007 & 0.229 & 0.217 & 56.4 \\
        \bottomrule
    \end{tabular}
\end{table}

\begin{table}[ht]
    \centering
    \caption{Study 2: bias with its Monte Carlo standard error (MCSE), standard deviation of the estimation error (SD), standard deviation predicted by Proposition \ref{prop:planning} (Pred.), mean standard error (SE), and coverage of 95\% intervals (\%).}
    \label{tab:sim2full}
    \footnotesize
    \setlength{\tabcolsep}{3pt}
    \begin{tabular}{cccccccccccccc}
        \toprule
        & & \multicolumn{6}{c}{$\htheta^{\rm c}$} & \multicolumn{6}{c}{$\htheta^{\rm f}$} \\
        \cmidrule(lr){3-8}\cmidrule(lr){9-14}
        $K$ & $N$ & Bias & MCSE & SD & Pred. & Mean SE & Cov. & Bias & MCSE & SD & Pred. & Mean SE & Cov. \\
        \midrule
        4 & 4 & 0.027 & 0.022 & 0.684 & 0.684 & 0.630 & 96.1 & 0.010 & 0.025 & 0.805 & 0.754 & 0.699 & 97.7 \\
        4 & 6 & $-0.005$ & 0.018 & 0.561 & 0.558 & 0.531 & 96.6 & $-0.008$ & 0.020 & 0.637 & 0.615 & 0.574 & 95.2 \\
        4 & 8 & $-0.017$ & 0.015 & 0.468 & 0.484 & 0.464 & 95.3 & $-0.008$ & 0.017 & 0.528 & 0.533 & 0.497 & 94.0 \\
        4 & 16 & $-0.007$ & 0.011 & 0.336 & 0.342 & 0.335 & 95.3 & $-0.004$ & 0.012 & 0.375 & 0.377 & 0.364 & 95.5 \\
        4 & 32 & $-0.004$ & 0.008 & 0.238 & 0.242 & 0.242 & 95.4 & $-0.008$ & 0.008 & 0.261 & 0.266 & 0.264 & 95.4 \\
        8 & 4 & 0.016 & 0.017 & 0.537 & 0.531 & 0.474 & 96.7 & 0.014 & 0.018 & 0.580 & 0.606 & 0.514 & 97.1 \\
        8 & 6 & $-0.017$ & 0.014 & 0.428 & 0.433 & 0.395 & 95.6 & $-0.033$ & 0.015 & 0.471 & 0.494 & 0.430 & 95.6 \\
        8 & 8 & 0.032 & 0.012 & 0.366 & 0.375 & 0.360 & 95.7 & 0.031 & 0.013 & 0.421 & 0.428 & 0.391 & 95.0 \\
        8 & 16 & 0.003 & 0.009 & 0.274 & 0.265 & 0.257 & 94.9 & 0.002 & 0.010 & 0.306 & 0.303 & 0.285 & 93.4 \\
        8 & 32 & $-0.009$ & 0.006 & 0.196 & 0.188 & 0.185 & 94.0 & $-0.009$ & 0.007 & 0.224 & 0.214 & 0.207 & 94.3 \\
        \bottomrule
    \end{tabular}
\end{table}

\begin{table}[ht]
    \centering
    \caption{Study 3: bias with its Monte Carlo standard error (MCSE), standard deviation of the estimation error (SD), mean standard error (SE), and coverage of 95\% intervals (\%), where the effect reaches its full size at the $R$th treated session.}
    \label{tab:sim3full}
    \small
    \begin{tabular}{clccccc}
        \toprule
        $R$ & Method & Bias & MCSE & SD & Mean SE & Coverage \\
        \midrule
        1 & $\htheta^{\rm b}$ & 0.016 & 0.015 & 0.485 & 0.485 & 95.4 \\
        & $\htheta^{\rm c}$ & 0.013 & 0.008 & 0.248 & 0.269 & 96.5 \\
        & $\htheta^{\rm f}$ & 0.013 & 0.010 & 0.305 & 0.310 & 95.1 \\
        3 & $\htheta^{\rm b}$ & 0.013 & 0.013 & 0.407 & 0.401 & 95.5 \\
        & $\htheta^{\rm c}$ & $-0.038$ & 0.009 & 0.274 & 0.280 & 95.9 \\
        & $\htheta^{\rm f}$ & 0.005 & 0.009 & 0.295 & 0.297 & 95.2 \\
        6 & $\htheta^{\rm b}$ & $-0.005$ & 0.011 & 0.345 & 0.339 & 94.7 \\
        & $\htheta^{\rm c}$ & $-0.143$ & 0.010 & 0.319 & 0.309 & 93.9 \\
        & $\htheta^{\rm f}$ & $-0.004$ & 0.009 & 0.296 & 0.289 & 95.2 \\
        \bottomrule
    \end{tabular}
\end{table}

\subsection{Study 4: Onsets Permuted Within Studies}
\label{sec:sim4}

This study checks Theorem \ref{thm:coverage} and Proposition \ref{prop:permuted} when the blocks are studies. There are $G=6$, 12, or 24 studies of four cases each. Within a study, the onsets 6, 8, 10 and 12 are assigned to the four cases by a uniformly distributed permutation, independently across studies. The baseline has a linear trend of 0.1 per session. We compute $\htheta^{\rm b}$, $\htheta^{\rm c}$ and $\htheta^{\rm f}$, with the studies as blocks. Table \ref{tab:sim4} reports the bias, the standard deviation of the estimation error, the mean standard error, and the coverage of the intervals.

The results match Theorem \ref{thm:coverage} and Proposition \ref{prop:permuted}. The biases are within two Monte Carlo standard errors of 0, and the intervals cover the pooled effect in 94.5\% to 97.4\% of the replications. The estimators $\htheta^{\rm b}$ and $\htheta^{\rm f}$ have the same bias and standard deviation in every row, because the two estimates are equal by Proposition \ref{prop:permuted}. Their intervals differ. With six studies, the mean standard errors are 0.305 and 0.341, because $\hV^{\rm f}$ contains the factor $(G-1)/(G-2)=1.25$.

\begin{table}[ht]
    \centering
    \caption{Study 4: bias with its Monte Carlo standard error (MCSE), standard deviation of the estimation error (SD), mean standard error (SE), and coverage of 95\% intervals (\%).}
    \label{tab:sim4}
    \small
    \begin{tabular}{clccccc}
        \toprule
        Studies & Method & Bias & MCSE & SD & Mean SE & Coverage \\
        \midrule
        6 & $\htheta^{\rm b}$ & 0.012 & 0.010 & 0.324 & 0.305 & 94.5 \\
        & $\htheta^{\rm c}$ & 0.005 & 0.008 & 0.262 & 0.258 & 96.3 \\
        & $\htheta^{\rm f}$ & 0.012 & 0.010 & 0.324 & 0.341 & 97.4 \\
        12 & $\htheta^{\rm b}$ & $-0.006$ & 0.007 & 0.216 & 0.226 & 96.1 \\
        & $\htheta^{\rm c}$ & $-0.009$ & 0.006 & 0.180 & 0.188 & 96.4 \\
        & $\htheta^{\rm f}$ & $-0.006$ & 0.007 & 0.216 & 0.237 & 96.7 \\
        24 & $\htheta^{\rm b}$ & 0.002 & 0.005 & 0.157 & 0.164 & 96.1 \\
        & $\htheta^{\rm c}$ & $-0.002$ & 0.004 & 0.129 & 0.137 & 96.3 \\
        & $\htheta^{\rm f}$ & 0.002 & 0.005 & 0.157 & 0.167 & 96.6 \\
        \bottomrule
    \end{tabular}
\end{table}

\subsection{Study 5: Standardized Effects}
\label{sec:sim5}

This study examines Corollary \ref{cor:standardized} and Theorem \ref{thm:coverage} for three scales. There are six studies, each with four cases. The outcomes of the six studies are multiplied by 1, 2, 5, 10, 20, and 50, so the studies measure the outcome in different units. The baseline has a linear trend of 0.1 per session before the multiplication. The scale $c_i$ is the standard deviation of the five pre-window outcomes of the case, the standard deviation of the pre-window outcomes pooled across the four cases of the study, or the pre-window mean of the case. We compute $\htheta^{\rm c}$ and $\htheta^{\rm f}$, with each case as a block. Table \ref{tab:sim5} reports the mean pooled effect, bias, standard deviation of the estimation error, mean standard error, and coverage of the intervals.

The results match Corollary \ref{cor:standardized} and Theorem \ref{thm:coverage}. The biases of $\htheta^{\rm c}$ are within two Monte Carlo standard errors of 0 for every scale, and the bias of $\htheta^{\rm f}$ is at most 1.5\% of the pooled effect. For the pre-window mean, the pooled effect is 0.203, and the biases are 0.001 and 0.003. The intervals cover the pooled effect in 96.1\% to 98.7\% of the replications. Coverage exceeds the nominal level because standardized case effects differ across cases, adding the last term of \eqref{eq:variance} to the expectation of the variance estimator.

\begin{table}[ht]
    \centering
    \caption{Study 5: mean pooled effect, bias with its Monte Carlo standard error (MCSE), standard deviation of the estimation error (SD), mean standard error (SE), and coverage of 95\% intervals (\%).}
    \label{tab:sim5}
    \small
    \setlength{\tabcolsep}{5pt}
    \begin{tabular}{lclccccc}
        \toprule
        Scale & Mean $\theta$ & Method & Bias & MCSE & SD & Mean SE & Coverage \\
        \midrule
        Standard deviation, case & 2.722 & $\htheta^{\rm c}$ & 0.018 & 0.014 & 0.448 & 0.526 & 98.7 \\
        & & $\htheta^{\rm f}$ & 0.015 & 0.019 & 0.604 & 0.627 & 97.0 \\
        Standard deviation, study & 2.256 & $\htheta^{\rm c}$ & 0.014 & 0.011 & 0.340 & 0.348 & 96.8 \\
        & & $\htheta^{\rm f}$ & 0.023 & 0.012 & 0.382 & 0.383 & 96.1 \\
        Mean, case & 0.203 & $\htheta^{\rm c}$ & 0.001 & 0.001 & 0.031 & 0.033 & 96.8 \\
        & & $\htheta^{\rm f}$ & 0.003 & 0.001 & 0.039 & 0.038 & 96.5 \\
        \bottomrule
    \end{tabular}
\end{table}

\subsection{Study 6: Number of Cases Under a Common Shift}
\label{sec:sim6}

This study checks Theorem \ref{thm:coverage} as the number of cases grows when the untreated outcome of every case shifts at the same session. We consider $N=6$, 12, 24, 48, or 96 cases, and the baseline pattern is a shift of 1.5 from session 10 onward, common to all cases. We compute $\htheta^{\rm b}$, $\htheta^{\rm c}$ and $\htheta^{\rm f}$. Table \ref{tab:sim6} reports the bias, the standard deviation of the estimation error, the mean standard error, and the coverage of the intervals.

The results match Theorem \ref{thm:coverage}. The bias of $\htheta^{\rm f}$ decreases as the number of cases grows, and the coverage of its intervals approaches the nominal level. The bias of $\htheta^{\rm f}$ is 0.174 with 6 cases, 0.072 with 12 cases, 0.033 with 24 cases, 0.023 with 48 cases, and 0.015 with 96 cases. It is a quarter of the standard deviation with 6 cases and less than a tenth from 24 cases onward. The estimators $\htheta^{\rm b}$ and $\htheta^{\rm c}$ have biases within two Monte Carlo standard errors of 0. The intervals of $\htheta^{\rm f}$ cover the pooled effect in 92.7\% and 92.5\% of the replications with 6 and 12 cases, and in 93.2\% to 95.7\% with 24 cases or more. The intervals of $\htheta^{\rm b}$ and $\htheta^{\rm c}$ cover it in 93.2\% to 96.1\%.

\begin{table}[ht]
    \centering
    \caption{Study 6: bias with its Monte Carlo standard error (MCSE), standard deviation of the estimation error (SD), mean standard error (SE), and coverage of 95\% intervals (\%).}
    \label{tab:sim6}
    \small
    \begin{tabular}{clccccc}
        \toprule
        Cases & Method & Bias & MCSE & SD & Mean SE & Coverage \\
        \midrule
        6 & $\htheta^{\rm b}$ & 0.037 & 0.043 & 1.359 & 1.264 & 96.1 \\
        & $\htheta^{\rm c}$ & $-0.007$ & 0.022 & 0.710 & 0.697 & 93.3 \\
        & $\htheta^{\rm f}$ & 0.174 & 0.022 & 0.685 & 0.641 & 92.7 \\
        12 & $\htheta^{\rm b}$ & 0.019 & 0.030 & 0.956 & 0.924 & 94.5 \\
        & $\htheta^{\rm c}$ & 0.009 & 0.016 & 0.510 & 0.507 & 93.2 \\
        & $\htheta^{\rm f}$ & 0.072 & 0.016 & 0.498 & 0.477 & 92.5 \\
        24 & $\htheta^{\rm b}$ & 0.022 & 0.021 & 0.677 & 0.672 & 95.5 \\
        & $\htheta^{\rm c}$ & $-0.006$ & 0.011 & 0.343 & 0.363 & 95.6 \\
        & $\htheta^{\rm f}$ & 0.033 & 0.011 & 0.333 & 0.349 & 94.5 \\
        48 & $\htheta^{\rm b}$ & 0.013 & 0.015 & 0.476 & 0.479 & 95.2 \\
        & $\htheta^{\rm c}$ & 0.010 & 0.008 & 0.253 & 0.259 & 94.5 \\
        & $\htheta^{\rm f}$ & 0.023 & 0.008 & 0.254 & 0.250 & 93.2 \\
        96 & $\htheta^{\rm b}$ & 0.014 & 0.011 & 0.340 & 0.337 & 93.8 \\
        & $\htheta^{\rm c}$ & 0.010 & 0.006 & 0.175 & 0.184 & 96.0 \\
        & $\htheta^{\rm f}$ & 0.015 & 0.005 & 0.172 & 0.178 & 95.7 \\
        \bottomrule
    \end{tabular}
\end{table}

\subsection{Study 7: The Scale}
\label{sec:sim7}

This study examines Proposition \ref{prop:scale} and Corollary \ref{cor:standardized} for the pooled standard deviation with and without the factor $\iota_i$. There are six studies of four cases each. The outcomes follow model \eqref{eq:working} with independent normal noise, $\tau_i=2$ and $\sigma_i=1$ for every case. Every case has 3, 5, or 8 pre-window sessions, seven window sessions, and eight post-window sessions, and the onsets are drawn independently and uniformly from the eight admissible sessions. The scale is the standard deviation of the pre-window outcomes pooled across the four cases of the study, with $\nu_i=8$, 16, or 28 degrees of freedom, or this standard deviation divided by $\iota_i$. We compute $\htheta^{\rm c}$ and $\htheta^{\rm f}$, with each case as a block. Table \ref{tab:sim7} reports the mean pooled effect, its value given by Proposition \ref{prop:scale}, the bias, standard deviation of the estimation error, mean standard error, and coverage of the intervals.

The results match Proposition \ref{prop:scale} and Corollary \ref{cor:standardized}. With the scale $c_i$, Proposition \ref{prop:scale} gives a mean pooled effect of $2/\iota_i$, which is $2/0.9027=2.216$ for $\nu_i=8$, and the simulated mean is 2.218. With the scale $c_i/\iota_i$, it yields 2, and the simulated mean is 2.002. The biases of both estimators are within two Monte Carlo standard errors of 0 for both scales, and the intervals cover the pooled effect in 94.8\% to 97.3\% of the replications.

\begin{table}[ht]
    \centering
    \caption{Study 7: mean pooled effect and its value given by Proposition \ref{prop:scale}, bias with its Monte Carlo standard error (MCSE), standard deviation of the estimation error (SD), mean standard error (SE), and coverage of 95\% intervals (\%). The scale is the pooled standard deviation $c_i$ or $c_i/\iota_i$.}
    \label{tab:sim7}
    \small
    \setlength{\tabcolsep}{4pt}
    \begin{tabular}{cclcclccccc}
        \toprule
        Pre-window & & & Mean & Proposition & & & & & Mean & \\
        sessions & $\nu_i$ & Scale & $\theta$ & \ref{prop:scale} & Method & Bias & MCSE & SD & SE & Coverage \\
        \midrule
        3 & 8 & $c_i$ & 2.218 & 2.216 & $\htheta^{\rm c}$ & $-0.008$ & 0.008 & 0.252 & 0.281 & 97.3 \\
        & & & & & $\htheta^{\rm f}$ & $-0.002$ & 0.011 & 0.341 & 0.346 & 95.2 \\
        & & $c_i/\iota_i$ & 2.002 & 2.000 & $\htheta^{\rm c}$ & $-0.007$ & 0.007 & 0.228 & 0.254 & 97.3 \\
        & & & & & $\htheta^{\rm f}$ & $-0.002$ & 0.010 & 0.308 & 0.312 & 95.2 \\
        5 & 16 & $c_i$ & 2.096 & 2.100 & $\htheta^{\rm c}$ & $-0.005$ & 0.007 & 0.236 & 0.232 & 96.4 \\
        & & & & & $\htheta^{\rm f}$ & $-0.010$ & 0.009 & 0.280 & 0.264 & 94.8 \\
        & & $c_i/\iota_i$ & 1.996 & 2.000 & $\htheta^{\rm c}$ & $-0.004$ & 0.007 & 0.224 & 0.221 & 96.4 \\
        & & & & & $\htheta^{\rm f}$ & $-0.010$ & 0.008 & 0.267 & 0.251 & 94.8 \\
        8 & 28 & $c_i$ & 2.051 & 2.056 & $\htheta^{\rm c}$ & $-0.007$ & 0.007 & 0.207 & 0.222 & 96.9 \\
        & & & & & $\htheta^{\rm f}$ & 0.000 & 0.007 & 0.231 & 0.238 & 95.5 \\
        & & $c_i/\iota_i$ & 1.995 & 2.000 & $\htheta^{\rm c}$ & $-0.007$ & 0.006 & 0.201 & 0.216 & 96.9 \\
        & & & & & $\htheta^{\rm f}$ & 0.000 & 0.007 & 0.225 & 0.231 & 95.5 \\
        \bottomrule
    \end{tabular}
\end{table}

\clearpage
\section{Details of the Application}
\label{sec:application}

\subsection{Designs and Case-Level Estimates}

This analysis reports the design of each study and the estimates for each student on the original scale. Table \ref{tab:designs} gives the numbers of sessions and the admissible onsets. In Study B, the three maintenance sessions that followed the intervention phase are not used. In Study C, the outcome is the total number of words written. Table \ref{tab:designs} also gives the pooled standard deviation $c_i$ of \eqref{eq:pooledsd}, its degrees of freedom $\nu_i$, the factor $\iota_i$, and the scale $c_i/\iota_i$ that the analysis uses. Table \ref{tab:cases} reports, for each student, the onset, the pre-window mean and standard deviation, the post-window mean, the estimates $\htau^{\rm a}_i$, $\htau^{\rm b}_i$ and $\htau^{\rm c}_i$, and the phase contrast.

The case-level estimates vary in the way that \eqref{eq:planb} and Proposition \ref{prop:biasc} indicate. In Study A, the estimate $\htau^{\rm a}_i$ ranges from $-11.32$ to 12.52, while the phase contrasts range from 4.14 to 6.00, because its weighted terms include the series level. The estimate $\htau^{\rm b}_i$ is larger for students with early onsets, whose sum of treated weights $X_{i1}$ is large. The estimate $\htau^{\rm c}_i$ ranges from 3.75 to 5.70 in Study A. In Study C, the two students with onset at session 6 have estimates $\htau^{\rm c}_i$ of $-31.61$ and 19.15 words. Their outcomes rose throughout the intervention phase, so their post-window means (156.33 and 108.67 words) exceed their first treated outcomes, which carry the largest weights.

\begin{table}[ht]
    \centering
    \caption{Designs of the three studies, with the pooled standard deviation $c_i$ of the pre-window outcomes, its degrees of freedom $\nu_i$, the factor $\iota_i$ and the scale $c_i/\iota_i$.}
    \label{tab:designs}
    \small
    \setlength{\tabcolsep}{4pt}
    \begin{tabular}{ccccccccccc}
        \toprule
        & & & Admissible & Pre- & & Post- & & & & \\
        Study & Students & Sessions & onsets & window & Window & window & $c_i$ & $\nu_i$ & $\iota_i$ & $c_i/\iota_i$ \\
        \midrule
        A & 6 & 18 & 4 to 9 & 3 & 5 & 10 & 0.75 & 12 & 0.936 & 0.80 \\
        B & 5 & 15 & 4 to 11 & 3 & 7 & 5 & 1.95 & 10 & 0.923 & 2.11 \\
        C & 4 & 12 & 6 to 10 & 5 & 4 & 3 & 4.86 & 16 & 0.952 & 5.10 \\
        \bottomrule
    \end{tabular}
\end{table}

\begin{table}[ht]
    \centering
    \caption{Case-level quantities on the original scale: onset, pre-window mean and standard deviation (SD), post-window mean, estimates $\htau^{\rm a}_i$, $\htau^{\rm b}_i$ and $\htau^{\rm c}_i$, and phase contrast.}
    \label{tab:cases}
    \small
    \setlength{\tabcolsep}{4pt}
    \begin{tabular}{clcccccccc}
        \toprule
        & & & \multicolumn{2}{c}{Pre-window} & Post-window & & & & Phase \\
        \cmidrule(lr){4-5}
        Study & Student & Onset & Mean & SD & mean & $\htau^{\rm a}_i$ & $\htau^{\rm b}_i$ & $\htau^{\rm c}_i$ & contrast \\
        \midrule
        A & Anna & 5 & 3.33 & 0.58 & 9.20 & 12.08 & 7.75 & 4.58 & 6.00 \\
        A & Bella & 8 & 3.00 & 1.00 & 8.80 & $-2.62$ & 1.28 & 5.69 & 5.82 \\
        A & Christina & 5 & 5.00 & 1.00 & 9.20 & 12.52 & 6.02 & 3.75 & 4.14 \\
        A & Dunja & 7 & 4.33 & 0.58 & 8.90 & 1.44 & 3.17 & 5.27 & 5.00 \\
        A & Egor & 6 & 4.67 & 0.58 & 9.20 & 5.60 & 3.73 & 4.01 & 4.48 \\
        A & Fabian & 9 & 4.67 & 0.58 & 8.90 & $-11.32$ & 1.47 & 5.70 & 4.78 \\
        B & Asena & 5 & 1.33 & 1.53 & 8.40 & 10.45 & 8.24 & 2.45 & 6.11 \\
        B & Eman & 9 & 1.67 & 1.53 & 9.80 & 1.07 & 2.56 & 7.82 & 7.98 \\
        B & Julia & 5 & 2.33 & 2.52 & 8.20 & 10.49 & 6.63 & 1.81 & 5.27 \\
        B & Leon & 5 & 3.00 & 2.65 & 9.00 & 12.81 & 7.84 & 2.91 & 5.61 \\
        B & Marvin & 6 & 2.00 & 1.00 & 10.00 & 11.54 & 9.75 & 7.76 & 8.20 \\
        C & Adal & 6 & 47.80 & 5.31 & 156.33 & 266.98 & 142.50 & $-31.61$ & 83.91 \\
        C & Berta & 8 & 37.60 & 4.72 & 180.00 & 96.04 & 96.04 & 134.61 & 138.74 \\
        C & Channa & 6 & 40.80 & 4.87 & 108.67 & 234.27 & 128.02 & 19.15 & 59.91 \\
        C & Daniel & 8 & 33.40 & 4.51 & 87.67 & 42.29 & 42.29 & 56.99 & 56.03 \\
        \bottomrule
    \end{tabular}
\end{table}

\subsection{Weights and Randomization Tests}

This analysis repeats the pooled analysis of the main text with equal weights for the three studies and gives the details of the randomization tests. With equal study weights, a student has weight $1/18$ in Study A, $1/15$ in Study B, and $1/12$ in Study C, and each student is a block. The randomization tests use the statistic $\cT(\bm s)$ equal to the pooled-anchor estimate with the scale $c_i/\iota_i$ and equal weights, and the $p$-values $P_{\rm ex}$ and $P_{\rm mc}$ of \eqref{eq:pvalues}. Within a study, the exact $p$-value sums over all onset vectors, of which there are $6^6=46{,}656$ in Study A, $8^5=32{,}768$ in Study B, and $5^4=625$ in Study C. For the 15 students together, there are about $9.6\times10^{11}$ vectors, and we use the Monte Carlo $p$-value with $B=100{,}000$ draws. Table \ref{tab:weights} reports the pooled estimates, their standard errors, and their 95\% intervals for both sets of weights and both scales.

The pooled-anchor estimates change little with the weights, and the $p$-values are valid by Proposition \ref{prop:test}. The pooled-anchor estimates with equal study weights are 7.47 standard deviations and 1.86 pre-window means, and those with equal student weights are 6.98 and 1.86. The exact $p$-values are $3/46{,}656=0.00006$ in Study A, $55/32{,}768=0.0017$ in Study B and $8/625=0.0128$ in Study C, and the Monte Carlo $p$-value for the 15 students is $20/100{,}001=0.0002$.

\begin{table}[ht]
    \centering
    \caption{Pooled standardized estimates for the 15 students: estimate, standard error (SE), and 95\% interval, for two scales and two sets of weights.}
    \label{tab:weights}
    \small
    \begin{tabular}{llccc}
        \toprule
        Scale and weights & Method & Estimate & SE & 95\% interval \\
        \midrule
        Standard deviation, equal students & $\htheta^{\rm b}$ & 8.41 & 2.21 & $[3.66, 13.15]$ \\
        & $\htheta^{\rm c}$ & 5.49 & 1.80 & $[1.63, 9.34]$ \\
        & $\htheta^{\rm f}$ & 6.98 & 1.78 & $[3.13, 10.84]$ \\
        & Phase contrast & 8.00 & 1.69 & $[4.36, 11.63]$ \\
        Standard deviation, equal studies & $\htheta^{\rm b}$ & 9.41 & 2.59 & $[3.86, 14.97]$ \\
        & $\htheta^{\rm c}$ & 5.67 & 2.24 & $[0.87, 10.47]$ \\
        & $\htheta^{\rm f}$ & 7.47 & 2.05 & $[3.05, 11.89]$ \\
        & Phase contrast & 8.68 & 2.02 & $[4.35, 13.01]$ \\
        Mean, equal students & $\htheta^{\rm b}$ & 2.25 & 0.43 & $[1.34, 3.17]$ \\
        & $\htheta^{\rm c}$ & 1.64 & 0.36 & $[0.86, 2.42]$ \\
        & $\htheta^{\rm f}$ & 1.86 & 0.37 & $[1.06, 2.66]$ \\
        & Phase contrast & 2.26 & 0.35 & $[1.51, 3.01]$ \\
        Mean, equal studies & $\htheta^{\rm b}$ & 2.35 & 0.41 & $[1.48, 3.23]$ \\
        & $\htheta^{\rm c}$ & 1.64 & 0.40 & $[0.78, 2.50]$ \\
        & $\htheta^{\rm f}$ & 1.86 & 0.39 & $[1.02, 2.70]$ \\
        & Phase contrast & 2.32 & 0.35 & $[1.58, 3.06]$ \\
        \bottomrule
    \end{tabular}
\end{table}

\subsection{Study C With Two Pairs as Blocks}
\label{sec:pairs}

This analysis repeats the pooled analysis and the randomization tests under a second randomization scheme for Study C. In Study C, two students started at session 6 and the other two at session 8. Suppose that the onset was drawn for each of these two pairs, independently and uniformly from sessions 6 to 10, and that both students of a pair received the onset of their pair. Each student is then treated at session $t$ with the same probability $\pi_{it}$ as under independent draws, so the case-level estimates and the pooled estimates are those of Tables \ref{tab:cases} and \ref{tab:weights}. The blocks are the 11 students in Studies A and B and the two pairs, so $G=13$. Study C alone has two blocks, and Definition \ref{def:estimators}(e) requires at least three. Therefore, we do not repeat the analysis by study for it. The randomization distribution of Study C has $5^2=25$ vectors of onsets. Table \ref{tab:pairs} reports the pooled estimates, their standard errors, and their 95\% intervals with 13 blocks and equal student weights.

The estimates equal those with 15 blocks, as the unchanged probabilities $\pi_{it}$ imply, and the standard errors of $\htheta^{\rm c}$ and $\htheta^{\rm f}$ are larger. The standard errors of the pooled-anchor estimate are 0.41 pre-window means and 2.24 standard deviations, compared to 0.37 and 1.78 with 15 blocks, and both intervals exclude 0. The exact $p$-value of Study C is $2/25=0.08$, and the smallest value that 25 vectors of onsets allow is 0.04; whereas the Monte Carlo $p$-value for the 15 students is $76/100{,}001=0.0008$.

\begin{table}[ht]
    \centering
    \caption{Pooled standardized estimates for the 15 students with the two pairs of Study C as blocks ($G=13$): estimate, standard error (SE), and 95\% interval, with equal student weights.}
    \label{tab:pairs}
    \small
    \begin{tabular}{llccc}
        \toprule
        Scale & Method & Estimate & SE & 95\% interval \\
        \midrule
        Mean & $\htheta^{\rm b}$ & 2.25 & 0.43 & $[1.32, 3.19]$ \\
        & $\htheta^{\rm c}$ & 1.64 & 0.41 & $[0.74, 2.53]$ \\
        & $\htheta^{\rm f}$ & 1.86 & 0.41 & $[0.95, 2.77]$ \\
        & Phase contrast & 2.26 & 0.34 & $[1.52, 3.00]$ \\
        Standard deviation & $\htheta^{\rm b}$ & 8.41 & 3.00 & $[1.88, 14.93]$ \\
        & $\htheta^{\rm c}$ & 5.49 & 2.28 & $[0.52, 10.45]$ \\
        & $\htheta^{\rm f}$ & 6.98 & 2.24 & $[2.05, 11.92]$ \\
        & Phase contrast & 8.00 & 2.04 & $[3.55, 12.44]$ \\
        \bottomrule
    \end{tabular}
\end{table}

\end{document}